\documentclass{LMCS}

\def\dOi{10(4:17)2014}
\lmcsheading%
{\dOi}
{1--45}
{}
{}
{Apr.~\phantom05, 2013}
{Dec.~26, 2014}
{}

\ACMCCS{[{\bf Theory of computation}]: Formal languages and automata
  theory; Logic---Modal and temporal logics; [{\bf Software and its
      engineering}]: Software organization and properties---Software
  functional properties---Formal methods---Software verification;
  Model checking}

\pdfoutput=1
\let\savecc\cc
\let\cc\relax
\usepackage{thm-restate}
\usepackage{nico}
\usepackage{qctl}
\usepackage{hyperref}

\let\cc\savecc

\usepackage{multirow}

\begin{document}

\title[Quantified CTL:  expressiveness and complexity]{Quantified CTL:
  expressiveness and complexity\rsuper*}

\titlecomment{{\lsuper*}This is a long version of paper~\cite{DLM12}, which appeared in CONCUR'12.}

\author[F.~Laroussinie]{Fran\c cois~Laroussinie\rsuper a}
\address{{\lsuper a}LIAFA -- Universit\'e Paris Diderot \& CNRS}
\email{francoisl@liafa.univ-paris-diderot.fr}

\author[N.~Markey]{Nicolas Markey\rsuper b}
\address{{\lsuper b}LSV -- ENS Cachan \& CNRS}
\email{markey@lsv.ens-cachan.fr}
\thanks{{\lsuper b}This work benefited from the support of the ERC Starting Grant
  EQualIS and of the EU-FP7 project Cassting.}


\begin{abstract}
  While it was defined long ago, the extension of \CTL with quantification
  over atomic propositions has never been studied extensively. Considering two
  different semantics (depending whether propositional quantification refers
  to the Kripke structure or to its unwinding tree), we~study its
  expressiveness (showing in particular that \QCTL coincides with Monadic
  Second-Order Logic for both semantics) and characterise the complexity of
  its model-checking and satisfiability problems, depending on the number of nested propositional
  quantifiers (showing that the structure semantics populates the polynomial
  hierarchy while the tree semantics populates the exponential hierarchy).
\end{abstract}

\maketitle

\section{Introduction}

\paragraph*{Temporal logics.}
Temporal logics extend propositional logics with modalities for specifying
constraints on the order of events in time. Since~\cite{Pnu77,CE82,QS82a}, they
have received much attention from the computer-aided-verification community,
since they fit particularly well for expressing and automatically verifying
(\emph{model checking}) properties of reactive systems.

Two important families of temporal logics have been considered: linear-time
temporal logics (\eg~\LTL~\cite{Pnu77}) can be used to express properties of
one single execution of the system under study, while branching-time temporal
logics (\eg~\CTL~\cite{CE82,QS82a} and \CTL*~\cite{EH86}) consider the
execution tree. Since the 90s, many extensions of these logics have been 
introduced, of which alternating-time temporal logics
(such~as \ATL, \ATL*~\cite{AHK97}) extend \CTL towards the study of
open systems (involving several agents).

In this landscape of temporal logics, both \CTL and \ATL enjoy the nice
property of having polynomial-time model-checking algorithms. In~return
for this, both logics have quite limited expressiveness. 
Several extensions have been defined in order to increase this limited expressive power.

\paragraph*{Our contributions.}
We~are interested in the present paper in the extension of \CTL (and~\CTL*)
with \emph{propositional quantification}~\cite{Sis83,ES84a}. 
In~that setting, propositional quantification can take different meaning,
depending whether the extra propositions label the Kripke structure under
study (\emph{structure semantics}) or its execution tree (\emph{tree semantics}).
While these extensions of \CTL with propositional quantification have been in
the air for thirty years, they have not been extensively studied yet:
some complexity results have been published for existential
quantification~\cite{Kup95a}, for the two-alternation fragment~\cite{KMTV00}
and for the full extension~\cite{Fre01}; but expressiveness issues, as well as
a complete study of model checking and satisfiability for the whole hierarchy, have been mostly
overlooked. 

We~answer these questions in the present paper: in~terms of
expressiveness, we prove that \QCTL and \QCTL* are equally expressive, and coincide with Monadic
Second-Order Logic\footnote{This claim assumes a special notion of equivalence
  between formulas, since \MSO is evaluated \emph{globally} on a structure
  while \QCTL formulas are evaluated at the initial state. This will be made
  clear in the paper.}. 
As~regards satisfiability and model-checking, we~characterise the complexity
of these problems depending on the quantifier alternation: under the structure
semantics, the model-checking problem populates the polynomial-time hierarchy
(and satisfiability is undecidable); for the tree semantics, the
model-checking problem  populates the
exponential-time hierarchy (and so does the satisfiability problem).
Finally, we~also characterise the model- and
formula-complexities of our problems, when one of the inputs to the
model-checking problem is fixed. All these results are summarized in
Tables~\ref{tab:results-structure}
and~\ref{tab:results-tree}, which are displayed in the conclusion of this paper.

\paragraph*{Applications to alternating-time temporal logics.}
Our initial motivation for this work comes from  alternating-time temporal logics.
Indeed \ATL also has several flaws in terms of expressiveness: namely, it~can
only focus on (some) zero-sum properties, \ie, on purely antagonist
games, in which two coalitions fight with opposite objectives. In~many
situations,
games are not purely antagonist, but involve several independent
systems, each having its own objective. 
Recently, several extensions of \ATL have
been defined to express properties of such non-zero-sum games. Among those, our logic
\ATLsc~\cite{DLM10} 
extends \ATL with \emph{strategy contexts}, which provides a way of
expressing interactions between strategies. Other similar approaches
include Strategy Logics~(\SL)~\cite{CHP07b,MMV10a}, 
(Basic) Strategy-Interaction Logic (\textsf{(B)SIL})~\cite{WHY11}, or
Temporal Cooperation Logic~(\textsf{TCL})~\cite{HSW13}.

\smallskip
Designing decision procedures for these
extensions is much more difficult than for the standard \ATL fragment. 
Interestingly, \QCTL appears to be a convenient, uniform intermediary
logic in order to obtain algorithms for \ATLsc, \SL and related
formalisms.  Indeed, strategies of the players can be
represented\footnote{Notice that the link between strategy quantification and
  propositional quantification already emerges in \Qdmu~\cite{Pin07a},
  which extends the \emph{decision \mucalculus} with some flavour of
  propositional quantification. Also, the main motivation
  of~\cite{KMTV00} for studying the two-alternation fragment of \QCTL
  is a hardness result for the control and synthesis of open systems.}
by a finite set of atomic propositions labelling the execution tree of
the game structure under study. Strategy quantification is then
naturally expressed as propositional quantification; since the
resulting labelling is \emph{persistent}, it~can encode interactions
between strategies.  Notice that while the \emph{tree semantics} of
\QCTL encodes plain strategies, the \emph{structure semantics} also
finds a meaning in that translation, as it may correspond to
\emph{memoryless strategies}~\cite{DLM12}.

Using such a translation, any instance of the model-checking problem
for \ATLsc (or~\SL) can be translated into an instance of the
model-checking problem for \QCTL. The algorithms proposed in this
paper then yield algorithms for the former problems, which can be
proved to have optimal complexity. Unfortunately, the satisfiability
problem cannot follow the same reduction scheme: indeed, when
translating an \ATLsc formula into a \QCTL~one, we~need to know the
set of agents and their allowed moves. It~turns out that
satisfiability is undecidable for \ATLsc and \SL (while we prove~it
decidable for \QCTL in the tree semantics). Interestingly, when
restricting satisfiability checking to \emph{turn-based} game
structures, an alternative translation into \QCTL can be used to
obtain decidability of the problem.

Because they involve a lot of new definitions and technical proofs, we
do not develop these questions here, and refer the interested reader
to~\cite{rr-ATLsc} for full details.

\paragraph*{Related works.}

Extending modal logics with quantification dates back to early works
of Kripke~\cite{Kri59}) and Fine~\cite{Fin70}.  We~refer
to~\cite{FM98,AP06,tC06} for more details.

(Propositional) quantification was first used in temporal logics by
Sistla and others, both for linear-time~\cite{Sis83,SVW87} and
branching-time temporal logics~\cite{ES84a}, mainly with the aim of
augmenting the expressiveness of the classical logics. In the
linear-time setting, the model-checking problem for the
$k$-alternation fragment was
shown \EXPSPACE[k]-complete~\cite{Sis83,SVW87}. 
The stutter-invariant fragment of \QLTL, with a restricted notion of
propositional quantification, was developed in~\cite{Ete99}.
Proof systems for \QLTL were developed in the linear-time setting, both with and without
past-time modalities~\cite{KP02a,FR03}.

As regards branching time, the extension of \CTL* with external
existential quantification (hereafter called \EQnCTL1*) was proved as
expressive as parity tree automata over binary
tree~\cite{ES84a}.  The existential logics \EQnCTL1 and \EQnCTL1* were
further studied in~\cite{Kup95a}, both in the structure- and in the
tree semantics; model checking \EQnCTL1 and \EQnCTL1* are shown \NP-
and \PSPACE-complete respectively (for the structure semantics) and
\EXPTIME- and \EXPTIME[2]-complete respectively (for the tree semantics).  The
extensions of those logics with past-time modalities were studied
in~\cite{KP95b}. The extensions with arbitrary quantification were
studied in~\cite{Kai97,Fre01} (in slightly different settings):
satisfiability of \QCTL* was proven undecidable in the structure semantics,
and decidable in the tree semantics~\cite{Fre01}. 

Several alternative semantics were proposed for quantification: the
\emph{amorphous} semantics defined in~\cite{Fre01} allows to take a
bisimilar structure before labelling~it. In~\cite{RP03a},
quantification is expressed as taking a synchronized product with a
labelling automaton. Finally, quantification over states (rather than
over atomic propositions) is studied in~\cite{PBDDC02,CDC04}, where
model checking is proved \PSPACE-complete (both for branching-time and
for linear-time).

\smallskip

Finally, quantified temporal logics have found applications in model
checking and control: \AQnLTL1 and \AQnCTL1* have been used to reason about
\emph{vacuity detection} (checking whether a formula is satisfied
``\emph{too~easily}'')~\cite{AFFGPTV03,GC04,GC12}. The one-alternation
fragments (which we call \EQnCTL2 and \EQnCTL2* hereafter) have been used
in~\cite{KMTV00} to prove hardness results for the control problem with \CTL
and \CTL* objectives. The linear-time logic \EQLTL~was used as the specification language for
supervisory control of Petri nets in~\cite{Mar10}. To~conclude, propositional
quantification was considered in the setting of \emph{timed} temporal logics
in~\cite{HRS98}.

\section{Preliminaries}

\subsection{Kripke structures and trees}
We fix once and for all a set~$\AP$ of atomic propositions.

\begin{definition}
  A \emph{Kripke structure}~$\calS$ is a $3$-tuple $\tuple{Q,R,\ell}$
  where $Q$~is a countable set of states, $R \subseteq Q^2$ is a
  total\footnote{\emph{I.e.}, for all~$q\in Q$, there exists~$q'\in Q$
    s.t. $(q,q')\in R$.}  relation and $\ell\colon Q \rightarrow
  2^\AP$ is a labelling function.
  The size of~$\calS$, denoted with~$\size\calS$, is the size of~$Q$ (which
  can be infinite).
\end{definition}

Let $\calS$ be a Kripke structure $\tuple{Q,R,\ell}$.  In the
following, we always assume that the set of states~$\calS$ is equipped with a
total linear order~$\preceq$.
We~use $\Succ_\calS(q)$ to denote the ordered list
$\tuple{q'_0,\ldots,q'_k}$ of successors of~$q$ in~$\calS$ (\ie, such that
$(q,q'_i)\in R$ for any ${0\leq i \leq k}$, and such that $q'_i\preceq q'_j$
if, and only~if, $i\leq j$). We~write $\deg_\calS(q)$ for the
degree of~$q\in Q$, \ie, the size of~$\Succ_\calS(q)$. Finally 
$\Succ_\calS(q,i)$ denotes the $i$-th successor of~$q$ in~$\calS$
for ${0\leq i < \deg_\calS(q)}$, and this notation is extended to finite words
over~$\Nat^*$ as follows: $\Succ_\calS(q,\varepsilon) = q$ and 
$\Succ_\calS(q,w\cdot i) = \Succ_\calS(\Succ_\calS(q,w),i)$ when $q'=
\Succ_\calS(q,w)$ is well defined and ${0\leq i < \deg_\calS(q')}$.

An execution (or path) in~$\calS$ is an infinite sequence $\rho =
(q_i)_{i\in\bbN}$ s.t. ${(q_i,q_{i+1})\in R}$ for all~$i\in\bbN$. We~use
$\Exec(q)$ to denote the set of executions issued from~$q$ and
$\Execf(q)$ for the set of all \emph{finite} prefixes of executions
of~$\Exec(q)$. 
%
Given~$\rho\in\Exec(q)$ and $i\in\bbN$, we~write $\rho^i$ for the
path~$(q_{i+k})_{k\in\bbN}$ of~$\Exec(q_i)$ (the~$i$-th suffix
of~$\rho$), $\rho_i$~for the finite prefix~$(q_k)_{k\leq i}$
(the~$i$-th prefix), and $\rho(i)$~for the $i$-th state~$q_i$.
Given a path~$\rho=(q_i)_{i\in\bbN}$, we~write $\ell(\rho)$ for the sequence
$(\ell(q_i))_{i\in\bbN}$, and $\Inf(l(\rho))$ for the set of letters
in~$\Sigma$ that appear infinitely many times along~$\ell(\rho)$.

\begin{definition}
Let~$\Sigma$ be a finite set.
A \emph{$\Sigma$-labelled tree} is
a pair~$\calT = \tuple{T, l}$, where
\begin{itemize}
\item $T\subseteq \Nat^*$ is a non-empty set of finite words on~$\Nat$
  satisfying the following constraints: for~any non-empty
  word~$x=y\cdot c$ in~$T$ with~$y\in \Nat^*$ and~$c\in \Nat$, the
  word~$y$ is in~$T$ and every word $y\cdot c'$ with $0\leq c' < c$ is
  also in $T$;
\item $l\colon T \to \Sigma$ is a labelling function.
\end{itemize}
\end{definition}

Let $\calT=\tuple{T,l}$ be a $\Sigma$-labelled tree. The elements of~$T$ are
the \emph{nodes} of $\calT$ and the empty word~$\varepsilon$ is the root
of~$\calT$.
Such a tree can be seen as a Kripke structure, with~$T$ as set
of states, and transitions from any node~$x\in T$ to any node of the form~$x\cdot
c\in T$, for~$c\in\bbN$. The size of~$\calT$, and the notions 
of \emph{successors} of a node~$x$ (written~$\Succ_\calT(x)$), 
of~\emph{degree} of a node~$x$ (written $\deg_\calT(x)$), 
of~\emph{path} issued from the root (whose set is written $\Exec_\calT$),
follow from this correspondence.

A~tree has \emph{bounded branching} if the degree of all its nodes is bounded.
Given a finite set of integers $\calD \subseteq \Nat$,
a~$\tuple{\Sigma,\calD}$-tree is a $\Sigma$-labelled tree $\tuple{T,l}$ whose
nodes have their degrees in~$\calD$ (\ie, for any $x\in T$, it~holds
$\deg_\calT(x) \in \calD$).
Given a node $x\in T$, we~denote with~$\calT_x$ the (sub)tree
$\tuple{T_x,l_x}$ rooted at~$x$, defined by $T_x= \{y \in T \mid
  \exists z \in T \mbox{ s.t. } z = x\cdot y\}$.

\begin{definition}
Given a finite-state Kripke structure $\calS = \tuple{Q,R,\ell}$ and
a state ${q\in Q}$, the \newdef{unwinding of~$\calS$ from~$q$} is the
(bounded-degree) $2^\AP$-labelled tree $\calT_\calS(q) =
\tuple{T_{\calS,q},\ell_\calT}$ defined as follows: 
\begin{enumerate}
\item $T_{\calS,q}$
contains exactly all nodes $x \in \Nat^*$ such that $\Succ_\calS(q,x)$ is well-defined 
\item  $\ell_\calT(x) = \ell(\Succ_\calS(q,x))$. 
\end{enumerate}
\end{definition}
If $\calD = \bigcup_{q\in Q} \{\deg_\calS(q)\}$, then $\calT_\calS(q)$ clearly
is a $\tuple{2^\AP,\calD}$-tree. Note also that any $2^\AP$-labelled tree can
be seen as an infinite-state Kripke structure.

\medskip
For a 
function $\ell\colon Q\to 2^\AP$ and $P\subseteq \AP$,
we write $\ell\cap P$ for the 
function defined as
$(\ell\cap P)(q)=\ell(q)\cap P$ for all $q\in Q$.

\begin{definition}
  For $P\subseteq \AP$, two (possibly infinite-state) Kripke
  structures $\calS=\tuple{Q,R,\ell}$ and $\calS'=\tuple{Q',R',\ell'}$
  are \emph{$P$-equivalent} (denoted by $\calS \equiv_P \calS'$) if
  $Q=Q'$, $R=R'$ and $\ell\cap P = \ell'\cap P$.
\end{definition}

\begin{figure}[!ht]
\centering
\begin{tikzpicture}[inner sep=0pt]
\begin{scope}
\path (0,0) node[above left=4mm] {$\calS_0$};
\draw (0,0) node[rond] (A) {$r$} node[below left=2.5mm] {$\scriptstyle q_0$};
\draw (0,-2) node[rond] (B) {} node[below left=2.5mm] {$\scriptstyle q_1$};
\draw[-latex'] (A) .. controls +(-115:8mm) and +(115:8mm) .. (B);
\draw[-latex'] (B) .. controls +(65:8mm) and +(-65:8mm) .. (A);
\draw[-latex'] (A) .. controls +(-30:10mm) and +(30:10mm) .. (A);
\draw[-latex'] (B) .. controls +(-30:10mm) and +(30:10mm) .. (B);
\end{scope}
\begin{scope}[xshift=3cm]
\path (0,0) node[above left=4mm] {$\calS_1$};
\draw (0,0) node[rond] (A) {$p,r$} node[below left=2.5mm] {$\scriptstyle q_0$};
\draw (0,-2) node[rond] (B) {} node[below left=2.5mm] {$\scriptstyle q_1$};
\draw[-latex'] (A) .. controls +(-115:8mm) and +(115:8mm) .. (B);
\draw[-latex'] (B) .. controls +(65:8mm) and +(-65:8mm) .. (A);
\draw[-latex'] (A) .. controls +(-30:10mm) and +(30:10mm) .. (A);
\draw[-latex'] (B) .. controls +(-30:10mm) and +(30:10mm) .. (B);
\end{scope}
\begin{scope}[xshift=6cm]
\path (0,0) node[above left=4mm] {$\calS_2$};
\draw (0,0) node[rond] (A) {$r$} node[below left=2.5mm] {$\scriptstyle q_0$};
\draw (0,-2) node[rond] (B) {$p$} node[below left=2.5mm] {$\scriptstyle q_1$};
\draw[-latex'] (A) .. controls +(-115:8mm) and +(115:8mm) .. (B);
\draw[-latex'] (B) .. controls +(65:8mm) and +(-65:8mm) .. (A);
\draw[-latex'] (A) .. controls +(-30:10mm) and +(30:10mm) .. (A);
\draw[-latex'] (B) .. controls +(-30:10mm) and +(30:10mm) .. (B);
\end{scope}
\begin{scope}[xshift=9cm]
\path (0,0) node[above left=4mm] {$\calS_3$};
\draw (0,0) node[rond] (A) {$p,r$} node[below left=2.5mm] {$\scriptstyle q_0$};
\draw (0,-2) node[rond] (B) {$p$} node[below left=2.5mm] {$\scriptstyle q_1$};
\draw[-latex'] (A) .. controls +(-115:8mm) and +(115:8mm) .. (B);
\draw[-latex'] (B) .. controls +(65:8mm) and +(-65:8mm) .. (A);
\draw[-latex'] (A) .. controls +(-30:10mm) and +(30:10mm) .. (A);
\draw[-latex'] (B) .. controls +(-30:10mm) and +(30:10mm) .. (B);
\end{scope}
\end{tikzpicture}
\caption{Four $\{r\}$-equivalent Kripke structures}\label{fig-peq}
\end{figure}
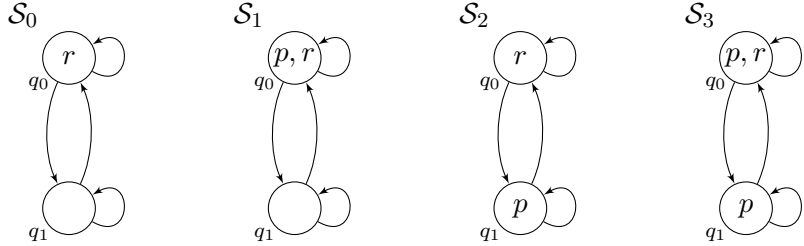

\subsection{\CTL and quantified extensions}

\begin{definition}
\label{def-QCTLs}
The syntax of \QCTL* is defined by the following grammar:
\begin{xalignat*}1
\phis,\psis \coloncolonequals\, & p \mid \non\phis \mid \phis\ou\psis
 \mid  \Ex \phip   \mid  \exists p.\;\phis   \\
 \phip,\psip \coloncolonequals\, &  \phis \mid \non\phip \mid \phip \ou \psip 
\mid \X\phip\mid  \phip\Until\psip 
\end{xalignat*}
where $p$ ranges over~$\AP$. Formulas defined as~$\phis$ are called
\emph{state-formulas}, while $\phip$ defines \emph{path-formulas}. Only state
formulas are \QCTL* formulas.
\end{definition}

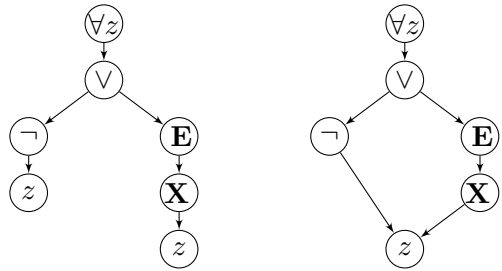
\begin{figure}[!ht]
\centering
\begin{tikzpicture}
\begin{scope}[yscale=.75]
\draw (0,0) node[draw,circle,minimum height=5mm] (All) {} node {$\forall z$};
\draw (0,-1) node[draw,circle,minimum height=5mm] (ou) {} node {$\ou$};
\draw (-1,-2) node[draw,circle,minimum height=5mm] (non) {} node {$\non$};
\draw (-1,-3) node[draw,circle,minimum height=5mm] (z1) {} node {$z$};
\draw (1,-2) node[draw,circle,minimum height=5mm] (Ex) {} node {$\Ex$};
\draw (1,-3) node[draw,circle,minimum height=5mm] (X) {} node {$\X$};
\draw (1,-4) node[draw,circle,minimum height=5mm] (z2) {} node {$z$};
\path[-latex'] (All) edge (ou) (ou) edge (non) (non) edge (z1);
\path[-latex'] (ou) edge (Ex) (Ex) edge (X) (X) edge (z2);
\end{scope}
\begin{scope}[yscale=.75,xshift=4cm]
\draw (0,0) node[draw,circle,minimum height=5mm] (All) {} node {$\forall z$};
\draw (0,-1) node[draw,circle,minimum height=5mm] (ou) {} node {$\ou$};
\draw (-1,-2) node[draw,circle,minimum height=5mm] (non) {} node {$\non$};
\draw (0,-4) node[draw,circle,minimum height=5mm] (z) {} node {$z$};
\draw (1,-2) node[draw,circle,minimum height=5mm] (Ex) {} node {$\Ex$};
\draw (1,-3) node[draw,circle,minimum height=5mm] (X) {} node {$\X$};
\path[-latex'] (All) edge (ou) (ou) edge (non) (non) edge (z);
\path[-latex'] (ou) edge (Ex) (Ex) edge (X) (X) edge (z);
\end{scope}
\end{tikzpicture}
\caption{Two representations of $\forall z. (\non z \ou \Ex \X z)$}
\label{fig-sizephi}
\end{figure}
Naturally, any \QCTL* formula~$\phi$ can be represented as a finite
tree~$\calT_\phi$, in which each node represents a subformula 
(see Fig.~\ref{fig-sizephi}).
Alternatively, formula~$\phi$ can be seen as a finite acyclic
Kripke structure~$\calS_\phi$
whose unwinding
is~$\calT_\phi$. The \emph{size} of~$\phi$ is the size
of~$\calT_\phi$, and its \emph{DAG-size} (for \emph{directed-acyclic-graph
  size}) is the size of the smallest Kripke structure~$\calS_\phi$ whose
unwinding is~$\calT_\phi$. Obviously, when sharing large subformulas, the
size of a formula can be significantly larger than its DAG-size.

We use standard abbreviations as: $\top \eqdef p\ou\non p$, $\bot
\eqdef \non\top$,  
$\F\phi \eqdef \top \Until \phi$, $\G\phi \eqdef
\non\F\non\phi$, 
$\All\phi \eqdef \non\Ex\non\phi$, and
$\forall p.\; \phi \eqdef \non\exists p.\; \non\phi$. 
%
The logic \QCTL is a fragment of \QCTL* where temporal modalities are
under the immediate scope of path quantifiers:
\begin{definition}
\label{def-QCTL}
The syntax of \QCTL is defined by the following grammar:
\begin{multline*}
\phis,\psis  \coloncolonequals  p \mid \non\phis \mid \phis\ou\psis  \mid  \exists p.\;\phis
    \mid \\ \Ex \phis \Until \psis  \mid \All \phis
  \Until \psis \mid  \Ex \X \phis  \mid  \All \X \phis.
\end{multline*}
\end{definition}

Standard definition of \CTL* and \CTL are obtained by removing the use
of quantification over atomic proposition ($\exists p. \phi$) in the
formulas. In the following, $\exists$~and~$\forall$ are called
\emph{(proposition) quantifiers}, while $\Ex$ and~$\All$ are
\emph{path quantifiers}.

Given \QCTL* (state) formulas~$\phi$ and~$(\psi_i)_i$ and atomic propositions~$(p_i)_i$
appearing free in~$\phi$ (\ie, not appearing as quantified propositions),
we~write $\phi[(p_i\to \psi_i)_i]$ (or~$\phi[(\psi_i)_i]$ when $(p_i)_i$ are 
understood from the context) 
for the formula obtained from~$\phi$ by replacing
each occurrence of~$p_i$ with~$\psi_i$. Given two sublogics~$L_1$ and~$L_2$
of~\QCTL*, we~write $L_1[L_2]=\{\phi[(\psi_i)_i] \mid \phi\in L_1,\ (\psi_i)_i\in L_2\}$.

\subsection{Structure- and tree semantics}
Formulas of the form $\exists p. \phi$ can be interpreted in different
manners (see~\cite{Kup95a,Fre01,RP03a}). Here we consider two
semantics: the~\emph{structure semantics} and the \emph{tree
  semantics}. 

\subsubsection{Structure semantics.}
Given a \QCTL* state formula~$\phi$, a~(possibly infinite-state)
Kripke structure~$\calS=\tuple{Q,R,\ell}$ and a state~$q\in Q$, we
write $\calS, q \models_s \phi$ to denote that $\phi$ holds
at~$q$ under the structure semantics.  
We~extend the notation
to~$\calS, \rho\models_s \phi$ when $\phi$ is a path formula
and~$\rho$ is a path in~$\calS$.
This~is defined as follows: 

\begin{xalignat*}1
\calS, q \sat_{s} p & \ \text{iff}\   p\in\ell(q) \\
\calS, q \sat_{s} \non\phis & \ \text{iff}\   \calS,  q \not\sat_{s} \phis \\
\calS, q \sat_{s} \phis \ou \psis & \ \text{iff}\   \calS,  q
\sat_{s} \phis \text{ or }  \calS,  q \sat_{s} \psis \\
\noalign{\pagebreak[2]}
\calS, q \sat_{s} \Ex \phip & \ \text{iff}\  \exists \rho\in
\Exec(q)  \text{ s.t. } \calS,  \rho  \sat_{s} \phip  \\
\noalign{\pagebreak[2]}
\calS, q \sat_{s}  \exists p. \phis & \ \text{iff}\  \exists
\calS' \equiv_{\AP\backslash\{p\}} \calS \:\text{s.t.}\: \calS',q  \sat_{s} \phis  \\
\noalign{\pagebreak[2]}
\calS, \rho \sat_{s} \phis  & \ \text{iff}\    \calS, \rho(0) 
\sat_{s} \phis \\
\calS, \rho \sat_{s} \non\phip & \ \text{iff}\   \calS,  \rho \not\sat_{s} \phip \\
\calS, \rho \sat_{s} \phip \ou \psip & \ \text{iff}\   \calS,  \rho
\sat_{s} \phip \:\text{or} \: \calS,  \rho \sat_{s} \psip \\
\noalign{\pagebreak[2]}
\calS, \rho \sat_{s} \X \phip  & \ \text{iff}\    \calS, \rho^1
\sat_{s} \phip \\
\calS, \rho \sat_{s} \phip \Until \psip & \ \text{iff}\  
 \exists i\geq 0.\
 \calS, \rho^i \sat_{s} \psip \ \text{and} \: \forall 0\leq j < i.\ \calS,  \rho^j \sat_{s}\phip
\end{xalignat*}

\begin{example}
As an example, consider the formula
$
  \selfloop \eqdef  \forall z.(z \impl \EX z)
$.
If~a~state~$q$ in~$\calS$ satisfies this formula, then the particular
labelling in which only~$q$ is labelled with~$z$ implies that~$q$ has to carry
a self-loop. Conversely, any state that carries a self-loop satisfies this
formula (for the structure semantics).

Let~$\phi$ be a \QCTL* formula, and consider now the formula
\[
  \uniq(\phi) \eqdef  \EF (\phi)  \et  \forall z.\Bigl(\EF(\phi \et z) \impl \AG(\phi \impl z)\Bigr).
\]
In~order to satisfy such a formula, at least one $\phi$-state must be
reachable. Assume now that two different such states~$q$ and~$q'$ are
reachable: then for the particular labelling where only~$q$ is labelled
with~$z$, the second part of the formula fails to hold. Hence $\uniq(\phi)$
holds in a state (under the structure semantics) if, and only~if, exactly one
reachable state satisfies~$\phi$.
Similarly, we can count the number of successors satisfying a given formula:
\begin{xalignat*}1
\Ex\X[1]\phi &\eqdef \Ex\X\phi \et \forall z.\Bigl(\Ex\X(\phi\et z) \impl
\All\X(\phi\impl z)\Bigr) \\
\Ex\X[\geq k]\phi &\eqdef\exists p_1,...,p_k.\ \Biggl[\All\X\Bigl(\ET_{i\not= j} \non
p_i \ou \non p_j\Bigr) \et \ET_{1\leq i\leq k} \Ex\X (p_i \et \phi)\Biggr]
\end{xalignat*}

\smallskip
As another example, let us mention that propositional quantification can be
used to flatten ``until'' formulas:
\begin{equation}
\Ex\phi_1\Until \phi_2 \equiv \exists z_1,z_2.\ 
  \Ex z_1\Until z_2 \et
  \All\G\left[
  z_1\thn \phi_1 \;\et\; z_2\thn \phi_2
  \right]
\label{eq-until}
\end{equation}
Actually, ``until'' can be expressed using only ``next'' and ``always''.
This is easily achieved in \QCTL*, where we would write e.g.
\begin{multline*}
\Ex \phi_1\Until \phi_2 \equiv 
  \exists z_1,z_2. \ \biggl[
    \Ex\Bigl(\bigl[z_2 \ou (z_1 \et \F z_2)\bigr]
    \et \\
    \G\bigl[z_1 \thn \X(z_1\ou z_2)\bigr] \Bigr)
  \et\All\G(\bigl[z_1\thn \phi_1 \;\et\; z_2\thn \phi_2\bigr]\biggr].
\end{multline*}

The expression in \QCTL is more involved. We~rely on a more general
translation through the $\mu$-calculus~\cite{Koz83}: in~this formalism, we can
express the ``until'' modality as a fixpoint:
\[
\Ex \alpha \Until \beta \equiv \mu T.\ (\beta \ou \Ex\X(\alpha \et T)).
\]
Now, a least-fixpoint formula $\mu T.\ \phi(T)$ (where we assume that
$\phi(T_1)\subseteq \phi(T_2)$ whenever $T_1\subseteq T_2$) can be
expressed\footnote{We have to be careful here with the exact notion of
  equivalence. We~keep it imprecise in this example, and develop the
  technical details in Section~\ref{sec-MSO}, where we prove that
  \QCTL without ``until'' can actually express the whole Monadic
  Second-Order Logic.}  in \QCTL as follows:
\[
\mu T.\ \phi(T) \equiv \exists t.\ \Bigl[\All\G(t \Leftrightarrow \phi(t)) \et
  \forall u.\ \bigl[
    \All\G(u \Leftrightarrow \phi(u)) \Rightarrow \All\G (t \thn u)
  \bigr]\Bigr]
\]
The first part of the formula (before quantifying over~$u$) states that the
labelling with~$t$ is a fixpoint. The second part enforces that it~precisely
corresponds to the least one. 
\end{example}


\subsubsection{Tree semantics.}

The tree-semantics is obtained from the structure semantics by seeing the
execution tree as an infinite-state Kripke structure.
We~write $\calS, q \models_t \phi$ 
to denote that formula~$\phi$ holds at~$q$ 
under the tree semantics.  Formally, seeing~$\calT_S(q)$ as an
infinite-state Kripke structure, we~define:
\begin{xalignat*}1
\calS, q\models_t \phi &\quad\text{iff}\quad 
  \calT_\calS(q),q\models_s \phi
\end{xalignat*}
Clearly enough, $\selfloop$ is always false under the tree semantics, while
$\uniq(\phi)$ holds if, and only if, $\phi$~holds at only one node of the
execution tree.

\begin{example}
Formula
$\displaystyle  \acyclic \eqdef  \All\G \bigl(
    \exists z.\ (z\et \uniq(z) \et \penalty0 \All\X\All\G\non z)
  \bigr)$
expresses that all infinite paths (starting from the current state)
are acyclic, which for \emph{finite} Kripke structures is always false
under the structure semantics and always true under the tree
semantics.
\end{example}

\subsubsection{Equivalences between \QCTL* formulas.}
We consider two kinds of equivalences depending on the semantics we
use. Two state formulas $\phi$ and~$\psi$ are said 
$s$-equivalent (resp.\ $t$-equivalent),  written $\phi\equiv_s \psi$
(resp.\ written $\phi\equiv_t \psi$)  if for any finite-state
Kripke structure~$\calS$ and any state~$q$ of~$\calS$, it~holds
$\calS, q \sat_s \phi$ iff $\calS, q \sat_s \psi$
(resp.\ $\calS, q \sat_t \phi$ iff $\calS, q \sat_t \psi$). 
We~write $\phi \equiv_{s,t} \psi$ when the
equivalence holds for both $\equiv_s$ and $\equiv_t$.

Note that both equivalences $\equiv_s$ and~$\equiv_t$ are
\emph{substitutive}, \ie,~a~subformula~$\psi$ can be replaced with any
equivalent formula~$\psi'$ without changing the truth value of the global
formula. Formally, if $\psi\equiv_s \psi'$ (resp.~$\psi\equiv_t
\psi'$), we have $\Phi[\psi] \equiv_s \Phi[\psi']$ (resp.~$\Phi[\psi]
\equiv_t \Phi[\psi']$) for any \QCTL* formula~$\Phi$.

\subsection{Fragments of \QCTL*.}\label{ssec-qnctl}

In the sequel, besides \QCTL and \QCTL*, we study several interesting fragments.
The first one is the fragment of \QCTL in \emph{prenex normal form}, \ie, in
which propositional quantification must be external to the \CTL formula.
We~write \EQCTL and \EQCTL* for the corresponding logics\footnote{Notice that
  the logics named \EQCTL and \EQCTL* in~\cite{Kup95a} are restrictions of
  our prenex-normal-form logics where only existential quantification is
  allowed. They correspond to our fragments \EQnCTL1 and \EQnCTL1*.}

We~also study the fragments of these logics 
with limited quantification.
For prenex-normal-form formulas, the fragments are defined as
follows:
\begin{itemize}
\item for any $\phi\in\CTL$ and any $p\in\AP$, $\exists p.\phi$ is an \EQnCTL1
  formula, and $\forall p.\phi$ is in \AQnCTL 1;
\item for any $\phi\in \EQnCTL k$ and any~$p\in \AP$, $\exists p.\phi$ is in
  \EQnCTL k and $\forall p.\phi$ is in \AQnCTL{k+1}. Symmetrically, if
  $\phi\in\AQnCTL k$, then $\exists p.\phi$ is in \EQnCTL {k+1} while $\forall
  p.\phi$ remains in \AQnCTL{k}.
\end{itemize}
Using similar ideas, we define fragments of \QCTL and \QCTL*. 
Again, the definition is inductive: 
\QnCTL1 is the logic $\CTL[\EQnCTL 1]$,
and 
$\QnCTL{k+1} = \QnCTL1[\QnCTL k]$. 
Notice that a more refined definition of \QnCTL k could be given,
where the index~$k$ would count quantifier alternation (in a way
similar to \EQnCTL k) instead of the mere quantifier depth that we use here.  This
however requires taking care of the number of negations between two
quantifiers, where ``negation'' here also includes hidden negations
(\eg a quantifier nested on the left-hand side of an ``until'' formula
should be considered negated). Our results would carry on to this
variant of \QnCTL k.

The corresponding extensions of \CTL*, which we respectively denote
with \EQnCTL k*, \AQnCTL k* and \QnCTL k*, are defined in a similar way.

\begin{remark}
Notice that \EQnCTL k and \AQnCTL k are (syntactically) included in \QnCTL k,
and \EQnCTL k* and \AQnCTL k* are fragments of \QnCTL k*.
\end{remark}

\section{Expressiveness}
\label{sec-express}

As a preliminary remark, let us mention that propositional
quantification increases the expressive power of~\CTL. For example, it
is easy to see that the formula $\uniq(P)$ defined in the previous
section allows us to distinguish between two bisimilar structures;
therefore such a formula cannot be expressed in~\CTL*. Note also that
it makes \QCTL (and~\QCTL*) to not be \emph{bisimilar invariant}. This
observation motivated an alternative semantics, called the
\emph{amorphous semantics}, for the propositional quantifications, in
order to make \QCTL (and~\QCTL*) bisimilar invariant. We~do not develop
this semantics further, and refer the reader to~\cite{Fre01} for more
details.

In this section we present several results about the expressiveness of
our logics for both the structure- and the tree semantics. We show
that \QCTL, \QCTL* and Monadic Second-Order Logic are equally
expressive. First we show that any \QCTL formula is equivalent to
a formula in prenex normal form (which extends to \QCTL* thanks to
Proposition~\ref{qctl-qctls}).

\subsection{Prenex normal form}

By translating path quantification into propositional
quantification, we can extract propositional quan\-tification out of
purely temporal formulas:\break for instance, $\EX  (\Q.\phi)$ where $\Q$ is
some propositional quantification is equivalent to\break 
$\exists z.\Q.  \Bigl( \uniq (z) \et \EX(z\et \phi) \Bigr)$. This
generalizes to full \QCTL under both semantics:

\begin{restatable}{proposition}{proprenex}\label{prop-fnp}
 In both semantics, \EQCTL and \QCTL are equally expressive.
\end{restatable}

\begin{proof}
  We prove the result for structure equivalence, turning a given a
  \QCTL formula~$\phi$ into prenex normal form. The transformation being
  correct also for infinite-state Kripke structures, the result
  for tree-equivalence follows.

\medskip

 In the following, we assume w.l.o.g. no atomic proposition is quantified twice. 
 We use $\Q$ to denote a sequence of quantifications, and write~$\bar{\Q}$ for the 
 dual sequence. 
  Our translation is defined as a sequence of rewriting rules that are to be
  applied in a bottom-up manner, first replacing innermost subformulas with
  $s$-equivalent ones. As~for~\CTL, we~only consider the temporal
  modalities \EX, \EU and \EG (which is sufficient since $\AX\phi \equiv_{s,t} \non
  \EX \non\phi$ and $\All \phi \Until \psi \equiv_{s,t}  \non \EG \non
  \psi \et  \non \Ex \non \psi \Until (\non \phi \et \non \psi)$).

\smallskip

For propositional and Boolean subformulas, we~have:
\begin{xalignat*}2
 \non \Q.\phi &\equiv_s \bar{\Q}\non\phi  &
 \Q_1.\phi_1 \ou \Q_2\phi_2 &\equiv_s \Q_1.\Q_2. (\phi_1 \ou  \phi_2) 
\end{xalignat*}

We now present the transformation for all three temporal modalities.
Extracting a bloc of quantifiers out of an~\EX operator can be done as follows:
\[
\Ex\X  \Q.\phi \equiv_s  \exists z.\Q.  \Bigl( \uniq (z) \et \Ex\X(z\et \phi) \Bigr)
\]
Here variable~$z$ (which is assumed to not appear in~$\Q.\phi$) is used to
mark an immediate successor that satisfies~$\Q.\phi$. We~require~$z$ to be
unique: allowing more than one successor would make the equivalence to be
wrong, as can be seen on the Kripke structures~$\calS_0$ of Fig.~\ref{fig-peq} using
formula $\Ex\X(\forall p.\;[(\Ex\F p) \thn p])$ (this formula is false, while formula $\exists
z.\; \forall p.\; \Ex\X(z \et [(\Ex\F p) \thn p])$ is true by labelling both states with~$z$).

Note that the right-hand-side formula is not yet in prenex form,
because $\uniq(z)$ involves a universal quantifier under a Boolean operator;
applying the above rules for Boolean subformulas concludes the translation for
this case.

\smallskip
For $\Ex\G(\Q.\phi)$, the idea again is to
label a short (lasso-shaped) path with~$z$, ensuring that $\Q.\phi$
always holds along that path:
\[
\Ex\G (\Q. \phi) \equiv_s  \exists z.  \forall z'.\Q. \Big(z \et \AG(z \impl
\Ex\X[1] z) \et (\uniq(z') \impl \AG((z\et z') \impl \phi))\Big).
\]

\smallskip
Finally, $\Ex (\Q_1.\phi_1) \Until (\Q_2.\phi_2)$ is handled by first
rewriting~it as
\[
\exists z_1,z_2.\ 
  \Ex z_1\Until z_2 \et
  \All\G\left[
  z_1\thn \Q_1.\phi_1 \;\et\; z_2\thn \Q_2.\phi_2
  \right]
\]
using Equivalence~\eqref{eq-until}, and by noticing that
$\AG (\Q.\phi) \equiv_s \forall z.\Q.(\uniq(z) \impl \AG(z \impl \phi))$.

\medskip
Before we prove correctness of the above equivalences, we~introduce a
useful lemma: 
\begin{lemma}\label{lemma-charac}
  Consider a Kripke structure $\calS=\tuple{Q,R,\ell}$, a~state~$q$ and
  a~\QCTL formula $\Q.\phi$ with $\Q= \Q_1 z_1\cdots \Q_k z_k$ and
  $\phi\in\CTL$. We~have $\calS, q \sat_s \Q.\phi$ if, and only~if, there is a non-empty
  family~$\xi$ of Kripke structures such that
\begin{enumerate} 
\item\label{eq-charac1} each $\calS'\in \xi$ is of the form $\tuple{Q,R,\ell'}$ where $\ell'$
  and $\ell$ coincide over $\AP\setminus \{z_1,\ldots,z_k\}$;
\item\label{eq-charac2} for any~$\calS'=\tuple{Q,R,\ell'}$ in~$\xi$, 
  any~$i$ with~$\Q_i=\forall$, and any ${\lab_{z_i}\colon Q \to
  2^{\{z_i\}}}$, there exists $\tuple{Q,R,\ell''} \in \xi$ such that
  $\ell'' \cap \{z_i\}=\lab_{z_i}$, and $\ell''$ and~$\ell'$ coincide over
  $\AP \setminus\{z_{i}\cdots z_k\}$; 
\item\label{eq-charac3} for all~$\calS'\in\xi$, it~holds $\calS', q \sat_s \phi$.
\end{enumerate}
\end{lemma}
A non-empty set $\xi$ satisfying the first two properties of
Lemma~\ref{lemma-charac} is said to be $(\Q,\calS)$-compatible.

\begin{proof}
The proof proceeds by induction on the number of quantifiers in~$\calQ$. The
equivalence is trivial when there is no quantification. Now assume that the
equivalence holds for some quantification~$\calQ$. 

\smallskip
We~first consider formula $\exists z.\calQ.\phi$. 
Assume $\calS, q\models_s \exists z.\calQ.\phi$. Then there exists a
structure~$\calS'=\tuple{Q,R,\ell'}$, with $\ell'$ coincides with~$\ell$
over~$\AP\setminus\{z\}$, such that $\calS',q\models_s \calQ.\phi$. Applying
the induction hypothesis to~$\calS'$, we~obtain a family of structures
satisfying conditions~\eqref{eq-charac1} to~\eqref{eq-charac3}
for~$\calS'$ and~$\calQ.\phi$. One easily checks that the very same family also fulfills
all three conditions for~$\calS$ and $\exists z.\calQ.\phi$.

Conversely, if there is a family of structures satisfying all three
conditions for~$\calS$ and $(\exists z.\calQ).\phi$.  Pick any
structure~$\calS'=\tuple{Q,R,\ell'}$ in that family. It~holds
$\calS',q\models_s \exists z.\calQ.\phi$, and moreover $\ell$
and~$\ell'$ coincide over $\AP\setminus\{z,z_1,\cdots,z_k\}$, where
$\{z_1,\cdots,z_k\}$ are the variables appearing in~$\calQ$. Hence
also $\calS, q\models_s \exists z.\calQ.\phi$.

\smallskip

Now consider formula $\forall z.\calQ.\phi$, and assume $\calS,
q\models_s \forall z.\calQ.\phi$. Then for any
$\calS'=\tuple{Q,R,\ell'}$ where $\ell$ and $\ell'$ coincide over
$\AP\setminus\{z\}$, we~have $\calS',q\models_s \calQ.\phi$. Applying
the induction hypothesis, for each such~$\calS'$, we~get a family of
Kripke structures satisfying all three conditions for $\calS'$ and
$\calQ.\phi$. Let~$\xi$ be the union of all those families. Then
$\xi$~clearly fulfills conditions~\eqref{eq-charac1}
and~\eqref{eq-charac3}.  Condition~\eqref{eq-charac2} for universal
quantifiers in~$\calQ$ follows from the induction hypothesis. For the
universal quantifier on~$z$, pick $\calS'=\tuple{Q,R,\ell'}$ and
$\lab_z$. Then by construction, $\xi$~contains a
structure~$\calS''=\tuple{Q,R,\ell''}$ where
$\ell''\cap\{z\}=\lab_z$. By~construction, $\xi$~contains
a family of structures satisfying all three conditions for~$\calS''$
and~$\calQ.\phi$, which entails the result.

If conversely there is a family~$\xi$ of structures satisfying the
conditions of the lemma, then for each $\lab_z$, this family contains a
structure~$\calS'=\tuple{Q,R,\ell'}$ with $\ell'\cap\{z\}=\lab_z$ and
satisfying~$\calQ.\phi$, which entails the result. 
\end{proof}

\smallskip 
We now proceed to the proof of the previous equivalences. 
We~omit the easy cases of propositional and Boolean formulas,
and focus on \EX and \EG:
\begin{itemize}
\item $\Ex\X(\Q.\phi)$: Assume $\calS, q \sat_s \Ex\X (\Q.\phi)$ with
  $\calS=\tuple{Q,R,\ell}$. Then there exists $(q,q')\in R$ such that $\calS, q'
  \sat_s \Q.\phi$. Therefore there exists a set~$\xi$ of Kripke structures that is
  $(\Q,\calS)$-compatible and such that $\calS',q' \sat_s \phi$ for every $\calS' \in
  \xi$. Now consider the set~$\xi'$ defined as follows:
\[
   \xi' \eqdef  \Bigl\{ \calS' =\tuple{Q,R,\ell'} \:\Bigm|\: \exists
   \tuple{Q,R,\ell''} \in \xi \text{ s.t. }  
  \ell' \eqdef \ell'' \oplus \{q'\mapsto z\} \Bigr\} 
\]
with:
\[
(\ell \oplus \{q\mapsto x\})(r)  \eqdef  
 \begin{cases} 
   \ell(r) \cup \{x\} & \text{ if } r=q \\ 
   \ell(r)\setminus \{x\} & \text{ otherwise} 
  \end{cases}  
\]

Then $\xi'$ is $(\exists z.\Q,\calS)$-compatible, and for
every Kripke structure~$\calS'\in\xi'$, we have: $\calS',q \sat_s
\uniq (z) \et \Ex\X(z\et \phi)$.  It~follows $\calS,q \sat_s
\exists z.\Q.( \uniq (z) \et \Ex\X(z\et \phi) )$.

\smallskip

Now assume $\calS, q \sat_s \exists z.\Q. \bigl( \uniq (z) \et
\Ex\X(z\et \phi) \bigr)$.  Then there exists a Kripke
structure~$\calS' \equiv_{\AP\setminus\{z\}} \calS$ such that
$\calS',q\sat_s \Q. \bigl( \uniq (z) \et \Ex\X(z\et \phi)
\bigr)$. In~particular, only one state~$q'$ of~$\calS'$ is labelled
with~$z$, and $q'$ is a successor of~$q$.  Moreover, there exists a
$(\Q, \calS')$-compatible set~$\xi$ such that for any~$\calS''\in\xi$,
it~holds $\calS'',q\sat_s \Ex\X(z\et\phi)$. Since only~$q'$ is
labelled with~$z$, we~have $\calS'',q'\sat_s \phi$, for
all~$\calS''\in\xi$. Hence $\calS',q'\sat_s \Q.\phi$, and
$\calS',q\sat_s \Ex\X(\Q.\phi)$. Finally, the formula is independent
of~$z$, so that also $\calS,q\models_s \Ex\X(\Q.\phi)$.

\smallskip
\item $\Ex\G (\Q. \phi)$: Assume $\calS,q \sat_s \Ex\G
  (\Q.\phi)$. There must exist a lasso-shape path $\rho = q_0q_1q_2\ldots
  (q_i\ldots  q_j)^\omega$,
  with $q_0=q$, along which $\Q.\phi$ always
  holds.  We~can also assume that $\rho$~is a \emph{direct}
  path, \ie, that $\calS$~does not contain a 
  transition~$(q_k,q_l)$ unless $l=k+1$ (otherwise a simpler witnessing path
  would exist). 
  Thus labeling all states of~$\rho$ with~$z$ makes the
  formula $(z\et \AG (z \impl \EX[1] z)$ hold at~$q$. Moreover,
  for every $k<j$, we~have $\calS,q_k\sat_s \Q.\phi$, so that there exists a
  set $\xi_k$ of Kripke structures that are $(\Q,\calS)$-compatible and such that 
  $\calS',q_k\sat_s \phi$ for every $\calS' \in \xi_k$. Now, let $\xi$
  be the following set of Kripke structures:
\begin{multline*}
  \xi \eqdef  \Bigl\{ \calS' =\tuple{Q,R,\ell'} \:\Bigm|\: \exists
  k<j.\ \exists \tuple{Q,R,\ell''} \in\xi_k  \mbox{ s.t. }   \\
  \ell' \eqdef \ell'' \oplus \{q_l\mapsto z\}_{l=0,\ldots,j-1}\oplus \{
  q_k \mapsto z' \} \Bigr\}.
\end{multline*}
For every $\calS'\in\xi$, we have 
\[
\calS',q \sat_s z \et \AG(z \impl
\EX[1] z) \et (\uniq(z') \impl \AG ((z \et z'\impl \phi))).  
\]
But the set $\xi$ is not $(\exists z. \forall z'. \Q,\calS)$-compatible:
it~only contains Kripke structures in which $z'$ labels a single state
of~$\rho$, while condition~\eqref{eq-charac2} requires that we consider all
labellings. It~suffices to extend~$\xi$ with arbitrary Kripke structures
involving all other forms of $z'$-labellings to obtain a compatible set~$\widehat{\xi}$.
Note that the additional Kripke structures still satisfy $(\uniq(z') \impl \AG
((z\et z'\impl \phi)))$. Applying Lemma~\ref{lemma-charac},
\[
\calS,q\models_s  \exists z.  \forall z'.\Q. \Big(z \et \AG(z \impl
\Ex\X[1] z) \et (\uniq(z') \impl \AG((z\et z') \impl \phi))\Big).
\]

Conversely, assume that this formula holds true at~$q$ in~$\calS$.
Accordingly, let $\calS'\equiv_{\AP\setminus\{z\}} \calS$ be the structure obtained
from~$\calS$ by extending its labelling with~$z$ in such a way that
\begin{enumerate}
\item $\calS',q \sat_s z \et \AG(z \impl\EX[1] z)$ 
\item $\calS',q \sat_s \forall z'.\Q (\uniq(z') \impl \AG((z\et z')\impl \phi))$. 
\end{enumerate}
The first property ensures that the $z$-labelling describes a lasso-shape
path starting from~$q$. The second one entails that there exists a
$(\forall z'\Q,\calS')$-compatible set~$\xi$ s.t. for every
$\calS''\in\xi$, we~have $\calS'',q \sat_s \uniq(z') \impl \AG((z\et
z')\impl \phi)$. This entails that for any position $k$ along
the $z$-path, there exists a $(\Q,\calS'')$-compatible set in which
$\calS,q_k\sat_s \phi$, which entails the result.\qed
\end{itemize}
\let\qed\relax
\end{proof}

\noindent Let us briefly measure the size and alternation depth of the resulting
formula: in terms of its size, the transformation never duplicates
subformulas of the initial formula, so that the final size is linear
in the size of the original formula. Regarding proposition
quantifiers, it can be checked that the alternation depth of the
resulting formula is at most two plus the sum of the number of nested quantifiers in
the original formula. In the end, the number of quantifier alternations of
the resulting formula is linear in the number of quantifiers in the
original formula.

\begin{remark}
  The translation used in the proof above to transform any \QCTL
formula into an equivalent formula in prenex normal form has been
defined for the structure semantics. It~is still correct when
considering the tree semantics but in this framework, we could define a
simpler transformation (in~particular, we can get rid of the
$\dpath(z_0,z_2)$ formula).
\end{remark}


\subsection{\QCTL and Monadic Second-Order Logic}
\label{sec-MSO}

We~briefly review Monadic Second-Order Logic (\MSO) over trees and
over finite Kripke structures (\ie, labelled finite graphs).  In both
cases, we~use constant monadic predicates~$\Pred{a}$ for ${a\in \AP}$ and
a relation~$\Edg$ either for the immediate successor relation in a
$2^\AP$-labelled 
tree $\tuple{T,l}$
or for the relation~$R$ in a finite KS $\tuple{Q,R,\ell}$.

\MSO is built with first-order (or individual) variables for vertices
(denoted with lowercase letters~$x,y,...$),  
monadic second-order variables for sets of vertices (denoted with
uppercase letters~$X,Y,...$).  Atomic formulas are of the form
$x=y$, $\Edg(x,y)$, $x\in X$, and~$\Pred{a}(x)$. Formulas are constructed
from atomic formulas using the boolean connectives and the first- and
second-order quantifier~$\exists$.  We~write
$\phi(x_1,...,x_n,X_1,...,X_k)$ to state that $x_1,...,x_n$
and $X_1,...,X_k$ may appear free (\ie not within the scope of a
quantifier) in~$\phi$. A~closed formula contains no free
variable. We~use the standard semantics for \MSO, writing
$\calM,s_1,...,s_n,S_1,...,S_k \sat
\phi(x_1,...,x_n,X_1,...,X_k)$ when $\phi$ holds on~$\calM$ when
$s_i$ (resp.~$S_j$) is assigned to the variable $x_i$ (resp.~$X_j$)
for ${i=1,...,n}$ (resp. ${j=1,...,k}$).

In the following, we compare the expressiveness of \QCTL with
\MSO over the finite Kripke structures (the structure semantics) and
the execution trees corresponding to a finite Kripke structure (tree
semantics). First note that \MSO formulas may express properties
on the whole trees or graphs, while our logics are interpreted over
\emph{states} of these structures. Therefore we use \MSO formulas with
one free variable~$x$, which represents the state where the
formula is evaluated.  Moreover, we~restrict the evaluation of \MSO
formulas to the \emph{reachable} part of the model from the given
state. This last requirement makes an important difference for the
structure semantics, since \MSO can express \eg that a graph is not
connected while \QCTL can only deal with what is reachable from the
initial state.

Formally, for the tree semantics, we say that $\phi(x)\in\MSO$ is
$t$-equivalent to some \QCTL* formula~$\psi$ (written $\phi(x) \equiv_t \psi$)
when for any finite Kripke structure~$\calS$ and any state $q\in \calT_\calS$,
it~holds $\calT_\calS(q),q\sat \phi(x)$ iff $\calT_\calS(q),q\sat_s \psi$.
Similarly, for the structure semantics: $\phi(x)$ is $s$-equivalent to~$\psi$
(written $\phi(x) \equiv_s \psi$) iff for any finite Kripke structure~$\calS$
and any state $q\in \calS$, it~holds $\calS_q,q\sat \phi(x)$ iff
$\calS_q,q\sat_s \psi$, where $\calS_q$~is the reachable part of~$\calS$
from~$q$.
For these definitions, we~have:
\begin{restatable}{proposition}{propmso}
\label{prop-qctl-mso}
Under both semantics, \MSO and \QCTL are equally expressive.
\end{restatable}

\begin{proof}
The translation from \QCTL to \MSO is easy: translating \CTL into \MSO
is standard and adding propositional quantifications can be managed with
second-order quantifications.

Now we consider the translation from \MSO to \QCTL.  This translation
(which is valid for both semantics) is defined inductively with a set
of rewriting rules. Given $\phi(x)\in\MSO$, we define
$\widehat{\phi}\in\QCTL$ as follows:
%
\begin{xalignat*}2
 \widehat{\non \phi} & \eqdef  \non \widehat{\phi} 
& 
  \widehat{\phi \et \psi} & \eqdef \widehat{\phi} \et \widehat{\psi}  
\\
   \widehat{\Pred{a}(x) } & \eqdef   a  
&
  \widehat{\Pred{a}(x_i) } & \eqdef   \EF (\pa_{x_i} \et a) 
\\
\noalign{\pagebreak[2]}
\widehat{x =x_i}  &\eqdef  \pa_{x_i} 
& 
\widehat{x_i=x_j} &\eqdef \EF (\pa_{x_i} \et \pa_{x_j}) 
\\
\noalign{\pagebreak[2]}
  \widehat{x \in X_i} & \eqdef    \pa_{X_i} 
&
\widehat{x_i \in X_j} & \eqdef    \EF (\pa_{x_i} \et \pa_{X_j}) 
\\
\noalign{\pagebreak[2]}
  \widehat{\Edg(x,x_i)}  & \eqdef   \EX \pa_{x_i} 
& 
  \widehat{\Edg(x_i,x_j)} & \eqdef  \EF(\pa_{x_i} \et \EX \pa_{x_j})   
\\
    \widehat{\exists X_i.\phi} & \eqdef  \exists \pa_{X_i}.\widehat{\phi} 
&
   \widehat{\exists x_i.\phi} & \eqdef  \exists \pa_{x_i}. \uniq(\pa_{x_i}) \et \widehat{\phi}  
\end{xalignat*}
%
The last rule not listed above concerns $\widehat{\Edg(x_i,x)}$, and
depends on the semantics: in the tree semantics, there is no edges
coming back to the root and the formula is then equivalent to false;
in the structure semantics, we have to mark the root with an
individual variable and use the same kind of rule as above:
\[
\widehat{\Edg(x_i,x)} \quad \eqdef \quad \begin{cases} \bot &
  \quad\mbox{in the tree semantics.} \\
  \EF(\pa_{x_i} \et \EX \pa_{x}) & \quad\mbox{in the structure
    semantics.}  \end{cases}
\]

The correctness of the translation w.r.t.\ both semantics is stated in
the two following Lemmas, whose inductive proofs are straightforward:
\begin{lemma}
  \label{lem-msoqctl-corr-tree}
For any $\phi(x,x_1,...,x_n,X_1,...,X_k)\in\MSO$, any finite
Kripke structure~$\calS$ and any state~$q$, we~have:
\[
\calT_\calS(q),q,s_1,...,s_n,S_1,...,S_k
\sat_s \phi(x,x_1,...,x_n,X_1,...,X_k) 
\quad \mbox{iff} \quad      
\calT'_\calS(q),q \sat_s \widehat{\phi}
\]
where $\calT_\calS$ and $\calT'_\calS$ only differ in the labelling of
propositions $\pa_{x_i}$ and~$\pa_{X_i}$: in~$\calT_\calS$, no~state
is labelled with these propositions, while in~$\calT'_\calS$, we~have
(1)~$\pa_{x_i} \in \ell(s)$ iff $s=s_i$ and (2)~$\pa_{X_i}\in\ell(s)$
iff ${s \in S_i}$.
\end{lemma}
As~a special case, we~get that $\calT_\calS(q),q\models \phi(x)$ if,
and only if, $\calT_\calS(q),q\models_s \widehat{\phi}$, which
entails $\phi(x) \equiv_t \widehat{\phi}$.

\smallskip
As regards the structure semantics, using similar ideas, we have:
\begin{lemma}
  \label{lem-msoqctl-corr-struct}
  For any $\phi(x,x_1,...,x_n,X_1,...,X_k)\in\MSO$, any finite Kripke
  structure~$\calS$ and and state~$q$, we~have:
\[
\calS_q,q,s_1,...,s_n,S_1,...,S_k
\sat \phi(x,x_1,...,x_n,X_1,...,X_k) 
\quad \mbox{iff} \quad      
\calS'_q,q \sat_s \widehat{\phi}
\]
where $\calS_q$ and $\calS'_q$ only differ in the labelling of
propositions $\pa_{x_i}$ and~$\pa_{X_i}$: in~$\calS_q$, no state is
labelled with these propositions, while in~$\calS'_q$ we have
(1)~$\pa_x \in \ell(s)$ iff $s=q$, (2)~$\pa_{x_i} \in \ell(s)$ iff
$s=s_i$ and (3)~$\pa_{X_i}\in\ell(s)$ iff $s \in S_i$.
\end{lemma}
In the end, after labelling state~$q$ with~$\pa_x$, we obtain
$\calS_q,q \sat \phi(x)$ if, and only if, $\calS'_q,q \sat_s
\widehat{\phi}$, where $\calS'_q$ only differs from~$\calS_q$ by the
fact that $q$ is labelled with~$p_x$.
It~follows that
$\phi(x) \equiv_s  \exists \pa_{x}. (\pa_{x} \et \uniq(\pa_x) \et  \widehat{\phi})$.
\end{proof}

\begin{remark}
\label{rem-muc}
One can also notice that it is easy to express fixpoint operators with
\QCTL in both semantics, thus \mucalculus can be translated into
\QCTL. 
For instance, 
the least fixpoint equation $\mu T. [b \ou (a \et \Ex\X T)]$ would be written as
\[
\exists T.\ \Bigl[ \AG \Big(
T \Leftrightarrow [b \ou (a \et \Ex\X T)]\Big) 
\;\et\;
  \forall U.\ \bigl\{ \AG (U \Leftrightarrow [b \ou (a \et \Ex\X U)])
   \Rightarrow  \AG (T \Rightarrow U) \bigr\}
\Bigr].
\]
Such a formula says that there is a fixpoint~$T$ such that for any
fixpoint~$U$, $T$~is included in~$U$; this precisely characterises
least fixpoints. 
Since the \mucalculus extended with counting capabilities has the same
expressiveness as~\MSO on trees~\cite{MR03}, we~get another evidence
that \QCTL can express all \MSO properties when interpreted over trees.
\end{remark}

\subsection{\QCTL and \QCTL*}

Finally, we show that \QCTL* and \QCTL are equally expressive for both
semantics. The main idea of the proof is an inductive replacement of
quantified subformulas with extra atomic propositions. Indeed note
that for any \CTL* state formula $\Phi$ and any \QCTL* 
state formula~$\psi$, we~have $\Phi[\psi] \mathrel{\equiv_{s,t}} \exists
p_\psi.\bigl( \Phi[p_\psi] \et \AG(p_\psi \Leftrightarrow \psi) \bigr)$
where $p_\psi$ is a fresh atomic proposition.
We have:

\begin{restatable}{proposition}{propnostar}
\label{qctl-qctls}
  Under both semantics, \QCTL* and \QCTL are equally expressive.
\end{restatable}

\begin{proof}
The result for the tree semantic has been shown in~\cite{Fre01}. Here
we give a different translation, which is correct for both semantics.  
Consider a \QCTL* formula~$\Phi$. The proof is
  by induction over the number~$k$ of  subformulas of~$\Phi$
  that  are not in~\QCTL.
  If $k=0$, $\Phi$~already belongs to~\QCTL.  Otherwise let~$\psi$ be
  one of the smallest $\Phi$-subformulas in
  $\QCTL*\setminus\QCTL$. Let $\alpha_i$s with $i=1,\ldots,m$ be the
  largest $\psi$-subformulas belonging to \QCTL (these are state
  formulas). 
Then $\psi[(\alpha_i \leftarrow p_i)_{i=1,\ldots,m}]$
is a \CTL* formula: every subformula of the form $\exists
p.\ \xi$ in~$\psi$ belongs to some \QCTL formula~$\alpha_i$, since
$\psi$~is one of the smallest $\QCTL*\setminus\QCTL$ subformula. 
Therefore $\psi$~is equivalent (w.r.t. both semantics)~to:
\[
\exists p_1 \ldots                 \exists
p_m.\       \Bigl(     \psi[(\alpha_i \leftarrow p_i)_{i=1,\ldots,m}] \et
\ET_{i=1,\ldots,m} \AG (p_i \iff \alpha_i) \Bigr)
\]
Since \CTL* can be translated into the \mucalculus~\cite{Dam94b}, 
and the \mucalculus can in turn be translated into \QCTL(see Remark.~\ref{rem-muc}), 
we~get that $\psi[(\alpha_i \leftarrow p_i)_{i=1,\ldots,m}]$ is equivalent to some \QCTL
formula~$\Omega$.  Hence
\[
\psi \equiv_{s,t} \exists p_1 \ldots \exists p_m.  \Bigl( \Omega \et
\ET_{i=1,\ldots,m} \AG (p_i \iff  \alpha_i) \Bigr)
\]
Now, consider the formula obtained from~$\Phi$ by replacing~$\psi$ 
with the right-hand-side formula above.
This formula is equivalent to~$\Phi$ and has at most $k-1$ subformulas in
$\QCTL*\setminus\QCTL$, so that the induction hypothesis applies.
\end{proof}

From Propositions~\ref{prop-fnp}, \ref{prop-qctl-mso} 
and~\ref{qctl-qctls}, we~get: 
\begin{corollary}
  Under both semantics, the four logics \EQCTL, \QCTL and \QCTL* and
  \MSO are equally expressive.
\end{corollary}

\medskip

\begin{remark}
\label{express-fqctls}
In~\cite{Fre01}, Tim~French considers a variant of \QCTL* (which we
call \FQCTL*), with  propositional quantification
within path formulas: $\exists p.\ \phipnormal$ is added in the definition
of path formulas. The semantics is defined as      follows: 
\[
\calS, \rho \sat_{s} \exists p. \phipnormal \quad\text{ iff }\quad 
\exists \calS' \equiv_{\AP\setminus\{p\}} \calS \text{ s.t. }\calS', \rho
\sat_{s} \phipnormal.
\]
It appears that this logic is not very different from~\QCTL* under the
tree semantics: French showed that \QCTL is as expressive as~\FQCTL*. 
Things are different in the structure-semantics
setting, where we now show that \FQCTL* is strictly more expressive than~\MSO. 
To begin with, consider the following formula:
\[
\Ex\G\bigl(\exists z. \forall z'. [\uniq(z) \et
\uniq(z') \et z \et \non z'] \impl  \X(\non z \Until z')\bigr).
\]
  This formula
  expresses the existence of an (infinite) path along which, between
  any two occurrences of the same state, all the other reachable states will be visited.
  This precisely corresponds to the existence of a Hamilton cycle, which is known not to
  be expressible in $\MSO$~\cite[Cor.~6.3.5]{EF95}. 
  However, note that the existence of a Hamilton cycle can be expressed in
  Guarded Second Order Logic \GSO{}\footnote{This logic is called \MSO[2]
    in~\cite{CE11}.}, in which quantification over sets of
  \emph{edges} is allowed (in~addition to quantification over sets of states).
  Still, \FQCTL* is strictly more expressive than~\GSO, as it~is easy to 
  modify the above formula  to express the existence
  of \emph{Euler} cycles: 
\begin{multline*}
  \Ex\G
  \Bigl(
  \exists x. \exists y. \forall x'. \forall y'. 
  \Bigl[ \step(x,y) \et \step(x',y') \et 
  \nstep(x,y) \et \non\nstep(x', y') \Bigr] \\
    \impl \X(\non \nstep(x, y) \Until \nstep(x', y'))
  \Bigr)
\end{multline*}
where $\step(x,y) \eqdef \uniq(x) \et \uniq(y) \et \Ex\F(x \et \X y)$
states that $x$ and~$y$ mark the source and target of a reachable transition, 
and $\nstep(x,y) \eqdef x \et \X y$ states that the next transition along the
current path jumps from~$x$ to~$y$. 
This can be seen to express the existence of an Euler cycle, which cannot be
expressed in~\GSO (otherwise evenness could also be expressed).

\begin{proposition}
  Under the structure semantics,   \FQCTL* is more expressive than
  \QCTL* and \MSO.  
\end{proposition}

\medskip
Nevertheless \FQCTL* model checking (see next section) is decidable:
for the tree semantics, it~suffices to translate \FQCTL* to \QCTL, as
proposed by~French~\cite{Fre01}.  The problem in the structure
semantics can then be encoded in the tree semantics:
for this we first need to extend the labelling of the Kripke
structure~$\calS$ with fresh propositions, one per state (\eg\ assume
that state~$q$ is labeled by~$\mathsf{p}_q$). Let~$\calS'$ be such
an extension (notice that the existence of an Euler path in such a
Kripke structure can be expressed in~\CTL). Then any quantification
$\exists P.\ \phi$ in some \FQCTL* formula~$\Phi$ (for the structure
semantics) is considered in the tree semantics. For this to be correct, 
we~augment~$\Phi$ with the extra requirement
that any two copies of the same state receive the same labelling.
We~thus build a formula $\widehat{\Phi}^{\calS'}$, by
replacing every subformula $\exists P.\ \psi$ in~$\Phi$ with 
\[
\exists P.\ \ET_{q\in Q} \Bigl( 
\EF\mathsf{p}_q \thn \bigl( \EF (\mathsf{p}_q \et P) \iff
\non\EF (\mathsf{p}_q \et \non P)\bigr)
\Bigr) \et \widehat{\psi}^{\calS'}.
\]
We then have: $\calS,q \sat_s \Phi$ iff $\calS',q \sat_t
\widehat{\Phi}^{\calS'}$.
\end{remark}

\section{Model checking}
We now consider the model-checking problem for \QCTL* and its
fragments under both semantics: given a finite Kripke
structure~$\calS$, a~state~$q$ and a formula~$\phi$, is~$\phi$ satisfied in
state~$q$ in~$\calS$ under the structure (resp.~tree) semantics? 
In this section, we~characterise the complexity of this problem. 
A~few results already exist, \eg for \EQnCTL1~ and \EQnCTL1* under both
semantics~\cite{Kup95a}. Hardness results for \EQnCTL2 and \EQnCTL2*
under the tree semantics can be found in~\cite{KMTV00}. Here we extend
these results to all the fragments of \QCTL* we have defined. We~also
characterize the program- and formula-complexities~\cite{Var82} of
model-checking for these fragments: the formula complexity (resp.\
program complexity) consists in evaluating the complexity of the
problem $\calS \sat \phi$ when the model $\calS$ (resp.~formula~$\phi$) is
assumed to be fixed. 
Except for Theorem~\ref{thm-form-qkctl-s}, our results hold true irrespective
of the notion of size (classical size of DAG-size) we use for \QCTL* formulas. 
Appendix~\ref{app-complex} proposes a short introduction to the complexity
classes used in the rest of the paper (especially the polynomial-time and
exponential hierarchies).

\subsection{Model checking for the structure semantics}
\label{sec-kripke}

\subsubsection{Fragments of \QCTL.}
First we consider the fragments of \QCTL with limited
quantifications: \EQnCTL k, \AQnCTL k, and \QnCTL k.
Prenex-normal-form formulas are (technically) easy to handle inductively:
a~formula in~\EQnCTL k can be checked by  
non-deterministically guessing a labelling and applying a
model-checking procedure for \AQnCTL{k-1}.  
We prove that the model-checking problems for these fragments populate the
polynomial-time hierarchy~\cite{Sto76}:

\begin{restatable}{theorem}{thmEQkCTLs}
\label{thm-eqkctl-s}
Under the structure semantics, model checking \EQnCTL k is \SSP
k-complete and model checking \AQnCTL k is \PPP k-complete.
\end{restatable}

\begin{proof}
We begin with noticing that an \AQnCTL k formula is nothing but the negation
of an \EQnCTL k formula. Hence it suffices to prove the result for \EQnCTL k.
The case where~$k=0$ corresponds to \CTL model-checking, which is
\PTIME-complete. 
For~$k>0$, hardness is easy, as \EQnCTL k model checking subsumes
the following problem, which is known to be \SSP k-complete~\cite{Pap94}:

\noindent\problem{\SSP k{}\textsf{SAT}}
  {$k$ families of variables $U_i = \{ u_1^i,
   \ldots,u_n^i\}$, and a propositional formula $\Phi(U_1,\ldots,U_k)$
   over $\bigcup_i U_i$}
  {is the 
   quantified Boolean 
   formula $\Q_1 U_1 \Q_2 U_2
   \ldots \Q_k U_k. \Phi(U_1,\ldots,U_k)$ true, where $\Q_i$ is $\exists$
   (resp.~$\forall$) when $i$ is~odd (resp.~even)?}
 Membership in \SSP k is proved inductively: an~\EQnCTL 1 instance
 $\exists u^1_1\ldots \exists u^1_k.\ \phi$  can be
 solved in \NP=\SSP1 by non-deterministically picking a labelling of the
 Kripke structure under study with atomic propositions~$u^1_1$ to~$u^1_k$, and then
 checking (in polynomial time) whether the \CTL formula~$\phi$ holds
 true in the resulting Kripke structure. Similarly, an \EQnCTL k
 formula $\exists u^1_1\ldots \exists u^1_k.\ \phi$, where $\phi$ is in
 \AQnCTL{k-1}, can be checked by first non-deterministically labelling the Kripke
 structure with atomic propositions~$u^1_1$ to~$u^1_k$, and checking the remaining
 \AQnCTL{k-1} formula~$\phi$ in the resulting Kripke structure. The latter is in
 \PPP{k-1} according to the induction hypothesis, so that the whole procedure
 is in \SSP k.
\end{proof}


When dropping the prenex-normal-form restriction, we~get

\begin{restatable}{theorem}{thmQkCTLs}
\label{thm-qkctl-s}
Under the structure semantics, model checking \QnCTL k is
\DDPlogn{k+1}-complete.
\end{restatable}

\begin{proof}
We~define the algorithm for \QnCTL k inductively: when $k=0$, we~just have a \CTL
model-checking problem, which is complete for $\PTIME=\DDPlogn 1$.
Assume that we have a \DDPlogn{k+1} algorithm for the \QnCTL k
model-checking problem, and 
consider a formula~$\phi\in\QnCTL {k+1}$: it~can
be written under the form
$
\phi \eqdef \Phi[(q_i \to \exists P_i.\ \psi_i)_i]
$
with $\Phi$ being a \CTL formula involving fresh atomic
propositions~$q_i$, and $\exists P_i.\ \psi_i$ are
subformulas\footnote{$\exists P_i$ denotes a sequence of existential
  quantifications.}  of~$\phi$.
The existential quantifiers in these subformulas are the outermost propositional
quantifiers in~$\phi$, and $\psi_i$ belongs to \QnCTL k, as we assume
that $\Phi$ is a \CTL formula.
As a consequence, $\exists P_i.\ \psi_i$ is a
state-formula, whose truth value only depends on the state in which it is
evaluated. For such a formula, we~can non-deterministically label
the Kripke structure with propositions in~$P_i$, and check whether $\psi_i$ holds in
the resulting Kripke structure. Computing the set of states satisfying
$\exists P_i.\ \psi_i$ is then achieved in $\NP^{\DDPlogn{k+1}}$, which is
equal to~$\SSP{k+1}$. 
Moreover, the queries for all the selected subformulas are independent
and can be made in parallel. It~just remains to check whether
the \CTL formula~$\Phi$ holds, which can be achieved in polynomial
time. This algorithm is thus in $\DDPlogn{k+2}$, since $\DDPlogn
{k+2}=\DDPpar{k+2}$ (see~\cite{Wag90}).

\medskip
We prove hardness using problems \PARITY{}(\SSP
k), defined as follows:
\problem{\PARITY{}(\SSP {k})}
  {$m$ instances of \SSP k{}\textsf{SAT}
   $Q^i_1 U^i_1 \ldots Q^i_k U^i_k.\penalty1000  \Phi^i(U^i_1,\ldots, U^i_k)$, where
  $Q^i_j=\exists$ when $j$ is odd and $Q^i_j=\forall$ otherwise}
 {is the number of positive instances even?}

This problem is \DDPlogn {k+1}-complete~\cite{Got95}. We~encode it
into a \QnCTL k model-checking problem.
Fix~$1\leq i\leq m$; the instance $\Psi_i=Q^i_1 U^i_1 \ldots Q^i_k
U^i_k.\penalty1000  \Phi^i(U^i_1,\ldots, U^i_k)$ of \SSP k{}\textsf{SAT} is
encoded as in the previous reduction, using a one-state Kripke structure~$\calS_i$ that
will be labelled with atomic propositions~$u^i_{j,l}$. The \QnCTL k formula to be
checked is then~$\Psi_i$ itself. We~label the unique
state of that Kripke structure with an atomic proposition~$x_i$, that
will be used in the sequel of the reduction.

Now, consider the Kripke structure~$\calS$  obtained as the ``union'' of the
one-state Kripke structures above, augmented with an extra state~$x_{m+1}$ and
transitions $(x_i,x_{i+1})$, for each $1\leq i\leq m$.  We~define
$
\phi \eqdef \OU_{1\leq i\leq m} (x_i \et \Psi_i)
$.
This formula holds true in those states~$x_i$ of~$\calS$ whose
corresponding \SSP k{}\textsf{SAT} instance is positive. It~remains to build a
formula for ``counting'' these sets: we~let
\begin{xalignat*}3
\psi_0 &\eqdef \Ex(\non\phi \Until x_{m+1}) 
&&\text{ and }&
\psi_{i+1} &\eqdef \Ex (\non\phi \Until (\phi \et \Ex\X \psi_i)).
\end{xalignat*}
It~is easily seen that $\psi_s$ holds true in state~$x_1$ of~$\calS$ iff
exactly $s$ of the $m$ instances of \SSP k{}\textsf{SAT} are
positive. Moreover, each~$\psi_i$ has quantifier height at most~$k$. 
The final formula is then the disjunction of the formulas~$\psi_{2i}$, for
$0\leq i\leq m/2$.
\end{proof}

\subsubsection{\EQCTL and extensions of \CTL*.} When considering
logics with no quantification restriction or the extensions of \CTL*, 
model-checking complexity becomes \PSPACE-complete:

\begin{restatable}{theorem}{thmQCTLs}
\label{thm-eqctl-s}\label{thm-qctl-s}\label{thm-eqkctl*-s}\label{thm-qkctl*-s}
\label{thm-eqctl*-s}\label{thm-qctl*-s}
Under the structure semantics, model
checking \EQCTL, \QCTL, \EQnCTL k*,
\AQnCTL k*, \QnCTL k*, \EQCTL* and \QCTL* is \PSPACE-complete.
\end{restatable}

\begin{proof}
  \PSPACE-hardness is straightforward because (1)~any instance of
  \textsf{QBF} is a special case of a model checking problem for every
  logic with unbounded quantifications (\EQCTL, \QCTL, \EQCTL* and
  \QCTL*) and (2)~the model-checking problem is \PSPACE-hard for \CTL*~\cite{SC85},
  hence also for any extension thereof.

  For \PSPACE membership, it is sufficient to show the result for
  \QCTL*. Consider a formula $\Phi \eqdef \exists p_1 \ldots \exists
  p_k.\phi$ with $\phi \in \CTL*$. We can easily consider the same
  kind of algorithm we used for \EQnCTL k in Theorem~\ref{thm-eqkctl-s}: we
  only replace the \CTL model-checking algorithm with a \CTL*
  model-checking algorithm running in polynomial space~\cite{CES86}. Since
  $\NP^{\PSPACE}=\PSPACE$, the resulting algorithm is in \PSPACE. This
  clearly provides a \PSPACE algorithm for any \QCTL* formula.
\end{proof}

\subsubsection{Program-complexity.}
Now we consider the \emph{program complexity} (or \emph{model complexity}) of
model checking for the structure semantics. In this context, we assume that
the formula is fixed, and the complexity is then expressed only in
terms of the size of the model.

\begin{restatable}{theorem}{thmprogEQkCTLs}
\label{thm-prog-eqkctl-s}
\label{thm-prog-eqkctl*-s}
  Under the structure semantics, for any~$k>0$, the program-complexity
 of model checking is \SSP k-complete for \EQnCTL k and \EQnCTL k*, and  \PPP
  k-complete for \AQnCTL k and \AQnCTL k*.
\end{restatable}

\begin{proof}
  Membership in~\SSP k for \EQnCTL k and \AQnCTL k comes directly from
  the general algorithms (Theorem~\ref{thm-eqkctl-s}). For \EQnCTL k*
  and \AQnCTL k*, we can use the same approach: 
  fix a formula $\Phi = \exists u_1^1 \ldots \exists u^1_k.\phi$ with $\phi
  \in \CTL*$. Deciding the truth value of $\Phi$ can be done in \NP
  by first non-deterministically guessing a labelling of the model with
  $\{u_1^1,...,u^1_k\}$ and then checking the \emph{fixed} formula~$\phi$
  (model checking a fixed formula of \CTL* is \NLOGSPACE-complete~\cite{Sch03a}).
  Thus with the
  same argument we used for the proof of Theorem~\ref{thm-eqkctl-s}, we
  get a \SSP k algorithm for any fixed \EQnCTL k* formula (and a
  \PPP k algorithm for a \AQnCTL k* formula).

We~now prove hardness in~\SSP k for \EQnCTL k (the results for \EQnCTL
k*, \AQnCTL k and \AQnCTL k* are proven
similarly). We~begin with the case where $k=1$ (for which the result is
already given in~\cite{Kup95a} with a proof derived from~\cite{HK94}):
quantification is encoded in the (fixed) \EQnCTL k formula, while the model encodes 
the SAT formula to be checked. We~begin with an \NP-hardness proof
for \EQnCTL 1, and then explain how it can be extended to~\EQnCTL k.

Consider an instance $\exists P.\ \phi_b(P)$, where $P$ is a set of
variables. We assume w.l.o.g. that propositional formula~$\phi_b$ is a conjunction of
disjunctive clauses. We begin with defining the model associated to~$\phi_b$,
and then build the formula, which will depend neither on~$\phi_b$, nor on~$P$.

Write~$\phi_b=\ET_{1\leq i\leq m} \OU_{1\leq j\leq n} \ell_{i,j}$, where
$l_{i,j}$ is in $\{p_k,\neg p_k \mid p_k\in P\}$. The model is defined as
follows:
\def\test{\mathsf{test}\xspace}
\begin{itemize}
\item it~has one initial state, named~$\phi_b$, $m$~states named~$C_i$ for
  $1\leq i\leq m$, and $3\size P$ states named~$p_k$, $\non p_k$ and
  $\test(p_k)$ for each~$p_k\in P$. 
\item there is a transition from~$\phi_b$ to each~$C_i$ and to
  each~$\test(p_k)$, a transition from each~$\test(p_k)$ to the
  corresponding~$p_k$ and~$\neg p_k$, and a transition from each~$C_i$ to its
  constitutive literals~$\ell_{i,j}$. Finally, each~$p_k$ and~$\neg p_k$
  carries a self-loop.
\item states~$\test(p_k)$ are labelled with an atomic proposition~$\test$,
  which is the only atomic proposition in the model.
\end{itemize}
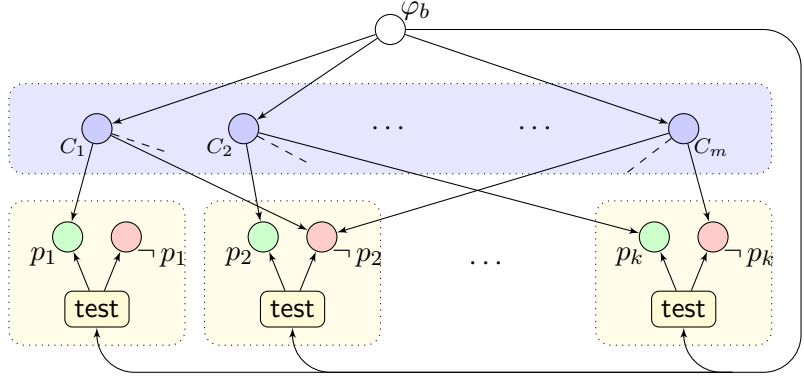
\begin{figure}[t]
\centering
\begin{tikzpicture}[yscale=1.2,xscale=1.3]
\draw (0,0.1) node[circle,draw,minimum width=4mm] (phi) {} node[above right] {$\phi_b$};
\draw[rounded corners=2mm,fill=blue!10!white,dotted] (-3,-1.5) -| (3.9,-.5) -| (-3.9,-1.5) -- (-3,-1.5);
\draw (-3,-1) node[circle,draw,minimum width=4mm,fill=blue!20!white] (c1) {} node[below left] {$\scriptstyle C_1$};
\draw (-1.5,-1) node[circle,draw,minimum width=4mm,fill=blue!20!white] (c2) {} node[below left] {$\scriptstyle C_2$};
\draw (3,-1) node[circle,draw,minimum width=4mm,fill=blue!20!white] (cm) {} node[below right] {$\scriptstyle C_m$};
\draw (0,-1) node {$\cdots$};
\draw (1.5,-1) node {$\cdots$};
\foreach \x/\i in {-3/1,-1/2,3/k} 
  {\draw[dotted,fill=yellow!10!white,rounded corners=2mm] (\x,-1.8) -| +(.9,-1.6) -- +(-.9,-1.6) |- +(0,0);
   \draw (\x,-3) node[draw,rounded corners=1mm,fill=yellow!20!white] (T\i) {$\test$};
   \draw (\x-.3,-2.2) node[draw,circle,minimum width=4mm,fill=green!20!white] (+\i) {} node[below left] {$p_{\i}$};
   \draw (\x+.3,-2.2) node[draw,circle,minimum width=4mm,fill=red!20!white] (-\i) {} node[below right] {$\non p_{\i}$};
}
\draw (1,-2.5) node {$\cdots$};
\foreach \x in {1,2,m} {\draw[-latex'] (phi) -- (c\x);}
\draw (3.5,-3.7) node[coordinate] (p) {};
\draw[rounded corners=5mm] (phi) -| +(4.2,-2) |- (p);
\foreach \i in {1,2,k} 
  {\draw[rounded corners=5mm,-latex'] (p) -| (T\i);
   \draw[-latex'] (T\i) -- (+\i);
   \draw[-latex'] (T\i) -- (-\i);
 }
\foreach \a/\b in {c1/+1,c1/-2,c2/+2,c2/+k,cm/-k,cm/-2}
  {\draw[-latex'] (\a) -- (\b);}
\foreach \a/\n in {c1/-20,c2/-30,cm/-140}
  {\draw[dashed] (\a) -- +(\n:+8mm);}
\end{tikzpicture}
\caption{The model used for proving \SSP k-hardness of model checking a fixed \EQnCTL k formula.}
\label{fig-fixedformula}
\end{figure}
Figure~\ref{fig-fixedformula} displays an example of this construction. 
The intuition is as follows: one of the states~$p_k$ and~$\neg p_k$ will be
labelled (via the \EQCTL formula) with an extra proposition~$\oplus$. That
exactly one of them is labelled will be checked by the $\test$-states. That
the labelling defines a satisfying assignment will be checked by the
$C_i$-states. The formula writes as follows:
\[
\Phi = \exists \oplus.\ [\All\X( \test \Rightarrow (\Ex\X\oplus \et \Ex\X\neg
\oplus))
\ \et 
\All\X(\non\test \Rightarrow \Ex\X\oplus)].
\]
One is easily convinced that a labelling with~$\oplus$ defines a valuation of
the propositions in~$P$ (by the first part of~$\Phi$), and that
$\phi_b$~evaluates to true under that valuation (by the second part of~$\Phi$).
Conversely, a satisfying assignment can be used to prove that $\Phi$~holds
true in the model.

\smallskip
This reduction can be extended to prove \SSP k-hardness of model
checking a fixed formula of~$\EQnCTL k$. Consider an instance of \SSP
k{}\textsf{SAT} of the form $\exists P_1\ldots \Q_k P_k.\ \phi_b(P_1,$ $\ldots,
P_k)$, assuming w.l.o.g. that the sets~$P_i$ are pairwise disjoint. The model
now involves $k$ $\test$-propositions~$\test_1$ to~$\test_k$, and a $\test$-state
associated with a proposition in~$P_l$ is labelled with $\test_l$. The rest of
the construction is similar. Assuming that $k$ is even (in which case $\Q_k$
is universal, and $\phi_b$~is a disjunction of conjunctive clauses---the dual
case being similar), formula~$\Phi_k$ then writes as follows:
\begin{multline*}
\Phi_k = \exists \oplus_1\ldots \forall \oplus_k. 
  [\All\X( \test_{2i+1} \Rightarrow (\Ex\X\oplus_{2i+1} \et \Ex\X\neg
\oplus_{2i+1}))
\ \et \\
[\All\X( \test_{2i+2} \Rightarrow (\Ex\X\oplus_{2i+2} \et \Ex\X\neg
\oplus_{2i+2}))] \Rightarrow (\All\X(\non\test \Rightarrow \Ex\X\oplus))].
\end{multline*}
\end{proof}

For \QnCTL k and \QnCTL k*, we have:

\begin{restatable}{theorem}{thmprogQkCTKs}
\label{thm-prog-qkctl-s}
\label{thm-prog-qkctl*-s}
Under the structure semantics, for any~$k>0$, the program-complexity of model
checking is \DDPlogn {k+1}-complete for \QnCTL k and \QnCTL k*.
\end{restatable}

\begin{proof}
\def\mplus{\oplus}
\def\mminus{\ominus}
To prove membership in \DDPlogn {k+1} for \QnCTL k*, we reuse the same
algorithm as for Theorem~\ref{thm-qkctl-s}: we~get the same complexity
for \QnCTL k* and \QnCTL k because program-complexity for 
\CTL* is in \PTIME (as~for~\CTL).

Now we prove hardness in \DDPlogn {k+1} for the fixed-formula
model-checking problem for \QnCTL k.  Fix some~$k$, and consider
of \PARITY(\SSP {k}), made of $m$ instances of \SSP k{}\textsf{SAT},
which we write $\Phi^i(U^i_1,\ldots, U^i_k)$ (assuming w.l.o.g. that
they all begin with an existential quantifier). We~begin with defining
a partial view of the Kripke structure that we will use for the
construction: it~has an initial state~$\mathsf{init}$ and a final
state~$\mathsf{final}$, and, for each~$1\leq i\leq m$, four states
labelled with~$i$ and either~$0$ or~$1$ (to indicate the parity of the
number of positive formulas up to~$\Phi_i$) and either~$\mplus$ or~$\mminus$ (to
indicate the validity of the $i$-th instance). Transitions are defined
as follows: from~$\mathsf{init}$, there is a transition to~$(1,0,\mplus)$
and~$(1,0,\mminus)$; from~$(i,0,\mplus)$ and $(i,1,\mminus)$, there are transitions
to~$(i+1,1,\mplus)$ and to~$(i+1,1,\mminus)$; from~$(i,1,\mplus)$ 
and~$(i,0,\mminus)$, there are transitions to $(i+1,0,\mplus)$ 
and~$(i+1,0,\mminus)$.  
Finally, there is a transition from~$(m,0,\mminus)$ and~$(m,1,\mplus)$
to~$\mathsf{final}$, and self-loops on $\mathsf{final}$, $(m,0,\mplus)$
and~$(m,1,\mminus)$. Consider a path in such a Kripke structure, and assume
that we can enforce that the path visits a~$\mplus$-state if, and only if,
the corresponding \SSP k{}\textsf{SAT} instance is positive. Then this
path reaches $\mathsf{final}$ if, and only if, the total number of
positive instances is even. Otherwise, the path will be stuck
in~$(m,0,\mplus)$ or in~$(m,1,\mminus)$. In~other words,
formula~$\Ex\F\,\mathsf{final}$ holds true if, and only if, the number
of positive instances is even.

It~remains to enforce the correspondence between positive states and positive
instances of \SSP k{}\textsf{SAT}. This is achieved using the reduction of the
proof of Theorem~\ref{thm-prog-eqkctl-s}: we~first extend the above Kripke
structure by plugging, at~each state~$(i,j,k)$, one copy of the Kripke
structure built in the proof of Theorem~\ref{thm-prog-eqkctl-s}. Now, the
formula to be checked in the resulting structure has to be reinforced as
follows:
\[
\Psi_k = \Ex( \mplus \Leftrightarrow \tilde\Phi_k) \Until \mathsf{final}
\]
where $\tilde\Phi_k$ is (a~slightly modified version~of) the formula built in
the proof of Theorem~\ref{thm-prog-eqkctl-s}.
\end{proof}

When model checking a fixed formula of \QCTL* (hence with fixed
alternation depth), there is no hope of being able to encode arbitrary
alternation: the program complexity of \QCTL* (and \QCTL) model
checking thus lies in the small gap between \PH and \PSPACE, unless
the polynomial-time hierarchy collapses:

\begin{restatable}{theorem}{thmprogQCTL-s}
\label{thm-prog-qctl-s}
\label{thm-prog-qctl*-s}
\label{thm-prog-eqctl-s}
\label{thm-prog-eqctl*-s}
Under the structure semantics, the program-complexity of model
checking is \PH-hard but not in \PH, and in \PSPACE but not \PSPACE-hard, for \EQCTL,
\QCTL, \EQCTL* and \QCTL*  (unless the polynomial-time
hierarchy collapses).
\end{restatable}

\begin{proof}
  From Theorem~\ref{thm-prog-eqkctl-s}, model-checking a fixed formula
  in \EQCTL or \EQCTL* is \PH-hard. Membership in \PSPACE follows from
  Theorem~\ref{thm-eqctl-s}.  If~these problems were in \PH, they~would
  lie in~$\SSP k$ for some~$k$, and the polynomial-hierarchy would
  collapse.  Similarly, if they were \PSPACE-hard, then a fixed
  formula (in \EQCTL or \EQCTL*, hence in \EQnCTL k or \EQnCTL k* for
  some~$k$) could be used to encode any instance of \textsf{QSAT},
  again collapsing the polynomial-time hierarchy. The same method
  applies to \QCTL and \QCTL* thanks to Theorem~\ref{thm-prog-qkctl-s}.
\end{proof}

\subsubsection{Formula-complexity.}
Now we consider the formula complexity of model checking for the
structure semantics. In this context we assume that the
the \emph{model} is assumed to be fixed, and the complexity is
then expressed only in term of the size of the formula. We will see
that every complexity result obtained for combined complexity also
holds for the formula complexity: these logics are expressive enough
to provide complexity lower bounds for fixed models.

In Theorem~\ref{thm-eqkctl-s}, complexity   lower-bounds for
model-checking \EQnCTL k and \AQnCTL k are proved with a fixed
model. Therefore these results apply also to formula complexity:  

\begin{restatable}{theorem}{thmformEQkCTLs}
\label{thm-form-eqkctl-s}
\label{thm-form-aqkctl-s}
  Under the structure semantics, the formula-complexity of model checking is \SSP k-complete for \EQnCTL
  k and \PPP k-complete for \AQnCTL k.
\end{restatable}

For \QnCTL k, we have the following result:
\begin{restatable}{theorem}{thmformQkCTL-s}
\label{thm-form-qkctl-s}
Under the structure semantics, the formula-complexity is \DDPlogn
{k+1}-complete for \QnCTL k when considering the \emph{DAG-size} of
\QCTL formulas. When considering the standard size of formulas, 
the~problem is in \DDPlogn{k+1} and \BH(\SSP k)-hard.
\end{restatable}

\begin{proof}
  Membership in \DDPlogn{k+1} can be proved using the same algorithm as in
  the proof of Theorem~\ref{thm-qkctl-s}, and noticing that its complexity is
  unchanged when considering the DAG-size of the
  formula. 

  In~order to prove hardness
  in~\DDPlogn {k+1}, we~again reduce \mbox{\PARITY(\SSP k)} to a
  model-checking problem for \QnCTL k over the Kripke structure~$S$ with one
  state and a self-loop as follows: consider an instance~$\mathcal{I}$
  of \PARITY(\SSP k) consisting in $m$ instances~$\Psi_i$ ($i=1,\ldots,m$) of
  \SSP k{}\textsf{SAT}. We~let $\alpha^1 = \non \Psi_1$ and $\alpha^{i+1} =
  (\non \Psi_{i+1} \et \alpha^i) \ou (\Psi_{i+1} \et \non \alpha^i)$. Clearly
  $\alpha^i$ holds in~$S$ iff there is an even number of positive instances in
  the set $\{\Psi_1, \ldots, \Psi_i\}$, so that the instance~$\mathcal{I}$ is
  positive iff $\alpha^m$ holds in~$S$. However, since $\alpha_i$ is
  duplicated in the definition of~$\alpha_{i+1}$, the reduction is in
  logarithmic space only if we represent the formula as a~DAG.

If we consider the usual notion of size of a formula, one can easily see that formula
complexity of \QnCTL k model checking is \SSP k-hard and \PPP k-hard.
Actually, as \CTL is closed under Boolean combinations, the problem is
hard for any level of the Boolean hierarchy \BH(\SSP k) over~\SSP k (we~refer
to~\cite{Hem98} for more details about Boolean hierarchies). 
\end{proof}

Finally formula complexity of \CTL* model-checking is already
\PSPACE-hard~\cite{Sch03a} and any \textbf{QBF} instance can be reduced to a
model-checking problem for \EQCTL  over a fixed structure. This
provides the complexity      lower-bounds of the following result (the
complexity upper-bound come from the general case, see Theorem~\ref{thm-qctl-s}):  

\begin{restatable}{theorem}{thmformQCTLs}
\label{thm-form-eqkctl*-s} \label{thm-form-qkctl*-s}
\label{thm-form-eqctl-s} \label{thm-form-qctl-s}
\label{thm-form-eqctl*-s} \label{thm-form-qctl*-s}
  Under the structure semantics, the formula-complexity is \PSPACE-complete for \EQnCTL k*, \AQnCTL k*, \QnCTL k*, \EQCTL, 
  \QCTL, \EQCTL*, and \QCTL*.  
\end{restatable}

\subsection{Model checking  for the tree semantics}
\label{sec-tree}


This section is devoted to \QCTL model checking over the tree semantics.
We~begin with proving a hardness result, extended techniques of~\cite{SVW87}
(for \QLTL) to the branching-time setting.

\paragraph{Hardness proof.} We prove that the \QnCTL k model-checking problems
populate the ex\-po\-nential-time hierarchy:

\begin{theorem}
\label{thm-hard-eqkctl-t} 
Model checking \EQnCTL k
under the tree semantics is \EXPTIME[k]-hard (for positive~$k$).
\end{theorem}

\begin{proof}  
The proof uses the ideas of~\cite{KMTV00,SVW87}: we~encode an
alternating Turing machine~$\calM$ whose tape is bounded by the following
recursively-defined function: 
\begin{xalignat*}2
E(0,n) &= n &  E(k+1,n) = 2^{E(k,n)}.
\end{xalignat*}
An~execution of~$\calM$ on an input word~$y$ of length~$n$ is then a tree. Our
reduction consists in building a Kripke structure~$K$ and a \QnCTL k
formula~$\phi$ such that $\phi$ holds true in~$K$ (for the tree semantics) iff
$\calM$ accepts~$y$.

\medskip
As~a first step, we~design a set of (polynomial-size) formulas of~\EQnCTL k that are
able to relate two states that are at distance $E(k,n)$ (actually,
a~slightly different value). This will be used in our reduction to
ensure that the content of one cell of the Turing machine is preserved
from one configuration to the next one, unless the tape head is
around.  
Define
\begin{xalignat*}2
F(0,n) &= n &  F(k+1,n) = F(k,n)\cdot 2^{F(k,n)},
\end{xalignat*}
and assume we are given a tree labelled with atomic propositions~$s$
and~$t$ (among others). We~first require that $s$ and~$t$ appear
exactly once along any branch, by means of the following formula
\[
\mathsf{once}(\phi) = \All\F\phi \et \All\G(\phi \impl \All\X\All\G\non\phi).
\]
Our formula for requiring one occurrence of~$s$ and~$t$ (in that
order) along each branch then reads
\begin{equation}
\mathsf{delimiters}(s,t) = \mathsf{once}(s) \et \mathsf{once}(t) \et
\All\G(s\impl \All\F t). 
\end{equation}

\begin{figure}[!ht]
\vspace{6 pt}
\begin{minipage}{.49\linewidth}
\centering
\begin{tikzpicture}[xscale=.8]
\path[use as bounding box] (-3.5,0.1) -- (3,-4.1);
\draw (0,0) node[coordinate] (A) {};
\draw (A) -- +(-120:3.99cm);
\draw (A) -- +(-60:3.99cm);
\draw[dashed] (A) -- +(-120:4.5cm);
\draw[dashed] (A) -- +(-60:4.5cm);
\draw[fill=black] (-70:12mm) circle(.5mm) node[coordinate] (B) {} node[left] {$s$};
\draw[snake=snake] (A) -- (B);
\draw (B) -- +(-100:2cm);
\draw (B) -- +(-80:2cm);
\draw[dashed] (B) -- +(-100:2.4cm);
\draw[dashed] (B) -- +(-80:2.4cm);
\path (B) -- +(-100:1.7cm) node[coordinate] (C) {};
\path (B) -- +(-80:1.7cm) node[coordinate] (D) {};
\draw (D) -- (C) node[midway,coordinate] (E) {} node[left] {$t$};
\draw[dotted] (E) -- +(2cm,0) node[pos=.95,coordinate] (G) {};
\draw[dotted] (B) -- +(2cm,0) node[pos=.95,coordinate] (F) {};
\draw[latex'-latex'] (F) -- (G) node[midway,right] {$F(k,n)$};
\draw[fill=black] (-110:18mm) circle(.5mm) node[coordinate] (B) {} node[right] {$s$};
\draw[snake=snake] (A) -- (B);
\draw (B) -- +(-100:2cm);
\draw (B) -- +(-80:2cm);
\draw[dashed] (B) -- +(-100:2.4cm);
\draw[dashed] (B) -- +(-80:2.4cm);
\path (B) -- +(-100:1.7cm) node[coordinate] (C) {};
\path (B) -- +(-80:1.7cm) node[coordinate] (D) {};
\draw (D) -- (C) node[midway,coordinate] (E) {} node[above left] {$t$};
\draw[dotted] (E) -- +(-1.8cm,0) node[pos=.95,coordinate] (G) {};
\draw[dotted] (B) -- +(-1.8cm,0) node[pos=.95,coordinate] (F) {};
\draw[latex'-latex'] (F) -- (G) node[midway,left] {$F(k,n)$};
\end{tikzpicture}
\caption{Chunks of height $F(k,n)$}\label{fig-yardstick}
\end{minipage}
\begin{minipage}{.49\linewidth}
\centering
\begin{tikzpicture}
\everymath{\scriptstyle}
\path[use as bounding box] (-1.5,.5) -- (1.5,-3.7);
\begin{scope}[xscale=.5,yscale=.27]
\begin{scope}[yshift=2cm]
 \draw (0,0) -| +(1,1) -| +(0,0); 
 \path (0,0) -- +(.5,.5) node {$\#$};
\end{scope}
\begin{scope}
\foreach \pos/\head/\cont in {0/1/1,1/0/1,2/0/0,3/0/\#}
  {\ifnum\head>0 \fill[black!30!white] (0,-\pos) -| +(1,1) -|
    +(0,0);\fi
   \draw (0,-\pos) -| +(1,1) -| +(0,0); 
   \path (0,-\pos) -- +(.5,.5) node {$\cont$};}
\draw[latex'-latex'] (2,1) -- (2,-3) node[midway,right] {$F(k,n)$};
\end{scope}
\begin{scope}[xshift=-1cm,yshift=-5cm]
\foreach \pos/\head/\cont in {0/0/0,1/1/1,2/0/0,3/0/\#}
  {\ifnum\head>0 \fill[black!30!white] (0,-\pos) -| +(1,1) -|
    +(0,0);\fi
   \draw (0,-\pos) -| +(1,1) -| +(0,0); 
   \path (0,-\pos) -- +(.5,.5) node {$\cont$};}
\end{scope}
\begin{scope}[xshift=1cm,yshift=-5cm]
\foreach \pos/\head/\cont in {0/0/1,1/1/1,2/0/0,3/0/\#}
  {\ifnum\head>0 \fill[black!30!white] (0,-\pos) -| +(1,1) -|
    +(0,0);\fi
   \draw (0,-\pos) -| +(1,1) -| +(0,0); 
   \path (0,-\pos) -- +(.5,.5) node {$\cont$};}
\end{scope}
\begin{scope}[xshift=-1.7cm,yshift=-10cm]
\foreach \pos/\head/\cont in {0/1/0,1/0/1,2/0/0,3/0/\#}
  {\ifnum\head>0 \fill[black!30!white] (0,-\pos) -| +(1,1) -|
    +(0,0);\fi
   \draw (0,-\pos) -| +(1,1) -| +(0,0); 
   \path (0,-\pos) -- +(.5,.5) node {$\cont$};}
\end{scope}
\begin{scope}[xshift=-.3cm,yshift=-10cm]
\foreach \pos/\head/\cont in {0/0/0,1/0/0,2/1/0,3/0/\#}
  {\ifnum\head>0 \fill[black!30!white] (0,-\pos) -| +(1,1) -|
    +(0,0);\fi
   \draw (0,-\pos) -| +(1,1) -| +(0,0); 
   \path (0,-\pos) -- +(.5,.5) node {$\cont$};}
\end{scope}
\begin{scope}[xshift=1cm,yshift=-10cm]
\foreach \pos/\head/\cont in {0/0/1,1/0/0,2/1/0,3/0/\#}
  {\ifnum\head>0 \fill[black!30!white] (0,-\pos) -| +(1,1) -|
    +(0,0);\fi
   \draw (0,-\pos) -| +(1,1) -| +(0,0); 
   \path (0,-\pos) -- +(.5,.5) node {$\cont$};}
\end{scope}
\draw[-latex'] (.5,2) -- (.5,1);
\draw[-latex'] (.3,-3) -- (-.5,-4);
\draw[-latex'] (.7,-3) -- (1.5,-4);
\draw[-latex'] (-.7,-8) -- (-1.2,-9);
\draw[-latex'] (-.3,-8) -- (.2,-9);
\draw[-latex'] (1.5,-8) -- (1.5,-9);
\draw[dashed] (1.3,-13) -- +(-.4,-.8);
\draw[dashed] (1.7,-13) -- +(.4,-.8);
\draw[dashed] (-1.2,-13) -- +(0,-.8);
\draw[dashed] (0,-13) -- +(-.4,-.8);
\draw[dashed] (.4,-13) -- +(.4,-.8);
\end{scope}
\end{tikzpicture}
\caption{Encoding runs of~$\calM$}\label{fig-TM}
\end{minipage}
\end{figure}
We now inductively build our ``yardstick'' formulas enforcing that,
along any branch, the distance between the occurrence of~$s$ and that
of~$t$ is precisely~$F(k,n)$ (see Fig.~\ref{fig-yardstick}). 
When $k=0$, this is easy:
\begin{equation}
\mathsf{yardstick}_0^n(s,t) = \All\G\Bigl(s \impl \bigl((\All\X)^n t \et
\ET_{0\leq k<n} (\All\X)^k \non t\bigr)\Bigr).
\label{eq-0}
\end{equation}
For the subsequent cases, we~use propositional quantification to
insert a number of intermediary points (labelled with~$r$), at
distance $F(k-1,n)$ apart. We~then associate with each occurrence
of~$r$ a counter, encoded in binary (with least significant bit on the
right) using a fresh proposition~$c$ on the
$F(k-1,n)$ cells between the present occurrence of~$r$ and the next
one.  Our global formula then looks as follows:
\begin{equation}
\mathsf{yardstick}_k^n = \exists r.\exists
c.\ (\mathsf{graduation}_k(r,s,t) \et 
 \mathsf{counter}_k(c,r,s,t)).
\label{eq-k}
\end{equation}

When~$k=1$, $\mathsf{graduation}_1(r,s,t)$ is rather easy (notice that
we allow graduations outside the $[s,t]$-interval):
\begin{equation*}
\mathsf{graduation}_1(r,s,t)= \All\G((s\ou t) \impl r) \et
    \mathsf{yardstick}_0^n(r,r).
\end{equation*}
As~regards the counter, we have to enforce that, between~$s$ and~$t$,
it~has value~zero exactly at~$s$ and value $2^n-1$ exactly at~$t$, and
that it increases between two consecutive $r$-delimited intervals:
\begin{xalignat*}1
\mathsf{counter}_1(c,r,s,t) &= \mathsf{zeros}_1(c,r,s,t) \et
  \mathsf{ones}_1(c,r,s,t) \et \mathsf{increment}_1(c,r,s,t) \\
\mathsf{zeros}_1(c,r,s,t) &= \All\G(s \iff (r \et \non c \et \All\X\All(\non
  c \Until r))) \\
\mathsf{ones}_1(c,r,s,t) &= \All\G((r \et \All\X\All(\non r \Until t))
  \impl \All(c\Until t)) \\
\mathsf{increment}_1(c,r,s,t) &= \All\G(s \impl 
  (\All\G(  (c \iff (\All\X)^n c)  \iff  \All\X\All(\non r \Until(\non
   c \et \non r))))).
\end{xalignat*}
The first two formulas are easy: $\mathsf{zeros}_1$ requires that
$s$~be the only position that can be followed by only zeros until the
next occurrence of~$r$; $\mathsf{ones}_1$ expresses that in the last
$r$-delimited interval before~$t$, $c$~always equals~$1$. Finally, 
$\mathsf{increment}_1$ requires that, starting from~$s$, the
value of~$c$ is changed from one interval to the next one if, and only
if, $c$~equals~one in all subsequent positions of the first
interval. One can check that this correctly encodes the incrementation of
the counter.
In the end, $\mathsf{yardstick}_1$ is an \EQnCTL1 formula.

For any~$k\geq 2$, $\mathsf{yardstick}_k$ is obtained using similar
ideas, with slightly more involved formulas. 
\begin{multline*}
\mathsf{graduation}_k(r,s,t)= \All\G((s\ou t) \impl r) \et 
  \forall u.\forall v.\ \Bigl[(\mathsf{delimiters}(u,v) \et
    \mathsf{yardstick}_{k-1}^n(u,v)) \impl \\ 
     (\All\G(u \impl \All\F(r \et \All\F v)) \et 
      \All\G((r \et \All\F v \et \non\All\F u) \impl
    \All\X \All (\non r \Until v)))\Bigr].
\end{multline*}
Roughly, this states that the labelling with~$r$ has to satisfy the
constraint that, between any two points $u$ and~$v$ at distance~$F(k-1,n)$
apart, there must be exactly one~$r$. Notice that formula
$\mathsf{yardstick}_{k-1}^n$ appears negated in
$\mathsf{graduation}_k$.
Regarding the counter, formulas $\mathsf{zeros}_k$ and
$\mathsf{ones}_k$ are the same as $\mathsf{zeros}_1$ and
$\mathsf{ones}_1$, respectively. Incrementation is handled using the
same trick as for $\mathsf{graduation}_k$: 
\begin{multline*}
\mathsf{increment}_k(c,r,s,t) = \forall u.\forall v.
  \Bigl[(\mathsf{delimiters}(u,v) \et  \mathsf{yardstick}_{k-1}^n(u,v))   \impl \\
   \All\G\bigl((s\et \All\F u) \impl (
    \All\G(((u\et c) \iff \All\G(v \impl c)) \iff  (\All\X\All \non r
    \Until (\non c \et \non r))
   ))\bigl)\Bigr]
\end{multline*}
This formula is a mix between $\mathsf{increment}_1$, in that it uses
the same trick of requiring that the value of~$c$ is preserved if
there is a zero at a lower position, and the labelling with~$u$
and~$v$ to consider all positions that are at distance $F(k-1,n)$
apart.

Now, since $\mathsf{yardstick}_{k-1}^n$ is, by~induction hypothesis,
in \EQnCTL {k-1}, formula $\mathsf{yardstick}_k^n$ is in \EQnCTL k
(notice that  $\mathsf{yardstick}_{k-1}^n$ appears negated after the
universal quantifiers on~$u$ and~$v$). 

\medskip

  We~now explain how we encode the problem whether a word~$y$ is accepted by an
  alternating Turing machine equipped with a tape of size~$E(k-1,\size y)$ 
  into an \EQnCTL k model-checking problem. 
  Assume we are given such a Turing machine $\calM=\tuple{Q,q_0,\delta,F}$ 
  on a two-letter alphabet~$\Sigma=\{\alpha,\beta\}$, and an input word~$y\in\Sigma^n$. 
  An~execution of~$\calM$ on~$y$ is encoded as (abstractly) depicted on
  Fig.~\ref{fig-TM}, with one configuration being encoded as a
  sequence of cells, and branching occurring only between two
  consecutive configurations.

  With~$\calM$, we associate a Kripke structure
  $\calS_\calM=\tuple{S,R,\ell}$ where
  $S=(Q\cup\{\epsilon\})\times(\Sigma\cup\{\circ\}) \cup\{\#\}$ (where $\circ$
  denotes empty cells of the tape and $\#$ will be used to delineate the
  successive configurations of~$\calM$), $R=S\times S$ is the complete
  transition relation, and $\ell$~labels each state with its name
  (hence the initial set of atomic propositions is~$S$). We~write~$s_0$ for
  the state~$\#$, where our formula will be evaluated.

  The execution tree of~$\calS_\calM$ from~$s_0$ contains as branches any word in~$s_0\cdot
  S^\omega$. We~use symbol~$\#$ to divide that tree into slices of
  height~$F(k-1,n)$: formula
  \begin{multline}
  \# \et \forall u.\forall v.\ 
    \Bigl[(\mathsf{delimiters}(u,v) \et  \mathsf{yardstick}_{k-1}^n(u,v))
      \impl \\
      \All\G((u \et \#) \impl \All\X\All (\non\#)
      \Until (v\et\#))\Bigr]
  \label{eq-slices}
  \end{multline}
  enforces that along any branch, symbol~$\#$ occurs at every level multiple
  of~$F(k-1,n)$.   Notice that this formula is in \AQnCTL {k-1}, since the
  $\mathsf{yardstick}_{k-1}^n$ formula is in \EQnCTL{k-1}.
  Notice that when~$k=1$, a simpler \CTL formula can be used.

  Now, not all branches of the execution tree of~$\calS_\calM$ are needed in
  order to represent an accepting execution of~$\calM$: only states labelled
  with~$\#$ may have several successors (see Fig.~\ref{fig-TM}). In~order to
  keep track of the relevant branches, we~label them with a fresh,
  existentially-quantified proposition~$a$. The fact that branching only occurs at
  $\#$-nodes can be expressed as
  \[
  a \et \All\G(a \thn \Ex\X a) \et 
    \All\G\Bigl[(a \et\non\#) \thn \ET_{p\not=q \in S} \non\bigl(\Ex\X(a \et p)
    \et \Ex\X(a \et q)\bigr)   \Bigr].
  \]
  Enforcing the initial state of the Turing machine (namely, that the tape
  contains~$y$, Turing machine is in state~$q_0$ and the tape head is on the
  first letter of~$y$) is straightforward (by expressing that $a$ must label the
  corresponding sequence of states in the tree). Expressing ``local''
  requirements on the encoding of a configuration (\eg that each configuration
  contains exactly one state representing a position for the tape head) is
  straightforward, using the delimiter~$\#$. The fact that an accepting state
  is reached along any~$a$-branch is also easy. It~only remains to express
  that there is a transition linking any configuration with its successor
  configurations. This can be achieved using similar formulas as
  formula~\eqref{eq-slices}, using delimiters~$u$ and~$v$ to ensure that the
  content of the tape is preserved and that the tape head has been moved by
  one position. 

  In the end, the global formula has an external existential quantification
  on~$a$, followed by formulas in~\AQnCTL{k-1} similar to
  formula~\eqref{eq-slices} (of \CTL formulas when~$k=1$). The whole
  formula is then in \EQnCTL k, which concludes the proof that model checking
  this logic is \EXPTIME[k]-hard.
\end{proof}

When using \CTL*, the above proof can be improved to handle one more
exponential: indeed, using \CTL*, formula $\textsf{yardstick}_0^n(s,t)$ can be
made to enforce that the distance between~$s$ and~$t$ is~$2^n$. This way,
using $k$ quantifier alternations, we~can encode the computation of an
alternating  Turing machine running in space~$(k+1)$-exponential. In~the~end:
\begin{theorem}\label{thm-k+1hard}
Model checking \EQnCTL k*
under the tree semantics is \EXPTIME[(k+1)]-hard (for positive~$k$).
\end{theorem}

\paragraph{Algorithms for the tree semantics.}
We~use tree-automata techniques to develop model-checking algorithms for our
logics. We recall the definitions and main results of this classical setting,
and refer to~\cite{MS87,MS95,Tho97b,KVW00} for a more detailed presentation.

\smallskip 
We~begin with defining alternating tree automata, which we will use
in the proof.  This~requires the following definition: the~set of
\newdef{positive boolean formulas} over a finite set~$P$ of
propositional variables, denoted with~$\PBF(P)$, is the set of formulas defined as
\[
\PBF(P) \ni \zeta \coloncolonequals p \mid \zeta\et\zeta \mid \zeta\ou\zeta
\mid \top \mid \bot
\] 
where~$p$ ranges over~$P$.
That a valuation~$v\colon P \to \{\top,\bot\}$ satisfies a formula
in~$\PBF(P)$ is defined in the natural way. We~abusively say that a
subset~$P'$ of~$P$ satisfies a formula~$\phi\in\PBF(P)$ iff the
valuation~$\mathds{1}_{P'}$ (mapping the elements of~$P'$ to~$\top$ and 
the elements of $P\smallsetminus P'$ to~$\bot$) satisfies~$\phi$. 
Since negation is not allowed,
if~$P'\models\phi$ and~$P'\subseteq P''$, then also $P''\models\phi$.

\begin{definition}
  Let~$\Sigma$ be a finite alphabet. Let $\calD \subseteq \Nat$ be a finite
  subsets of degrees.  An~\newdef{alternating parity $\calD$-tree
    automaton on~$\Sigma$}, or \APTA{\calD,\Sigma}, is a $4$-tuple
  $\Aut=\tuple{Q, q_0,\tau, \Omega}$ where
\begin{itemize}
\item 
$Q$ is a finite set of states,
\item 
$q_0\in Q$ is the initial state,
\item 
$\tau$ is a family of transition functions $(\tau_d)_{d\in\calD}$
  such that for all~$d\in \calD$, it~holds $\tau_d \colon Q\times\Sigma \to \PBF(\{0,\ldots,d-1\}\times Q)$,
\item 
$\Omega\colon Q \to \{0,\ldots,k-1\}$ is a parity acceptance condition. 
\end{itemize}
The~size of~$\Aut$, denoted by~$\size\Aut$, is the number of states in~$Q$.
The range~$k$ of~$\Omega$ is the \newdef{index} of $\Aut$, denoted
by~$\idx\Aut$.

A \newdef{non-deterministic $\calD$-tree automaton on~$\Sigma$}, or
\NPTA{\calD,\Sigma}, is a \APTA{\calD,\Sigma} where for any $d\in\calD$, $q\in
Q$ and $\sigma\in\Sigma$, we have: $\tau_d(q,\sigma) = \OU\limits_i \Bigl( \ET\limits_{0
  \leq c < d} (c,q_{i,c})\Bigr)$. 
\end{definition}

We now define the semantics of our tree automata. Notice that contrary to the
classical setting, where tree automata are defined to deal with fixed-arity
trees, we~better use the setting of~\cite{KVW00}, where the transition
function depends on the arity of the node where it is applied.
Let~$\Aut=\tuple{Q, q_0, \tau, \Acc}$ be an \APTA{\calD,\Sigma},
and~$\calT=\tuple{T,l_\calT}$ be a $\tuple{\Sigma,\calD}$-tree.
An~\newdef{execution tree} of~$\Aut$ on~$\calT$ is a $T\times
Q$-labelled tree $\calE=\tuple{E, p}$ such that
\begin{itemize}
\item 
$p(\epsilon)=(\epsilon,q_0)$, 
\item
 for each node~$e\in E$ with $p(e)=(t,q)$ and  $d = \deg_\calT(t)$,
 there exists a subset $\xi = \{(c_0,q'_0),\ldots,(c_m,q'_m)\} \subseteq
 \{0,\ldots,d-1\}\times Q$ such that
 $\xi \sat \tau_d(q,l_\calT(t))$, and for $i=0,\ldots,m$, we have
 $e\cdot i\in E$ and $p(e\cdot i)=(t\cdot c_i,q'_i)$.  
\end{itemize}
We write $p_Q$ for the labelling function restricted to the
second component: when $p(e)=(t,q)$, then $p_Q(e) = q$.
Given an infinite path~$\pi\in\Exec_\calE$ in an execution tree, 
$p_Q(\pi)$ is the set of states of visited along~$\pi$, and $\Inf(p_Q(\pi))$
is the set of states visited infinitely many times. 
An execution tree is \newdef{accepting} if $\min\{\Omega(q) \mid
q\in\Inf(p_Q(\pi))\}$ is even for every infinite path
$\pi\in\Exec_\calE$.
A~tree~$\calT$ is \newdef{accepted} by~$\calA$ if there exists an accepting execution
tree of~$\calA$ on~$\calT$.

Deciding whether a given tree is accepted by a tree automaton is decidable.
More precisely, given a tree automaton~$\calA$ and a regular tree~$\calT$
(\ie, a tree for which there exists a finite Kripke structure~$\calS$ and a
state~$q$ such that $\calT=\calT_{\calS}(q)$), the problem whether $\calT$ is
accepted by~$\calA$ is decidable. Moreover, given a tree automaton~$\calA$,
the problem whether $\calA$~accepts some tree at all is also decidable%
\footnote{Note that for an \APTA{}, emptiness checking and universality checking
have the same complexity because building the complement automaton 
can be done efficiently.}, and
when the answer is positive, $\calA$~accepts a regular tree. We summarise
these results in the following theorem:
\begin{theorem}\label{theorem-decision-proc}
The problem whether an \APTA{} $\calA$ with $d$ priorities accepts
  regular tree~$\calT$ represented as a Kripke structure~$\calS$ can be solved
  in time $O((\size\calA\cdot\size\calS)^d)$.
Checking the emptiness of an \APTA{}~$\calA$ is \EXPTIME-complete 
~\cite{Loeding2012}. 
Additionally, 
If~$\calA$ accepts some infinite tree, then it accepts a regular~one~\cite{Rab72}.
\end{theorem}

We now recall some standard properties of \APTA{}, which we will use later to
define our model-checking algorithm for \QnCTL k. 
First note that the use of Boolean formulae in the transition function makes the treatment of operations like union, intersection and complement easy to handled with  
\APTA{} and there is no cost in term of the size of the resulting automata~\cite{MS87}. 
Now we assume we are given
a Kripke structure $\calS=\tuple{Q,R,\ell}$ on a set~$\AP$ of atomic
propositions, and we write $\calD$ for the set of degrees in~$\calS$.

\begin{lemma}\cite{KVW00}\label{lemma-ctl}
Given a \CTL formula~$\phi$ over~$\AP$ and a set~$\calD$ of  degrees,
we~can construct a \APTA{\calD,2^{\AP}}~$\calA_\phi$ accepting exactly the
$2^{\AP}$-labelled $\calD$-trees satisfying~$\phi$. The~automaton~$\calA_\phi$
has size linear
in the size of~$\phi$, and uses a constant number of priorities.
\end{lemma}

\begin{proof}[Sketch of proof]
We~only describe the construction, and refer to~\cite{KVW00} for a detailed
proof of the result. W.l.o.g.~we assume that negations in~$\phi$ are
  followed by atomic propositions; This might require adding the two extra modalities
  $\Ex\WUntil$ and $\All\WUntil$, which satisfy the following equivalences:
\begin{xalignat*}1
\non (\Ex \phi \Until \psi) &  \equiv \All (\non \psi)
\WUntil (\non \psi \et \non \phi)   \\
\non (\All \phi \Until \psi) &  \equiv  \Ex (\non \psi)
\WUntil (\non \psi \et \non \phi)  
\end{xalignat*}
  The automaton $\calA_\phi = \tuple{Q_\phi,q_0,\tau,\Omega_\phi}$ is
  defined as follows:
\begin{itemize}
\item $Q_\phi$ is the set of state subformulas (not including $\top$ and~$\bot$),
\item the initial state~$q_0$ is~$\phi$,
\item given a degree $d\in\calD$, $\psi \in Q_\phi$ and
  $\sigma \in 2^{\AP}$, we define
  $\tau_d(\psi,\sigma)$ as follows:
\begin{xalignat*}3
  \tau_d(P,\sigma) &= \begin{cases} \top & \mbox{if} \: P\in \sigma \\ 
\bot & \mbox{otherwise} \end{cases} &&& 
  \tau_d(\non P,\sigma) &= \begin{cases} \bot & \mbox{if} \: P\not\in \sigma \\ 
\top & \mbox{otherwise} \end{cases} \\
\tau_d(\psi_1\et\psi_2) &= \tau_d(\psi_1) \et \tau_d(\psi_2) &&&
\tau_d(\psi_1\ou\psi_2) &= \tau_d(\psi_1) \ou \tau_d(\psi_2) \\[1.2ex]
\tau_d(\EX \psi,\sigma) &= \OU_{0\leq c <d} (c,\psi) &&&
\tau_d(\AX \psi,\sigma) &= \ET_{0\leq c <d} (c,\psi) \\[1.2ex]
\noalign{\pagebreak[1]}
\tau_d(\Ex \psi_1 \Until \Psi_2,\sigma)  &\multicolumn{5}{l}{$=\tau_d(\psi_2,\sigma) \ou
\Big(\tau_d(\psi_1,\sigma) \et \OU_{0\leq c <d} (c,\Ex \psi_1 \Until \psi_2)\Big)$}
\\[1.2ex]
\noalign{\pagebreak[1]}
\tau_d(\Ex \psi_1 \WUntil \Psi_2,\sigma) &\multicolumn{5}{l}{$=\tau_d(\psi_2,\sigma) \ou
\Big(\tau_d(\psi_1,\sigma) \et \OU_{0\leq c <d} (c,\Ex \psi_1 \WUntil \psi_2)\Big)$}
\\[1.2ex]
\noalign{\pagebreak[1]}
\tau_d(\All \psi_1 \Until \Psi_2,\sigma) &\multicolumn{5}{l}{$=\tau_d(\psi_2,\sigma) \ou
\Big(\tau_d(\psi_1,\sigma) \et \ET_{0\leq c <d} (c,\All \psi_1 \Until \psi_2)\Big)$}
\\[1.2ex]
\tau_d(\All \psi_1 \WUntil \Psi_2,\sigma) &\multicolumn{5}{l}{$=\tau_d(\psi_2,\sigma) \ou
\Big(\tau_d(\psi_1,\sigma) \et \ET_{0\leq c <d} (c,\All \psi_1 \WUntil \psi_2)\Big)$}
\end{xalignat*}
\item the acceptance condition is defined as
  \begin{xalignat*}1
    \Omega_\phi(\Ex \psi_1 \Until \psi_2) = \Omega_\phi(\All  \psi_1
      \Until \psi_2) &= 1 \\
    \Omega_\phi(\Ex \psi_1 \WUntil \psi_2) = \Omega_\phi(\All  \psi_1
      \WUntil \psi_2) &= 2
    \end{xalignat*}
  \end{itemize}
  Note that the definition of $\Omega_\phi$ for other nodes (\ie, boolean combinations)
  is not needed as these nodes cannot be visited infinitely often along a
  branch. 
\end{proof}

\QCTL quantifications over atomic propositions will be handled with the
projection operation on tree automata. 
To this aim, we will use the following lemma (notice that it requires
\emph{non-deterministic} tree automata as input):

\begin{lemma}\cite{MS85}\label{lemma-proj}
  Let~$\calA$ be a \NPTA{\calD,2^{\AP}}, with $\AP=\AP_1\cup\AP_2$. For
  all~$i\in\{1,2\}$, we~can build a \NPTA{\calD,2^{\AP}}~$\calB_i$ such that, for
  any $2^{\AP}$-labelled $\calD$-tree~$\calT$, it~holds: 
  $\calT \in \Lang(\calB_i)$ iff $\exists
  \calT'\in\Lang(\calA).\ \calT \equiv_{\AP_i} \calT'$.
  The size and index of~$\calB_i$ are those of~$\calA$.
\end{lemma}

In order to use this result, we~will have to apply the \emph{simulation
  theorem}, which allows for turning \APTA{}s into \NPTA{}s. 
Having varying degrees does not change the result (for example, one can 
adapt the proofs of Lemma~3.9 and Theorem~3.10 in~\cite{Loeding2012} in order
to get the result in our extended setting):
\begin{lemma}\cite{MS95}\label{lemma-sim}
Let $\calA$ be a \APTA{\calD,2^{\AP}}. We can build an \NPTA{\calD,2^{\AP}}~$\calN$
accepting the same language as~$\calA$, and such that $\size{\calN} \in
2^{O(\size\calA\idx\calA\cdot \log(\size\calA\idx\calA))}$ and $\idx\calN\in
  O(\size\calA\cdot \idx\calA)$. 
\end{lemma}

Now we are ready to describe the construction of the automaton for \QnCTL k:
\begin{theorem}\label{thm-aut-qnctl}
  Given a \QnCTL k formula~$\phi$ over~$\AP$ and a set~$\calD$ of degrees,
  we~can construct a \APTA{\calD,2^{\AP}}~$\calA_\phi$ accepting exactly the
  $2^{\AP}$-labelled $\calD$-trees satisfying~$\phi$. The automaton
  $\calA_\phi$ has size $k$-exponential and number
  of priorities $(k-1)$-exponential in the size of~$\phi$.
\end{theorem}

\begin{proof}
We proceed by induction over $k$.
\begin{itemize}
\item if $\phi \in \QnCTL 1$, then $\phi$ is of the form $\Phi[(\psi_i)_{1\leq
    i\leq m}]$ where $\Phi$ is a \CTL formula and $(\psi_i)_{1\leq i\leq m}$
  are \EQnCTL 1 formulas. We handle each $\psi_i$ separately. Assume that
  $\psi_i = \exists p_1\ldots\exists p_l.\ \psi'$ with $\psi' \in \CTL$. From
  Lemma~\ref{lemma-ctl}, one can build an \APTA{} $\calA_{\psi'}$ recognizing
  the $\calD$-trees satisfying~$\psi'$; moreover, $\size{\calA_{\psi'}}$ is in
  $O(\size{\psi'})$ and $\idx\calA_{\psi'}=2$. Applying Lemma~\ref{lemma-sim},
  we~get an equivalent \NPTA{}~$\calN_{\psi'}$ whose size is in
  $2^{O(\size{\psi'}\cdot \log(\size{\psi'}))}$ and number of priorities is in
  $O(\size{\psi'})$. Applying Lemma~\ref{lemma-proj} (to~$\calN_{\psi'}$ and for
  atomic propositions~$p_1,\ldots,p_l$), we~get an \NPTA{} $\calB_{\psi_i}=\tuple{Q_{\psi_i},q_0^{\psi_i},\tau_{\psi_i},\Omega_{\psi_i}}$
  recognizing the models of~$\psi_i$. The size of $\calB_{\psi_i}$ is in
  $2^{O(|\psi_i|\cdot \log(|\psi_i|))}$, and its number of priorities is in
  $O(|\psi_i|)$.

  Now to complete the construction it remains to construct the final automaton $\calA_\phi=\tuple{Q,q_0,\tau,\Omega}$. It is   based on the \APTA{} associated with the \CTL context~$\Phi[-]$
  (w.r.t.\ Lemma~\ref{lemma-ctl}) and the different \NPTA{}s built for the
  subformulas~$\psi_i$. Indeed the transition function~$\tau$ follows the
  rules of Lemma~\ref{lemma-ctl} for~$\Phi[-]$ and we just add the two
  following rules to deal with the subformulae~$\psi_i$ and their
  negations\footnote{Remember the construction for \CTL formulae assumes that
    negations precede atomic propositions.}: 
  \begin{itemize}
\item $\tau(\psi_i,\sigma) =   \tau_{\psi_i}(q_0^{\psi_i},\sigma)$, and   
\item $\tau(\non\psi_i,\sigma) =   \tau_{\bar{\psi_i}}(q_0^{\bar{\psi_i}},\sigma)$ where $\tau_{\bar{\psi_i}}$ is the transition function of     $\overline{\calB_{\psi_i}}$ (the dual of $\calB_{\psi_i}$). 
\end{itemize}
Therefore $\calA_\phi$ is an \APTA{} whose size is in
  $2^{O(|\phi|\cdot \log(|\phi|))}$ and its number of priorities is in
  $O(|\phi|)$.

\item if $\phi \in \QnCTL k$ with $k>1$, the construction follows almost the
  same steps as in the base case. Here 
  $\phi$~is of the form $\Phi[(\psi_i)_{1\leq i\leq m}]$, where $\Phi$ is a
  \CTL formula and each $\psi_i$ belongs to $\EQnCTL 1[\QnCTL {k-1}]$, \ie,
  is of the form $\exists p_1\ldots \exists p_l.\ \psi'$ with $\psi'\in \QnCTL
  {k-1}$. 

  From the~induction hypothesis, we~can build an \APTA{} $\calA_\psi'$
  recognizing the $\calD$-trees satisfying $\psi'$, and whose size is
  $(k-1)$-exponential and whose number~$d$ of priorities is
  $(k-2)$-exponential in $\size{\psi'}$.
  Applying Lemma~\ref{lemma-sim}, we get an equivalent \NPTA{} $\calN_{\psi'}$
  whose size is $k$-exponential in $\size{\psi'}$ (precisely in
  $2^{O(\size{\calA_\psi'}\cdot \log(\size{\calA_\psi'}\cdot d))}$) and whose number of
  priorities is $\size{\calA_{\psi'}}\cdot d$, \ie, $(k-1)$-exponential in
  $\size{\psi'}$. From Lemma~\ref{lemma-proj} (applied to~$\calN_{\psi'}$ and
  for propositions~$p_1,\ldots,p_l$), we~get
  an \NPTA{} $\calB_{\psi_i}$ recognizing the models of $\psi_i$. 
  Again the size and number of
  priorities of $\calB_{\psi}$ are identical to those of~$\calN_{\psi'}$. 

  Now we finish the construction as before in combining these \NPTA{}s with
  the \APTA{} provided by the \CTL context $\Phi$. This provides an \APTA{}
  $\calA_\phi$ whose size is $k$-exponential in $|\phi|$ and its number of
  priorities is $k-1$-exponential.\qed
\end{itemize}
\def\qed{}
\end{proof}

\noindent Combining this with the result of
Theorem~\ref{theorem-decision-proc}, we get our final result for
\QnCTL k:
%
\begin{theorem}
\label{thm-easy-qkctl-t} 
Model checking \QnCTL k
under the tree semantics is in \EXPTIME[k] (for positive~$k$).
\end{theorem}

The proof is easily adapted to the quantified extensions of \CTL*:
\begin{restatable}{theorem}{thmEQkCTL*t}
\label{thm-qkctl*-t} 
Model checking \QnCTL k* 
under the tree semantics is in \EXPTIME[(k+1)] (for positive~$k$).
\end{restatable}

\begin{proof}
  The proof proceeds along the same lines as in the proof of
  Theorem~\ref{thm-aut-qnctl}. However, we have to build automata for a \CTL*
  formula in the base case, so that the automaton for a \QnCTL 1* formula has
  size 2-exponential and number of priorities exponential.

  During the induction step, we~consider automata for \QnCTL {k-1}* formulas,
  apply the simulation theorem and projection, and combine them  with
  an (exponential-size) automaton for a \CTL* formula~\cite{KVW00}. One can easily see that
  the resulting automaton has size $k+1$-exponential, and number of priorities
  $k$-exponential. Theorem~\ref{theorem-decision-proc} then entails the result.
\end{proof}

From Theorems~\ref{thm-hard-eqkctl-t}, \ref{thm-k+1hard},
\ref{thm-easy-qkctl-t} and~\ref{thm-qkctl*-t}, we obtain:

\begin{corollary}
\label{coro-qkctl-t} 
\label{coro-qkctl*-t} 
Under the tree semantics, for any~$k>0$, 
model checking  \EQnCTL k, \AQnCTL k and \QnCTL k
is \EXPTIME[k]-complete, and model checking \EQnCTL k*, \AQnCTL k* and \QnCTL k*
is \EXPTIME[(k+1)]-complete.

\end{corollary}

It follows: 
\begin{restatable}{theorem}{thmEQCTLt}
\label{thm-eqctl-t} 
\label{thm-qctl-t} 
\label{thm-eqctl*-t}  
\label{thm-qctl*-t}  
Model checking \EQCTL, \QCTL, \EQCTL* and \QCTL*  
under the tree semantics is \TOWER-complete. 
\end{restatable}

\subsubsection{Program-complexity.} \label{tree-prog-comp}
When fixing the formula, the problem becomes much easier (in~terms of its
theoretical complexity): given~$\phi$, one~can build an
automaton~$\calA^2_\phi$ such that for every Kripke  structure~$\calS$,
deciding whether $\calS,q \sat_t \phi$ is equivalent to deciding whether 
the  unfolding of a \emph{variant} of $\calS$ is accepted by $\calA^2_\phi$   
which can be done in polynomial time (because 
$|\calA^2_\phi|$ is assumed to have constant size). 

First consider a finite Kripke structure $\calS=\tuple{Q,R,\ell}$. We define
$\calS_2=\tuple{ Q_2,R_2,\ell_2}$ with $Q_2 = Q \cup \Qint$, $\ell_2(q) =
\ell(q)$ for $q\in Q$ and $\ell_2(q)=\{\pint\}$ for $q\in \Qint$ ($\pint$ is a
fresh atomic proposition). Finally \Qint and $R_2$ are defined in order to
mimic $\calS$-transitions with nodes of degree~$2$: if $q_1,\ldots,q_k$ (with
$k>1$) are the $k$ successors of~$q$ in~$\calS$, we~have a transition
$(q,q_\epsilon)\in R_2$ and $q_\epsilon$ is the root of a complete binary tree
with $k$ leaves $q_1,\ldots,q_k$. For this, we need $k-1$ internal nodes
(including~$q_\epsilon$). If $q$ has a unique successor~$q'$ in~$\calS$, we
just add the transition $(q,q')$ to~$R_2$. Thus the size of~\Qint is $\sum_{q
  \in Q} \deg(x)-1$.

Now consider a \QCTL formula~$\Phi$ (that does not use~$\pint$). 
Let~$\widehat{\Phi}$ be the \QCTL formula defined as follows:
\begin{xalignat*}2
\widehat{\Ex \phi \Until \psi} &=  \Ex (\pint \ou
  \widehat{\phi}) \Until (\non\pint \et \widehat{\psi}) & 
\widehat{\phi \et \psi} &= \widehat{\phi} \et \widehat{\psi}
\\
  \widehat{\All \phi \Until \psi} &= \All (\pint \ou
  \widehat{\phi}) \Until (\non\pint \et \widehat{\psi}) & 
\widehat{\phi \ou \psi} & = \widehat{\phi} \ou \widehat{\psi}
\\
  \widehat{\EX \phi} & = \EX \Ex\, \pint 
   \Until (\non\pint \et \widehat{\phi}) & 
\widehat{P} & = P 
\\
  \widehat{\exists P.\ \phi} & =  \exists P.\ \widehat{\phi} &   \widehat{\non \phi} & = \non \widehat{\phi} 
\end{xalignat*}

The following lemma which relates the truth value of~$\Phi$ on~$\calS$ with that
of~$\widehat{\Phi}$ on~$\calS_2$:
\begin{lemma}
  For any Kripke structure $\calS=\tuple{Q,R,\ell}$, for any \QCTL
  formula~$\Phi$, and for any~$q\in Q$, we~have
$\calS, q \sat_t \Phi$ 
if, and only~if,
$\calS_2,q \sat_t \widehat{\Phi}$.
\end{lemma}

\begin{proof}
  The proof is done by structural induction over $\Phi$.
\end{proof}

\begin{restatable}{theorem}{thmprogALLt}
\label{thm-prog-all-t}
  Under the tree semantics, the program-complexity of model-checking
  is \PTIME-complete for all our fragments between \EQnCTL 1 and \AQnCTL 1 to
  \QCTL*.
\end{restatable}

\begin{proof}
  Let $\Psi$ be a \QCTL* formula. By Proposition~\ref{qctl-qctls}, we
  can build an equivalent \QCTL formula $\Phi$. Now consider
  $\widehat{\Phi}$ as above. Let $\calA_{\widehat{\Phi}}$ be the
  alternating tree automaton built for proving
  Theorem~\ref{thm-easy-qkctl-t} when $\calD=\{1,2\}$. 
It is easy to see that
  $\calA_{\widehat{\Phi}}$ recognizes the $2^\AP$-labeled trees having
  nodes with degrees in $\{1,2\}$ and satisfying $\widehat{\Phi}$.

Now consider a model-checking instance $\calS, q \sat_t \Phi$, it
clearly reduces to $\calS_2,q \sat_t \widehat{\Phi}$ and then to
$\calT_{\calS_2}(q) \in \calL(\calA_{\widehat{\Phi}})$. This last
problem can be solved in polynomial time (in $|\calS_2|$) when the size of
$\calA_{\widehat{\Phi}}$ is assumed to be fixed: following
Theorems~\ref{theorem-decision-proc} and~\ref{thm-aut-qnctl}, the polynomial
has degree $(k-1)$-exponential in~$\size\Phi$ when~$\Phi\in\QnCTL k$. 

\medskip 

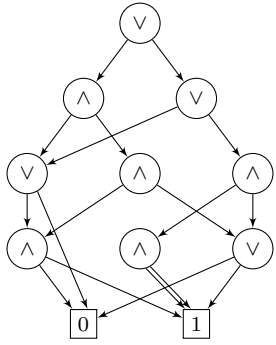
\begin{figure}
\centering
\begin{tikzpicture}[xscale=1.5,yscale=.5]
\draw (0,0) node[circle,draw,inner sep=1mm] (A) {$\scriptstyle\vee$};
\draw (.5,-2) node[circle,draw,inner sep=1mm] (B1) {$\scriptstyle\vee$};
\draw (-.5,-2) node[circle,draw,inner sep=1mm] (B2) {$\scriptstyle\wedge$};
\draw (1,-4) node[circle,draw,inner sep=1mm] (C1) {$\scriptstyle\wedge$};
\draw (0,-4) node[circle,draw,inner sep=1mm] (C2) {$\scriptstyle\wedge$};
\draw (-1,-4) node[circle,draw,inner sep=1mm] (C3) {$\scriptstyle\vee$};
\draw (1,-6) node[circle,draw,inner sep=1mm] (D1) {$\scriptstyle\vee$};
\draw (0,-6) node[circle,draw,inner sep=1mm] (D2) {$\scriptstyle\wedge$};
\draw (-1,-6) node[circle,draw,inner sep=1mm] (D3) {$\scriptstyle\wedge$};
\draw (.5,-8) node[draw,inner sep=1mm] (E1) {$\scriptstyle 1$};
\draw (-.5,-8) node[draw,inner sep=1mm] (E2) {$\scriptstyle 0$};
\path[-latex'] (A) edge (B1) edge (B2)
  (B1) edge (C1) edge (C3) 
  (B2) edge (C2) edge (C3)
  (C1) edge (D1) edge (D2)
  (C2) edge (D1) edge (D3)
  (C3) edge (D3) (C3.-60) edge (E2.75)
  (D3) edge (E2) edge (E1)
  (D2.-73) edge (E1.127) (D2.-60) edge (E1.110)
  (D1) edge (E2) edge (E1);
\end{tikzpicture}
\caption{A monotone boolean circuit}\label{fig-CV}
\end{figure}
We~now prove hardness in \PTIME, by reducing the \CIRCUITVALUE problem
(actually a simplified version of~it, with no negation~\cite{GJ79}) to a
model-checking problem for a fixed formula in \EQnCTL 1. Pick a monotone
circuit with fan-out~$2$ (\ie, an acyclic Kripke structure in which all states
are labelled with~$\circledOU$ or~$\circledET$ and have two outgoing edges,
except for two ``terminal'' states respectively labelled with~$0$ and~$1$, and
which only have a self-loop as outgoing edges). The value of the circuit is the
value of its initial node, defined in the natural way. See~Fig.~\ref{fig-CV}
for an instance of \CIRCUITVALUE (self-loops omitted) whose initial node
evaluates to~$1$. 

Consider the \CTL  formula
\[
\phi =  \All\G{}\Bigl[(1\thn p) \et (0\thn \non p)
 \et \bigl(\circledET \thn (p \iff \All\X p)\bigr)
 \et \bigl(\circledOU \thn (p  \iff \Ex\X p)\bigr)
\Bigr]
\]
and choose a labelling of the circuit with a new atomic proposition~$p$
satisfying~$\phi$ (in~the initial state). By~an easy induction, we get that a
state is labelled with~$p$ if, and only~if, its value is~$1$. As a
consequence, the circuit has value one if, and only if, the initial state of
its associated Kripke structure satisfies the fixed \EQnCTL 1 formula $\exists
p.\ (p \et \phi)$. 
The same equivalence holds for the \AQnCTL 1 formula $\forall p.(\phi \thn
p)$, as it is easily seen that only one $p$-labelling satisfies~$\phi$. 
\end{proof}

\subsubsection{Formula-complexity.}
The reductions used in general cases can be made to work with a fixed
model. Thus we have:

\begin{restatable}{theorem}{thmformALLkt}
\label{thm-form-all-k-t}
Under the tree semantics, the formula-complexity of model-checking is
\EXPTIME[k]-complete for \EQnCTL k and \QnCTL k, and it is
\EXPTIME[(k+1)]-complete for \EQnCTL k* and \QnCTL k*, for positive~$k$.
\end{restatable}

\begin{proof}
  Membership in \EXPTIME[k] (resp.~in \EXPTIME[(k+1)]) follows from the
  general case. We can adapt the hardness proof for \EQnCTL k
  (Theorem~\ref{thm-hard-eqkctl-t}) to work with a fixed Kripke
  structure~$\calS$. Let $\calM$ be a \EXPSPACE[(k-1)] alternating Turing
  machine $\tuple{Q,q_0,\delta,F}$ with $Q=\{q_1,\ldots,q_l\}$ and let $y$ be
  an input word over~$\Sigma$. Assume (w.l.o.g.) that $\Sigma =
  \{\alpha,\beta\}$.
  In~order to reuse the reduction performed for the general case, we~need to
  encode the unfolding of $\calS_\calM$ (defined in the proof of
  Theorem~\ref{thm-hard-eqkctl-t}) with an unfolding of some fixed Kripke
  structure~$\calS$ (hence with a fixed arity and a fixed labelling). For this we use
  the Kripke structure $\calS=\tuple{S,R,\ell}$ with $S=\{s_0,s_1\}$,
  $R=\{(s_0,s_0),(s_0,s_1), (s_1,s_0), (s_1,s_1)\}$ and $\ell(s_i)=\emptyset$
  for $i=0,1$: its unfolding is a binary tree and it contains no atomic
  proposition. We~use intermediary states (labeled by~$\pint$) to encode
  the arity of the states of~$\calS_\calM$ (\ie, every ``regular'' node~$s$ will
  be the root of some finite binary tree labeled with~$\pint$ and whose leaves
  correspond to all possible successors of~$s$ in~$K_\calM$), and we use an
  additional existential propositional quantification to describe the states
  (content of the tape, control state, separator,~...).
  Let $\AP$ be the set $\{{q_1},\ldots,{q_l},\alpha,\beta,\#\}$.
  The \EQnCTL k specification used to describe the existence of an accepting
  run of~$\calM$ over~$y$ is the following formula~$\Phi$ where
  $\Phi_{\calM,y}$ is a slight variant of the \EQnCTL k formula used in the
  proof of Theorem~\ref{thm-hard-eqkctl-t} and $\Phi_b$ is a \CTL formula
  described below:
\[
\Phi \eqdef 
\exists \pint \exists q_0. 
\ldots \exists q_l. \exists \alpha. \exists \beta. \exists \#. \Big( \Phi_{b}  \; \et \;  \Phi_{\calM,y} \Big)
\]
Formula~$\Phi_{b}$ describes the labelling of the binary tree 
with $\{\pint,p_{q_0},\ldots,p_{q_l},\alpha,\beta,\#\}$:
\begin{itemize}
\item  $\non\pint$ is true initially,
\item every state satisfying $\non\pint$ 
satisfies exactly one proposition in $\{\alpha,\beta,\#\}$ and at most one in $\{p_{q_0},\ldots,p_{q_l}\}$ can be associated with $\alpha$ or $\beta$, 
\item every $\non\pint$ state  can reach (along a $\pint$-labeled path) every  $\non\pint$ state  with one the following  labellings:
 \[
 \{\{\alpha\},\{\beta\},\{\#\},\{\alpha,{q_0}\},\ldots,\{\alpha,{q_l}\},\{\beta,{q_0}\},\ldots,\{\beta,{q_l}\}\}
 \]
\item and every $\pint$ state satisfies $\AF \non \pint$ (every intermediary state is inevitably followed by a regular state).
\end{itemize}

\noindent In $\Phi_{\calM,y}$, the only change we have to do is for
$\mathsf{yardstick}_0^n(s,t)$: we replace $\AX^n$ by a sequence of $n$
nested formulae of the form ``$\AX (\All \pint \Until (\non \pint \et
\ldots))$'' in order to find the $n$-th cell without considering
intermediary states.
Clearly $\Phi$ is still in \EQnCTL k and we have $y \in
\mathcal{L}(\calM)$ iff $K,s_0 \sat \Phi$. 

The same approach can be done for \EQnCTL k* and \QnCTL k* with an extra exponential level due to the ability of \CTL* to describe a run of length $2^n$. 
\end{proof}

As a consequence of the previous result, we have:

\begin{restatable}{theorem}{thmformEQCTLt}
\label{thm-form-EQCTL-t}
  Under the tree semantics, the formula-complexity of model-checking
  is \TOWER-complete for \EQCTL, \QCTL, \EQCTL* and \QCTL*.
\end{restatable}

\section{Satisfiability}

We now address the satisfiability problem for our two semantics: given
formula~$\phi$ in \QCTL* (or~a fragment thereof), does there exist a finite
Kripke structure~$\calS$ such that $\calS, q \models \phi$ for some state~$q$
of~$\calS$?

\paragraph{Structure semantics.}
As a corollary of our Prop.~\ref{prop-qctl-mso} and of the undecidability of
\MSO over finite graphs~\cite{See76}, we~have
\begin{theorem}
For the structure semantics, satisfiability of \QCTL and \QCTL* is undecidable.
\end{theorem}
Notice that satisfiability of \QCTL and \QCTL* is clearly recursively enumerable: the
finite Kripke structures can be enumerated and checked against the considered
formula, using our model-checking algorithms.

\smallskip
Similar undecidability results were already proved by French:
in~\cite{Fre01}, he~proved that satisfiability of \QCTL (more
precisely, \AQnCTL1) is undecidable over \emph{infinite-state} Kripke
structure; in~\cite{Fre03}, he~proved that satisfiability of \QLTL
over finite-state Kripke structures is undecidable (from which we~get
the result for \AQnCTL1*).

To complete the picture, we prove the following results:
\begin{theorem}\label{thm-undec}
Under the structure semantics, the satisfiability problem is
undecidable for \AQnCTL k (when~${k\geq 1}$) and \EQnCTL k
(when~${k\geq 2}$).
\end{theorem}

Notice that satisfiability for \EQnCTL 1 and \EQnCTL 1* under the structure
semantics is nothing but satisfiability of \CTL and \CTL*, respectively:
if~$\phi$ is in \CTL*, $\exists p. \phi$ is satisfiable if, and only~if,
$\phi$~is.

The core of the reduction is an \EQnCTL2 formula characterising those
Kripke structures that are finite grids\footnote{This follows the same
  ideas as in~\cite{Fre01} for proving undecidability of \AQnCTL1 over
  \emph{infinite} Kripke structures. However, the restriction to
  finite grids comes with many more technicalities (not~all states
  will have two successors, to~begin~with).}. We~say that a Kripke
structure is a finite grid when it is \emph{linear} (\ie, the
underlying graph~$(Q,R)$ is isomorphic to the graph $
L_m=([1,m],\{(i,i+1)\mid 1\leq i\leq m-1\} \cup\{(m,m)\}) $ for
some~$m$) or when it~is isomorphic to a product~$L_m\times L_n$ for
some integers~$m$ and~$n$.  The outermost existential quantifiers of
our \EQnCTL2 formula can then be handled together with the existential
quantification on the Kripke structure (because we deal with
satisfiability), so that our reduction will work for \AQnCTL 1.

The main idea behind our formula is a labelling of every other
\emph{horizontal (resp. vertical) line} of the structures with atomic
propositions~$h$ (resp.~$v$). Using universal quantification,
we~impose alternation between lines of~$h$-states and lines of $\non
h$-states using the following formula of the following form, enforcing
small \emph{squares} as depicted on Fig.~\ref{fig-grid}:
\[
\forall \gamma. \All\G\Bigl[(h \et \Ex\X (h \et\Ex\X (\non h\et\gamma)))  \thn
  \Ex\X(\non h \et \Ex\X(\non h \et \gamma))\Bigr].
\]
The full formula and a proof that it correctly
characterises finite grids are given in Appendix~\ref{app-grid}.
  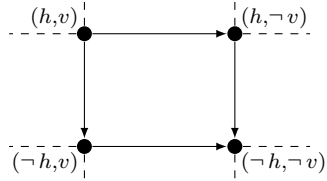
\begin{figure}[!ht]
    \centering
    \begin{tikzpicture}[inner sep=2pt]
      \everymath{\scriptstyle}
      \draw (0,0) node[fill,circle] (A) {} node[above left] {$(h,v)$};
      \draw (2,0) node[fill,circle] (B) {} node[above right] {$(h,\non v)$};
      \draw (0,-1.5) node[fill,circle] (C) {} node[below left] {$(\non h,v)$};
      \draw (2,-1.5) node[fill,circle] (D) {} node[below right] {$(\non h,\non v)$};
      \path[-latex] (A) edge (B) (B) edge (D);
      \path[-latex] (A) edge (C) (C) edge (D);
      \draw[dashed] (B) -- +(0:1cm);
      \draw[dashed] (D) -- +(0:1cm);
      \draw[dashed] (A) -- +(0:-1cm);
      \draw[dashed] (C) -- +(0:-1cm);
      \draw[dashed] (B) -- +(90:5mm);
      \draw[dashed] (D) -- +(-90:5mm);
      \draw[dashed] (A) -- +(90:5mm);
      \draw[dashed] (C) -- +(-90:5mm);
    \end{tikzpicture}
    \caption{The horizontal and vertical lines of a grid}\label{fig-grid}
  \end{figure}

\begin{proof}[Proof of Theorem~\ref{thm-undec}]
We~consider the following tiling problem: given a finite set of
different tiles, is~it possible to tile any finite grid? This problem
is easily shown undecidable, \eg by encoding the computations of a
Turing machine.

Now, we~encode the dual problem, expressing the existence of a finite
grid for which any tiling has a local mismatch. This can be achieved
using the formula for grids and an additional (universal)
quantification for placing tiles on that grid, for which a \CTL
formula will express the existence of a position where neighboring
tiles do not match.
\end{proof}

\paragraph{Tree semantics.}
We~prove the following result:
\begin{restatable}{theorem}{thmQCTL-sat-tree-k}
\label{thm-qctl-sat-tree-k} 
Under the tree semantics, satisfiability is \EXPTIME[(k+1)]-complete for
\AQnCTL k and \QnCTL k, and it is \EXPTIME[k]-complete for \EQnCTL k (for positive~$k$).
\end{restatable}

\begin{proof}
  We first prove \EXPTIME[(k+1)]-hardness for \AQnCTL k satisfiability. This
  is a straightforward adaptation of the proof of
  Theorem~\ref{thm-hard-eqkctl-t}: in that proof, from an input word~$y$ and
  an alternating Turing machine~$\calM$ equipped with a tape of
  size~$E(k,\size y)$, we built an \EQnCTL{k+1} formula of the form~$\exists
  a.\ \Phi$, where $\Phi$ is in \AQnCTL{k}, which is true in the Kripke
  structure~$\calS_\calM$ if, and only~if, $\calM$~accepts the input word~$y$.
  Now we claim that it~is also equivalent to the satisfiability of~$\Phi \et
  \chi_\calM$, where $\chi_\calM$ is an \AQnCTL 1-formula that encodes the
  Kripke structure~$\calS_\calM$. Letting
  $\AP=(Q\cup\{\epsilon\}) \times (\Sigma\cup\{\circ\}) \cup\{\#\}$,
  formula~$\chi_\calM$ expresses the fact that $\#$ is the initial state, and
  that any reachable state has exactly one successor labelled with~$p$, for
  each~$p\in \AP$ (see~the proof of Theorem~\ref{thm-hard-eqkctl-t} for the
  definition of~$\calS_\calM$). More precisely, $\chi_\calM$~can be written as
\[
\# \et 
\All\G\Bigl[\OU_{p\in \AP} \Bigl(p \et \ET_{q\in\AP\setminus\{p\}} \non q\Bigr)\Bigr]
\et
\forall z.
\ET_{p\in \AP}
  \All\G\Bigl[\Ex\X p \et (\Ex\X(p \et z) \thn \All\X(p \thn z))\Bigr].
\]
  We~write $\Phi_\calM$ for $\Phi \et \chi_\calM$. 

  First assume that $\exists a.\ \Phi$ is true in~$\calS_\calM$ (for
  the tree semantics). Then also $\exists a.\ (\Phi_\calM)$ is
  true in that Kripke structure (notice that $\chi_\calM$ does not
  depend on~$a$).  Applying Theorem~\ref{thm-aut-qnctl} to~$\Phi_\calM$ for
  $\calD=\{\size \AP\}$, we
  know that there is an alternating tree automaton~$\calA_{\Phi_\calM}$
  accepting exactly the $\calD$-trees in which~$\Phi_\calM$~holds. Since~$\exists
  a.\ (\Phi_\calM)$ is true in~$\calS_\calM$, we~know that~$\calA_\Phi$
  accepts at least one tree. From Theorem~\ref{theorem-decision-proc},
  it~accepts a regular tree. As~a consequence, there is a finite-state
  Kripke structure where~$\Phi_\calM$~holds.
  Conversely, if~$\Phi_\calM$~is satisfiable, then $\exists a.\ \Phi_\calM$
  also~is (for~the tree semantics). Pick a Kripke structure~$\calS$ satisfying
  $\exists a.\ \Phi_\calM$. Then $\calS$~satisfies $\chi_\calM$, so that the
  unwindings of~$\calS$ and of~$\calS_\calM$ are the same, and state~$\#$
  of~$\calS_\calM$ satisfies $\exists a.\ \Phi$.

      This proves the lower bound for \AQnCTL k and \QnCTL k. 
      As \EQnCTL k contains \AQnCTL {k-1}, it also gives the
      complexity lower-bound for \EQnCTL k.

\smallskip
      Now we prove membership in \EXPTIME[(k+1)] for \QnCTL k satisfiability. Let $\Phi$ be
      a \QnCTL k formula. Let $\widehat{\Phi}$ be the corresponding formula
      over $\{1,2\}$-trees, as defined at Section~\ref{tree-prog-comp}. Let $\Phi_2$ be
      the \CTL formula 
      $\non\pint \et \All\G \All\F \non \pint$.
      The following lemma allows us to reduce the satisfiability problem
      for $\Phi$ to the satisfiability problem for $\widehat{\Phi}\et\Phi_2$ for
      $\{1,2\}$-trees:
 \begin{lemma}
 \label{lem-sat-12tree}
   There exists a Kripke structure $\calS$ with a state $q$ such that $\calS,q
   \sat_t \Phi$ if, and only if, there exists a regular $\{1,2\}$-tree $\calT$
   such that $\calT,\varepsilon \sat_s \widehat{\Phi} \et \Phi_2$.
\end{lemma}

\begin{proof}
  If $\calS,q \sat_t \Phi$ then $\calS_2,q \sat_t
  \widehat{\Phi}$ (see   Section~\ref{tree-prog-comp}). Since $\calS_2,q\sat_t
  \Phi_2$, the result follows.

  We~now prove the converse direction. Assume $\calT,\varepsilon \sat_s
  \widehat{\Phi}\et\Phi_2$. Note that the root satisfies $\non\pint$ (enforced
  by~$\Phi_2$). Now consider
  the tree $\calT'$ where the intermediary state (those labeled with $\pint$)
  has been removed, and replaced with direct
  transitions from their father node to their child nodes. 
  This new tree is still regular and it can be easily proved
  to satisfy~$\Phi$ (by construction of~$\widehat\Phi$).
\end{proof}

From the previous result, it remains to build an \APTA{\{1,2\},2^\AP} for
$\widehat{\Phi}\et\Phi_2$, as described in Theorem~\ref{thm-aut-qnctl}. The
size of this automaton is $k$-exponential in~$\Phi$. Checking emptiness can be
done in time $(k+1)$-exponential.

Now consider the case of some \EQnCTL k formula~$\Phi$. Let~$\Phi'$ be the
\AQnCTL {k-1} formula obtained from~$\Phi$ by removing the outermost
existential block of quantifiers. Clearly~$\Phi'$ is satisfiable iff~$\Phi$
is: satisfiability implicitly uses an existential quantification
over atomic propositions, which can include the first block of existential
quantifications. This provides us with a \EXPTIME[k] algorithm for \EQnCTL k satisfiability.
\end{proof}

The result is lifted to fragments of \QCTL* as follows:
\begin{restatable}{theorem}{thmQCTLs-sat-tree-k}
\label{thm-qctls-sat-tree-k} 
Under the tree semantics, satisfiability is \EXPTIME[(k+2)]-complete for \AQnCTL
k* and \QnCTL k*, and it is \EXPTIME[(k+1)]-complete for \EQnCTL k*.
\end{restatable}
\begin{proof}
  As explained above, using \CTL* allows us to gain an exponential level in
  the reduction: $\mathsf{yardstick}_0^n(s,t)$ can enforce that the distance
  between $s$ and~$t$ is~$2^n$.
\end{proof}

Finally we have:
\begin{restatable}{corollary}{thmQCTL-sat-tree}
\label{cor-qctl-sat-tree} 
Under the tree semantics, satisfiability is \TOWER-complete for \EQCTL, \QCTL,
\EQCTL* and \QCTL*.
\end{restatable}

\begin{remark}
  In our definition, the satisfiability problem for the tree semantics asks
  for the existence of a finite Kripke structure. Another satisfiability
  problem can be considered, asking for the existence of a labeled tree~$T$
  satisfying the input formula. This version of the problem is also decidable,
  and is equivalent to the previous one. Indeed, assume that such a $T$
  exists, then there is a finitely branching tree~$T'$ satisfying~$\Phi$,
  because $\Phi$ is equivalent to some MSO formula and MSO has the
  finite-branching property. As $T' \sat_s \Phi$, we clearly have that
  $\widehat{\Phi}\et\Phi_2$ is satisfiable by some $\{1,2\}$-tree; thanks to
  Lemma~\ref{lem-sat-12tree}, we get that there exists a finite Kripke
  structure satisfying~$\Phi$ for the tree semantics.
\end{remark}

\section{Conclusions}


\begin{table}[ht]
\hfuzz=50pt
  \centerline{%
\def\arraystretch{1.2}
\begin{tabular}{|c||W|X|X|Z|}
\cline{2-5}
\multicolumn{1}{c||}{} & satisfiability & 
  model checking & formula-compl. & program-compl. \tabularnewline
\hline
\hline
\EQnCTL k & &
  \multicolumn{3}{c|}{\SSP k-complete  (Th.~\ref{thm-eqkctl-s},~\ref{thm-prog-eqkctl-s},~\ref{thm-form-eqkctl-s})} 
  \tabularnewline
\cline{1-1}\cline{3-5}
\QnCTL k & \EXPTIME-c. for \EQnCTL 1  &
   \multicolumn{3}{c|}{\DDPlogn {k+1}-complete\footnotemark{} 
  (Th.~\ref{thm-qkctl-s},~\ref{thm-prog-qkctl-s},~\ref{thm-form-qkctl-s})} \tabularnewline
\cline{1-1}\cline{3-5}
\EQnCTL k* & \EXPTIME[2]-c. for \EQnCTL 1* &
   \multicolumn{2}{c|}{} &
   \SSP k-complete  (Th.~\ref{thm-prog-eqkctl*-s}) 
   \tabularnewline
\cline{1-1}\cline{5-5}
\QnCTL k* & Undecidable otherwise & \multicolumn{2}{c|}{\PSPACE-complete}
&
  \leavevmode\hbox to 0pt{\hss\DDPlogn {k+1}-c.(Th.~\ref{thm-prog-qkctl*-s})\hss}
  \tabularnewline
\cline{1-1}\cline{5-5}
\allQCTL &  & 
     \multicolumn{2}{c|}{}
   & \PH-hard, in~\PSPACE 
  \tabularnewline
\allQCTL* & (Th.~\ref{thm-undec}) &
  \multicolumn{2}{c|}{(Th.~\ref{thm-eqctl-s},~\ref{thm-form-eqctl-s})} &
  (Th.~\ref{thm-prog-eqctl-s})
  \tabularnewline
\hline
\end{tabular}}
\centering
  \caption{Results for the structure semantics}
  \label{tab:results-structure}

\medskip
\medskip





  \centerline{%
  \def\arraystretch{1.1}
\begin{tabular}{|c||W|X|X|Z|}
\cline{2-5}
\multicolumn{1}{c||}{} & satisfiability &
  model checking & formula-compl. & program-compl. \tabularnewline
\hline\hline
\EQnCTL k &\mbox{\EXPTIME[k]-c.(Th.~\ref{thm-qctl-sat-tree-k})}  &
  \multicolumn{2}{c|}{\EXPTIME[k]-complete} &
   \tabularnewline
\cline{1-2}
\QnCTL k &\mbox{\EXPTIME[(k+1)]-c.(Th.~\ref{thm-qctl-sat-tree-k})} &
   \multicolumn{2}{c|}{(Cor.~\ref{coro-qkctl-t},Th.~\ref{thm-form-all-k-t})} & \tabularnewline
\cline{1-4}
\EQnCTL k* &\mbox{\EXPTIME[(k+1)]-c.(Th.~\ref{thm-qctls-sat-tree-k})} &
   \multicolumn{2}{c|}{\EXPTIME[(k+1)]-complete} & \PTIME-complete
   \tabularnewline
\cline{1-2}
\QnCTL k*  &\mbox{\EXPTIME[(k+2)]-c.(Th.~\ref{thm-qctls-sat-tree-k})} &
  \multicolumn{2}{c|}{(Cor.~\ref{coro-qkctl*-t}, Th.~\ref{thm-form-all-k-t})}
  &   (Th.~\ref{thm-prog-all-t}) 
  \tabularnewline
\cline{1-4}
\allQCTL &
  \multicolumn{3}{c|}{\TOWER-complete} &
  \tabularnewline
\allQCTL* &
  \multicolumn{3}{c|}{(Th.~\ref{thm-eqctl-t},~\ref{thm-form-EQCTL-t}, Cor.~\ref{cor-qctl-sat-tree})} &
  \tabularnewline
\hline
\end{tabular}}
  \caption{Results for the tree semantics.}
  \label{tab:results-tree}
\end{table}


\footnotetext{Hardness in \DDPlogn{k+1} for the formula-complexity of \QnCTL k
  model checking only holds when
  considering the DAG-size of formulas. With the standard size, the problem is
  \BH(\SSP k)-hard, and in \DDPlogn{k+1}.}

While it was introduced thirty years ago, the extension of \CTL with
propositional quantifiers had never been studied in details. Motivated by a
correspondence with temporal logics for multi-player games, we~have proposed
an in-depth study of that logic, in terms of its expressiveness (most notably,
we proved that propositional quantification fills in the gap between temporal
logics and monadic second-order logic), of its model-checking problem (which
we proved is decidable---see a summary of our results in Tables~\ref{tab:results-structure}
and~\ref{tab:results-tree}) and of its satisfiability (which is partly decidable).

Temporal logics extended with propositional quantification are simple, yet
very powerful extensions of classical temporal logics. They have a natural
semantics, and optimal algorithms based on standard automata constructions. 
Their powerful expressiveness is very convenient to encode more intricate
problems over extensions of temporal logics, as was done for multi-agent
systems in~\cite{rr-ATLsc}. We~expect that \QCTL will find more applications
in related areas shortly.

\paragraph*{Acknowledgement.} We thank Thomas~Colcombet, Olivier~Serre 
  and Sylvain~Schmitz for helpful comments during the redaction of
  this paper. We~also thank the reviewers for their many valuable
  remarks.

\bibliographystyle{myalpha}
\bibliography{bibexport}

\newcommand{\etalchar}[1]{$^{#1}$}
\begin{thebibliography}{KMTV00}

\bibitem[AFF{\etalchar{+}}03]{AFFGPTV03}
R.~Armoni, L.~Fix, A.~Flaisher, O.~Grumberg, N.~Piterman, A.~Tiemeyer, and
  M.~Y. Vardi.
\newblock Enhanced vacuity detection in linear temporal logic.
\newblock In {\em {P}roceedings of the 15th {I}nternational {C}onference on
  {C}omputer {A}ided {V}erification ({CAV}'03)}, LNCS 2725, p.~368--380.
  Springer, 2003.

\bibitem[AHK97]{AHK97}
R.~Alur, T.~A. Henzinger, and O.~Kupferman.
\newblock Alternating-time temporal logic.
\newblock In {\em {P}roceedings of the 38th {A}nnual {S}ymposium on
  {F}oundations of {C}omputer {S}cience ({FOCS}'97)}, p.~100--109. IEEE Comp.
  Soc. Press, 1997.

\bibitem[AP06]{AP06}
H.~Arl{\'o}{-}Costa and E.~Pacuit.
\newblock First-order classical modal logic.
\newblock {\em Studia Logica}, 84(2):171--210, 2006.

\bibitem[CDC04]{CDC04}
K.~Chatterjee, P.~Dasgupta, and P.~P. Chakrabarti.
\newblock The power of first-order quantification over states in branching and
  linear time temporal logics.
\newblock {\em Information Processing Letters}, 91(5):201--210, 2004.

\bibitem[CE82]{CE82}
E.~M. Clarke and E.~A. Emerson.
\newblock Design and synthesis of synchronization skeletons using
  branching-time temporal logic.
\newblock In {\em {P}roceedings of the 3rd {W}orkshop on {L}ogics of {P}rograms
  ({LOP}'81)}, LNCS 131, p.~52--71. Springer, 1982.

\bibitem[CE11]{CE11}
B.~Courcelle and J.~Engelfriet.
\newblock {\em Graph Structure and Monadic Second-Order Logic, a Language
  Theoretic Approach}.
\newblock Cambridge University Press, 2011.

\bibitem[CES86]{CES86}
E.~M. Clarke, E.~A. Emerson, and A.~P. Sistla.
\newblock Automatic verification of finite-state concurrent systems using
  temporal logic specifications.
\newblock {\em ACM Transactions on Programming Languages and Systems},
  8(2):244--263, 1986.

\bibitem[CHP07]{CHP07b}
K.~Chatterjee, T.~A. Henzinger, and N.~Piterman.
\newblock Strategy logic.
\newblock In {\em {P}roceedings of the 18th {I}nternational {C}onference on
  {C}oncurrency {T}heory ({CONCUR}'07)}, LNCS 4703, p.~59--73. Springer, 2007.

\bibitem[Dam94]{Dam94b}
D.~R. Dams.
\newblock {CTL{\textsuperscript{*}}} and {ECTL{\textsuperscript{*}}} as
  fragments of the modal {{$\mu$}}-calculus.
\newblock {\em Theoretical Computer Science}, 126(1):77--96, 1994.

\bibitem[DLM10]{DLM10}
A.~Da{~}Costa, F.~Laroussinie, and N.~Markey.
\newblock {ATL} with strategy contexts: Expressiveness and model checking.
\newblock In {\em {P}roceedings of the 30th {C}onferentce on {F}oundations of
  {S}oftware {T}echnology and {T}heoretical {C}omputer {S}cience
  ({FSTTCS}'10)}, LIPIcs~8, p.~120--132. Leibniz-Zentrum f{\"u}r Informatik,
  2010.

\bibitem[DLM12]{DLM12}
A.~Da{~}Costa, F.~Laroussinie, and N.~Markey.
\newblock Quantified~{CTL}: Expressiveness and model checking.
\newblock In {\em {P}roceedings of the 23rd {I}nternational {C}onference on
  {C}oncurrency {T}heory ({CONCUR}'12)}, LNCS 7454, p.~177--192. Springer,
  2012.

\bibitem[EF95]{EF95}
H.-D. Ebbinghaus and J.~Flum.
\newblock {\em Finite Model Theory}.
\newblock Springer, 1995.

\bibitem[EH86]{EH86}
E.~A. Emerson and J.~Y. Halpern.
\newblock {"}{S}ometimes{"} and {"}not never{"} revisited: On~branching versus
  linear time temporal logic.
\newblock {\em Journal of the ACM}, 33(1):151--178, 1986.

\bibitem[ES84]{ES84a}
E.~A. Emerson and A.~P. Sistla.
\newblock Deciding full branching time logic.
\newblock {\em Information and Control}, 61(3):175--201, 1984.

\bibitem[Ete99]{Ete99}
K.~Etessami.
\newblock Stutter-invariant languages, {$\omega$}-automata, and temporal
  logic.
\newblock In {\em {P}roceedings of the 11th {I}nternational {C}onference on
  {C}omputer {A}ided {V}erification ({CAV}'99)}, LNCS 1633, p.~236--248.
  Springer, 1999.

\bibitem[Fin70]{Fin70}
K.~Fine.
\newblock Propositional quantifiers in modal logic.
\newblock {\em Theoria}, 36(3):336--346, 1970.

\bibitem[FM98]{FM98}
M.~Fitting and R.~L. Mendelsohn.
\newblock {\em First-Order Modal Logic}, Synthese Library.
\newblock Number 277 in Synthese Library. Springer, 1998.

\bibitem[FR03]{FR03}
T.~French and M.~Reynolds.
\newblock A~sound and complete proof system for~{QPTL}.
\newblock In {\em {P}roceedings of the 4th {W}orkshop on {A}dvances in {M}odal
  {L}ogic ({AIML}'02)}, p.~127--148. King's College Publications, 2003.

\bibitem[Fre01]{Fre01}
T.~French.
\newblock Decidability of quantified propositional branching time logics.
\newblock In {\em {P}roceedings of the 14th {A}ustralian {J}oint {C}onference
  on {A}rtificial {I}ntelligence ({AJCAI}'01)}, LNCS 2256, p.~165--176.
  Springer, 2001.

\bibitem[Fre03]{Fre03}
T.~French.
\newblock Quantified propositional temporal logic with repeating states.
\newblock In {\em {P}roceedings of the 10th {I}nternational {S}ymposium on
  {T}emporal {R}epresentation and {R}easoning and of the 4th {I}nternational
  {C}onference on {T}emporal {L}ogic ({TIME}-{ICTL}'03)}, p.~155--165. IEEE
  Comp. Soc. Press, 2003.

\bibitem[GC04]{GC04}
A.~Gurfinkel and M.~Chechik.
\newblock Extending extended vacuity.
\newblock In {\em {P}roceedings of the 5th {I}nternational {C}onference on
  {F}ormal {M}ethods in {C}omputer-{A}ided {D}esign ({FMCAD}'04)}, LNCS 3312,
  p.~306--321. Springer, 2004.

\bibitem[GC12]{GC12}
A.~Gurfinkel and M.~Chechik.
\newblock Robust vacuity for branching temporal logic.
\newblock {\em ACM Transactions on Computational Logic}, 13(1), 2012.

\bibitem[GJ79]{GJ79}
M.~R. Garey and D.~S. Johnson.
\newblock {\em Computers and Intractability: A Guide to the Theory of
  {NP}-Completeness}.
\newblock W. H. Freeman \& Co., 1979.

\bibitem[Got95]{Got95}
G.~Gottlob.
\newblock {NP} trees and {C}arnap's modal logic.
\newblock {\em Journal of the ACM}, 42(2):421--457, 1995.

\bibitem[Hem98]{Hem98}
H.~Hempel.
\newblock {\em Boolean Hierarchies~-- On~Collapse Properties and Query Order}.
\newblock PhD thesis, Friedrich-Schiller Universit{\"a}t Jena, Germany, 1998.

\bibitem[HK94]{HK94}
J.~Y. Halpern and B.~M. Kapron.
\newblock Zero-one laws for modal logic.
\newblock {\em Annals of Pure and Applied Logic}, 69(2-3):157--193, 1994.

\bibitem[HRS98]{HRS98}
T.~A. Henzinger, J.-F. Raskin, and P.-Y. Schobbens.
\newblock The regular real-time languages.
\newblock In {\em {P}roceedings of the 25th {I}nternational {C}olloquium on
  {A}utomata, {L}anguages and {P}rogramming ({ICALP}'98)}, LNCS 1443,
  p.~580--591. Springer, 1998.

\bibitem[HSW13]{HSW13}
C.-H. Huang, S.~Schewe, and F.~Wang.
\newblock Model-checking iterated games.
\newblock In {\em {P}roceedings of the 19th {I}nternational {C}onference on
  {T}ools and {A}lgorithms for {C}onstruction and {A}nalysis of {S}ystems
  ({TACAS}'13)}, LNCS 7795, p.~154--168. Springer, 2013.

\bibitem[Kai97]{Kai97}
R.~Kaivola.
\newblock {\em Using Automata to Characterise Fixed Point Temporal Logics}.
\newblock Phd thesis, School of Informatics, University of Edinburgh, UK, 1997.

\bibitem[KMTV00]{KMTV00}
O.~Kupferman, P.~Madhusudan, P.~S. Thiagarajan, and M.~Y. Vardi.
\newblock Open systems in reactive environments: Control and synthesis.
\newblock In {\em {P}roceedings of the 11th {I}nternational {C}onference on
  {C}oncurrency {T}heory ({CONCUR}'00)}, LNCS 1877, p.~92--107. Springer, 2000.

\bibitem[Koz83]{Koz83}
D.~C. Kozen.
\newblock Results on the propositional {$\mu$}-calculus.
\newblock {\em Theoretical Computer Science}, 27:333--354, 1983.

\bibitem[KP95]{KP95b}
O.~Kupferman and A.~Pnueli.
\newblock {\emph{Once}} and {\emph{for all}}.
\newblock In {\em {P}roceedings of the 10th {A}nnual {S}ymposium on {L}ogic in
  {C}omputer {S}cience ({LICS}'95)}, p.~25--35. IEEE Comp. Soc. Press, 1995.

\bibitem[KP02]{KP02a}
Y.~Kesten and A.~Pnueli.
\newblock Complete proof system for~{QPTL}.
\newblock {\em Journal of Logic and Computation}, 12(5):701--745, 2002.

\bibitem[Kri59]{Kri59}
S.~A. Kripke.
\newblock A completeness theorem in modal logic.
\newblock {\em Journal of Symbolic Logic}, 24(1):1--14, 1959.

\bibitem[Kup95]{Kup95a}
O.~Kupferman.
\newblock Augmenting branching temporal logics with existential quantification
  over atomic propositions.
\newblock In {\em {P}roceedings of the 7th {I}nternational {C}onference on
  {C}omputer {A}ided {V}erification ({CAV}'95)}, LNCS 939, p.~325--338.
  Springer, 1995.

\bibitem[KVW00]{KVW00}
O.~Kupferman, M.~Y. Vardi, and P.~Wolper.
\newblock An automata-theoretic approach to branching-time model-checking.
\newblock {\em Journal of the ACM}, 47(2):312--360, 2000.

\bibitem[LM14]{rr-ATLsc}
F.~Laroussinie and N.~Markey.
\newblock \href {http://www.lsv.ens-cachan.fr/Publis/RAPPORTS_LSV/PDF/
  rr-lsv-2014-05.pdf} {Augmenting {ATL} with strategy contexts}.
\newblock Research Report LSV-14-05, Laboratoire Sp{\'e}cification et
  V{\'e}rification, ENS Cachan, France, 2014.
\newblock 45~pages.

\bibitem[LMS01]{LMS01}
F.~Laroussinie, N.~Markey, and {\relax Ph}.~Schnoebelen.
\newblock Model checking {CTL}{$^+$} and {FCTL} is hard.
\newblock In {\em {P}roceedings of the 4th {I}nternational {C}onference on
  {F}oundations of {S}oftware {S}cience and {C}omputation {S}tructure
  ({FoSSaCS}'01)}, LNCS 2030, p.~318--331. Springer, 2001.

\bibitem[L{\"o}d13]{Loeding2012}
C.~L{\"o}ding.
\newblock {A}utomata on {I}nfinite {T}rees (preliminary version for the
  handbook of the {A}uto{M}ath{A} project), 2013.

\bibitem[Mar10]{Mar10}
M.~B. Martins.
\newblock {\em Supervisory Control of {P}etri Nets using Linear Temporal
  Logic}.
\newblock Th\`ese de doctorat, Instituto Superior T\'ecnico, Universidade
  T\'ecnica de Lisboa, Portugal, 2010.

\bibitem[MMV10]{MMV10a}
F.~Mogavero, A.~Murano, and M.~Y. Vardi.
\newblock Reasoning about strategies.
\newblock In {\em {P}roceedings of the 30th {C}onferentce on {F}oundations of
  {S}oftware {T}echnology and {T}heoretical {C}omputer {S}cience
  ({FSTTCS}'10)}, LIPIcs~8, p.~133--144. Leibniz-Zentrum f{\"u}r Informatik,
  2010.

\bibitem[MR03]{MR03}
F.~Moller and A.~Rabinovich.
\newblock Counting on {CTL}{\textsuperscript{*}}: on~the expressive power of
  monadic path logic.
\newblock {\em Information and Computation}, 184(1):147--159, 2003.

\bibitem[MS85]{MS85}
D.~E. Muller and P.~E. Schupp.
\newblock Alternating automata on infinite objects, determinacy and {R}abin's
  theorem.
\newblock In {\em {A}utomata on {I}nfinite {W}ords~-- {{\'E}}cole de
  {P}rintemps d'{I}nformatique {T}h{\'e}orique ({EPIT}'84)}, LNCS 192,
  p.~99--107. Springer, 1985.

\bibitem[MS87]{MS87}
D.~E. Muller and P.~E. Schupp.
\newblock Alternating automata on infinite trees.
\newblock {\em Theoretical Computer Science}, 54(2-3):267--276, 1987.

\bibitem[MS95]{MS95}
D.~E. Muller and P.~E. Schupp.
\newblock Simulating alternating tree automata by nondeterministic automata:
  New results and new proofs of the theorems of {R}abin, {M}c{N}aughton and
  {S}afra.
\newblock {\em Theoretical Computer Science}, 141(1-2):69--107, 1995.

\bibitem[Pap94]{Pap94}
{\relax Ch}.~H. Papadimitriou.
\newblock {\em Computational Complexity}.
\newblock Addison-Wesley, 1994.

\bibitem[PBD{\etalchar{+}}02]{PBDDC02}
A.~C. Patthak, I.~Bhattacharya, A.~Dasgupta, P.~Dasgupta, and P.~P.
  Chakrabarti.
\newblock Quantified computation tree logic.
\newblock {\em Information Processing Letters}, 82(3):123--129, 2002.

\bibitem[Pin07]{Pin07a}
S.~Pinchinat.
\newblock A generic constructive solution for concurrent games with expressive
  constraints on strategies.
\newblock In {\em {P}roceedings of the 5th {I}nternational {S}ymposium on
  {A}utomated {T}echnology for {V}erification and {A}nalysis ({ATVA}'07)}, LNCS
  4762, p.~253--267. Springer, 2007.

\bibitem[Pnu77]{Pnu77}
A.~Pnueli.
\newblock The temporal logic of programs.
\newblock In {\em {P}roceedings of the 18th {A}nnual {S}ymposium on
  {F}oundations of {C}omputer {S}cience ({FOCS}'77)}, p.~46--57. IEEE Comp.
  Soc. Press, 1977.

\bibitem[QS82]{QS82a}
J.-P. Queille and J.~Sifakis.
\newblock Specification and verification of concurrent systems in {CESAR}.
\newblock In {\em {P}roceedings of the 5th {I}nternational {S}ymposium on
  {P}rogramming ({SOP}'82)}, LNCS 137, p.~337--351. Springer, 1982.

\bibitem[Rab72]{Rab72}
M.~O. Rabin.
\newblock {\em Automata on infinite objects and {C}hurch's thesis}, Regional
  Conference Series in Mathematics.
\newblock Number~13 in Regional Conference Series in Mathematics. American
  Mathematical Society, 1972.

\bibitem[RP03]{RP03a}
S.~Riedweg and S.~Pinchinat.
\newblock Quantified {{$\mu$}}-calculus for control synthesis.
\newblock In {\em {P}roceedings of the 28th {I}nternational {S}ymposium on
  {M}athematical {F}oundations of {C}omputer {S}cience ({MFCS}'03)}, LNCS 2747,
  p.~642--651. Springer, 2003.

\bibitem[SC85]{SC85}
A.~P. Sistla and E.~M. Clarke.
\newblock The complexity of propositional linear temporal logics.
\newblock {\em Journal of the ACM}, 32(3):733--749, 1985.

\bibitem[Sch03]{Sch03a}
{\relax Ph}.~Schnoebelen.
\newblock The complexity of temporal logic model checking.
\newblock In {\em {P}roceedings of the 4th {W}orkshop on {A}dvances in {M}odal
  {L}ogic ({AIML}'02)}, p.~481--517. King's College Publications, 2003.

\bibitem[Sch13]{Sch13}
S.~Schmitz.
\newblock Complexity hierarchies beyond elementary.
\newblock Research Report cs.CC/1312.5686, arXiv, 2013.

\bibitem[See76]{See76}
D.~G. Seese.
\newblock {\em Entscheidbarkeits- und {I}nterpretierbarkeitsfragen monadischer
  {T}heorien zweiter {S}tufe gewisser {K}lassen von {G}raphen}.
\newblock PhD thesis, Humboldt-Universität zu Berlin, German Democratic
  Republic, 1976.

\bibitem[Sis83]{Sis83}
A.~P. Sistla.
\newblock {\em Theoretical Issues in the Design and Verification of Distributed
  Systems}.
\newblock PhD thesis, Harvard University, Cambridge, Massachussets, USA, 1983.

\bibitem[Sto76]{Sto76}
L.~J. Stockmeyer.
\newblock The polynomial-time hierarchy.
\newblock {\em Theoretical Computer Science}, 3(1):1--22, 1976.

\bibitem[SVW87]{SVW87}
A.~P. Sistla, M.~Y. Vardi, and P.~Wolper.
\newblock The complementation problem for {B}{\"u}chi automata with
  applications to temporal logics.
\newblock {\em Theoretical Computer Science}, 49:217--237, 1987.

\bibitem[tC06]{tC06}
B.~ten Cate.
\newblock Expressivity of second order propositional modal logic.
\newblock {\em Journal of Philosophical Logic}, 35(2):209--223, 2006.

\bibitem[Tho97]{Tho97b}
W.~Thomas.
\newblock Languages, automata and logics.
\newblock In {\em Handbook of Formal Languages}, p.~389--455. Springer, 1997.

\bibitem[Var82]{Var82}
M.~Y. Vardi.
\newblock The complexity of relational query languages.
\newblock In {\em {P}roceedings of the 14th {A}nnual {ACM} {S}ymposium on the
  {T}heory of {C}omputing ({STOC}'82)}, p.~137--146. ACM Press, 1982.

\bibitem[Wag90]{Wag90}
K.~W. Wagner.
\newblock Bounded query classes.
\newblock {\em SIAM Journal on Computing}, 19(5):833--846, 1990.

\bibitem[WHY11]{WHY11}
F.~Wang, C.-H. Huang, and F.~Yu.
\newblock A~temporal logic for the interaction of strategies.
\newblock In {\em {P}roceedings of the 22nd {I}nternational {C}onference on
  {C}oncurrency {T}heory ({CONCUR}'11)}, LNCS 6901, p.~466--481. Springer,
  2011.

\end{thebibliography}

\clearpage
\appendix
\section{The polynomial-time and exponential hierarchies}\label{app-complex}

In this section, we briefly define the complexity classes that are used in the
paper. More details about those complexity classes can be found \eg
in~\cite{Pap94}. 

\medskip

We~write \PTIME (resp.~\NP) for the class of decision problems that can be decided in
polynomial time by a deterministic (resp.~non-deterministic) Turing machine. 

Given a decision problem~$\calL$ (for instance~SAT), a~\newdef{Turing machine
  with oracle~$\calL$} is a Turing machine equipped with an extra tape (the
\emph{oracle tape}) and three special states~$q_{\text{oracle}}$, $q_{\text{yes}}$
and~$q_{\text{no}}$. During the computation, whenever the Turing machine
enters state~$q_{\text{oracle}}$, it~directly jumps to~$q_{\text{yes}}$ if the 
instance of~$\calL$ written on the oracle tape is positive, and to~$q_{\text{no}}$
otherwise.

We~write $\PTIME^{\calL}$ for the set of languages accepted by deterministic
Turing machines with oracle~$\calL$ and halting in polynomial time (in~the
size of the input). We~write $\NP^{\calL}$ for the analogous class defined
with non-deterministic machines. When $\calL$ is complete for some complexity
class~$\calC$, we~also write $\PTIME^{\calC}$ for $\PTIME^{\calL}$ (and
$\NP^{\calC}$ for $\NP^{\calL}$). Notice that these definitions do not depend
on the selected $\calC$-complete problem.
We~also define $\PTIME^{\calC[\log n]}$ as the subclass of problems of
$\PTIME^{\calC}$ that can be solved by a deterministic
Turing machine in polynomial time but with only a logarithmic number of visits
to the $q_{\text{oracle}}$-state. Finally, $\PTIME^{\calC}_{\|}$ is the
subclass of problems of $\PTIME^{\calC}$ that can be solved in polynomial time
by a deterministic Turing machine slightly different from the previous ones:
several queries can be written on the oracle tape, but once the oracle state
is visited, the oracle tape cannot be modified anymore. This can be seen as
solving all queries in parallel.

The \newdef{polynomial-time hierarchy} is the following sequence of complexity
classes, defined recursively with $\SSP0=\PTIME$ and
\begin{xalignat*}3
\SSP{i+1} &= \NP^{\SSP i} & 
\PPP{i+1} &= \co\NP^{\SSP i} &
\DDP{i+1} &= \PTIME^{\SSP i} \\
\DDPlogn{i+1} &= \PTIME^{\SSP i[\log n]}&  
\DDPpar{i+1} &= \PTIME^{\SSP i[\log n]}_{\|}
\end{xalignat*}
The classes in this hierarchy lie between \PTIME and \PSPACE (and it is an
open problem whether the hierarchy collapses). Notice for instance that
$\SSP1=\NP$, since an oracle for \PTIME problems is useless to a
non-deterministic Turing machine running in polynomial time. Similarly,
$\DDP1=\PTIME$. On the other hand, $\DDP2=\PTIME^{\NP}$ contains those
problems that can be solved in polynomial time by a deterministic Turing
machine with an \NP oracle. This includes both \NP and \co\NP.

\smallskip
A~natural problem in \SSP i is
the problem $\textsf{QSAT}_i$ of finding the truth value of
\[
\exists X_1.\ \forall X_2.\ \exists X_3 \ldots Q_i X_i.\ 
  \phi(X_1,X_2,X_3,\ldots,X_i)
\]
where $\phi$ is a boolean formula with variables in~$X_1$ to~$X_i$. The
dual problem (where the sequence of quantifications begins with a
universal~one) is \PPP i-complete. The problem   
$\textsf{SNSAT}_i$, made of several instances of $\textsf{QSAT}_i$
where the $k$-th instance uses the truth value of the $k-1$ previous
ones, is an example of a \DDP{i+1}-complete problem~\cite{LMS01}.

\smallskip
The complexity class \PH is the union of all the classes in the
polynomial-time hierarchy. Equivalently, $\PH = \bigcup_{i\in\bbN} \SSP i$.
Clearly, $\PH\subseteq \PSPACE$. Whether these two classes are equal is open.
Notice that \PH is not known to contain complete problems: if a problem is
complete for \PH, then it~is at least as hard as any other problem in \PH, but
on the other hand is belongs to~$\SSP i$ for some~$i$, which implies that the
polynomial-time hierarchy would collapse.

\medskip

The class \EXPTIME[k] (resp.~\EXPSPACE[k]) is the class of decision problems
that can be solved by a deterministic Turing machine running in time
(resp.~space) $O(\exp_k(p(n)))$ for some polynomial~$p$, where $\exp_k$ is defined
inductively as follows:
\begin{xalignat*}2
\exp_1(n)  & = 2^n &
  \exp_k(n) &= 2^{\exp_{k-1}(n)}
\end{xalignat*}
For instance, $\exp_5(n) = 2^{2^{2^{2^{2^{n}}}}}$ (and $\exp_5(1)$ is a number
with $65.536$ binary digits). It~is not difficult to prove that 
\[
\EXPTIME[k]\subseteq \EXPSPACE[k] \subseteq \EXPTIME[(k+1)],
\]
and one of the two inclusions is strict (\ie, $\EXPTIME[k]\subsetneq
\EXPTIME[(k+1)]$). 

In the same way as for the polynomial-time hierarchy, we~let \ELEM be the
union of all the classes in the exponential hierarchy: $\ELEM = \bigcup_{k\in
  \bbN} \EXPTIME[k]$. Since this hierarchy is known to be strict, \ELEM~cannot
have complete problems. 

\smallskip
In order to define classes above \ELEM, 
we~define the function $\textsf{tower}(n) = \exp_n(1)$. The class \TOWER is
then the class of problems that can be decided by a deterministic Turing
machine in time~$O(\textsf{tower}(p(n)))$ where $p$ is a polynomial. It~is
quite clear that $\ELEM\subseteq \TOWER$ (because $\exp_k(p(n)) \leq
\textsf{tower}(k+p(n))$). But now, as~argued in~\cite{Sch13}, \TOWER~has
complete problems, and \TOWER-hardness can be proved by showing
\EXPTIME[k]-hardness for all~$k$ using \emph{uniform} reductions.

\section{Characterising grid-like structures with \EQnCTL2}\label{app-grid}

Let~$\calS=\tuple{Q,R,\ell}$ be an arbitrary (finite-state) Kripke structure.  Then
$\calS$ is a grid (in a sense that will be made precise at
Prop.~\ref{prop-grid}) if it can be labelled with atomic
propositions $s$, $h$, $v$, $l$, $r$, $t$ and~$b$ in such a way that
the following conditions are fulfilled:
\begin{itemize}
\item all states have at least one successor, and at most two: 
\begin{multline}
\All\G\Ex\X\true\ \et\ \forall \alpha,\beta.\ \All\G\biggl[\bigl(\Ex\X(\alpha \et \beta)
  \et \Ex\X(\alpha \et \non \beta)\bigr) \thn \\
\biggl((\All\X \alpha)  \et 
\ET_{\substack{h_s\in\{h,\non h\}\\v_s\in\{v,\non v\}}} 
  \Bigl[(h_s \et v_s) \thn (\Ex\X(\non h_s \et v_s) \et \Ex\X(h_s\et \non v_s))\Bigr]
\biggr)\biggr]
\label{eq-2suc}
\end{multline}
  Notice that the formula also requires that when a state has two
  successors, then they are labelled differently w.r.t. both~$h$
  and~$v$.
\item the Kripke structure has exactly one self-loop, which has only
  one outgoing transition (the loop itself). It~is the role of atomic
  proposition~$s$ to mark that state:
\begin{equation}
\uniq(s)  \et  \All\G(s \thn \All\G s) \et \All\F s.
\end{equation}
The first two conjuncts impose that the state labelled with~$s$ has only a
self-loop as outgoing transition. The third conjunct requires that all paths
eventually reach~$s$, which means that the structure is acyclic (except at~$s$).

\item the atomic propositions~$h$ and~$v$ (for \emph{horizontal}
  and~\emph{vertical}) are used to define the direction of the grid. This
  contains several formulas: first, we~impose that the initial state has two
  successors:
\begin{equation}
(h \et v) \et  \Ex\X(h \et \non v) \et \Ex\X(\non v \et h)
\label{eq-init}
\end{equation}

We then impose that if a state has two successors, then those two successors
have a common successor, as depicted on Fig.~\ref{fig-grid}. This is expressed
as follows:
\begin{multline}
\forall \gamma.\ 
  \All\G\biggl[
    \smash{\ET_{d\in\{h,\non h,v,\non v\}}}
    (d \et\Ex\X d \et \Ex\X\non d) \thn \\ 
    \left( \et 
      \begin{array}{c}
        \Ex\X(d \et \All\X \gamma) \thn \Ex\X(\non d \et \Ex\X(\non d \et \gamma)) \\[2mm]
        \Ex\X(\non d \et \All\X \gamma) \thn \Ex\X(d \et \Ex\X(\non d \et \gamma)) 
      \end{array} 
    \right)
  \biggr]
\label{eq-square1}
\end{multline}

In the same vein, we~also impose another formula, which we will use as a
sufficient condition for having two successors:
\begin{equation}
\forall \gamma.\ 
  \All\G\left[
    \ET_{d\in\{h,\non h,v,\non v\}}
    d \thn \left(
      \begin{array}{c}
        \Ex\X(d \et \Ex\X(\non d \et \gamma)) \\
        \iff \\
        \Ex\X(\non d \et\non s \et \Ex\X(\non d \et \gamma))
      \end{array}
    \right)
  \right]
\label{eq-square2}
\end{equation}

We also impose globally that each state can be reached from the
initial state by two particular paths that are made of one
``horizontal part'' followed by a ``vertical part'' (and
conversely). This is expressed as follows:
\begin{multline}
\forall \gamma.\ \Bigl[\Ex\F\gamma \thn 
  \smash{\OU_{\substack{h_q\in\{h,\non h\} \\v_q\in\{v,\non v\}}}\
  \OU_{\substack{h_\gamma\in\{h,\non h\} \\v_\gamma\in\{v,\non v\}}}}
  \Ex h_q \Until (h_q \et \Ex v_\gamma \Until (v_\gamma\et \gamma)) \et \\ 
  \Ex v_q \Until (v_q \et \Ex h_\gamma \Until (h_\gamma\et \gamma)) \Bigr]
\label{eq-hv}
\end{multline}
Symmetrically, from any state, the $s$-state can be reached using
similar paths. 
Here we~have to enumerate the possible values for~$h$ and~$v$ in the $s$-state:
\begin{multline}
\smash{\OU_{\substack{h_s\in\{h,\non h\} \\v_s\in\{v,\non v\}}}}
\All\G\Bigl[
\smash{\OU_{\substack{h_q\in\{h,\non h\} \\v_q\in\{v,\non v\}}}}
  \Ex h_q \Until (h_q \et \Ex v_s \Until (v_s \et s)) \et  \\
  \Ex v_q \Until (v_q \et \Ex h_s \Until (h_s \et s))\Bigr]
\label{eq-HV}
\end{multline}

\item finally, we~label the left, right, top and bottom borders of the grid,
  which we define as follows:
\begin{multline}
 \All(v \et l)\Until(\non v \et \non l) \et \All\G(\non v \thn
  \All\G\non l)  
 \et  \All\G(r \iff (\All\G v \ou \All\G \non v)) \\
 \et\  \All(h \et t)\Until(\non h \et \non t) \et \All\G(\non h \thn
  \All\G\non t) 
 \et  \All\G(b \iff (\All\G h \ou \All\G \non h))
\label{eq-lrtb}
\end{multline}

\end{itemize}

\noindent Propositions $h$ and~$v$ are used to mark ``horizontal'' and
``vertical'' lines. A~successor~$u'$ of a state~$u$ is a
\newdef{horizontal successor} if both~$u$ and~$u'$ have the same
labelling w.r.t.~$h$. Similarly, it~is a \newdef{vertical successor}
when they have the same labelling w.r.t.~$v$.  Similarly,
a~\newdef{horizontal path} (resp.~\newdef{vertical path}) is a path
in~$\calS$ along which the truth value of~$h$ (resp.~of~$v$) is
constant.  Finally, we~let $L$, $R$, $T$ and~$B$ be the sets of states
labelled with~$l$, $r$, $t$ and~$b$, respectively.

We now give a few intermediary results that will be convenient for proving
that the conjunction of the formulas above characterises grids. These lemmas
assume that the initial state~$q$ of~$\calS$ satisfies the conjunction of
Formulas~\eqref{eq-2suc} to~\eqref{eq-lrtb}, and that all states
of~$\calS$ are reachable from~$q$.

\begin{lemma}\label{lemma-alt}
  Pick two states~$u$ and~$u'$ in~$\calS$, such that there is a transition
  between~$u$ and~$u'$ (in~either direction). Then
  \begin{itemize}
  \item either $u$ and~$u'$ are the same state (hence they are
    labelled with~$s$),
  \item or they have the same labelling with~$h$ if, and only if, they
    have different labelling with~$v$.
  \end{itemize}
\end{lemma}

In~other terms, a~successor of a state is either a horizontal successor or a
vertical successor, and not both (except for the state carrying the
self-loop).

\begin{proof}
  We proceed by contradiction: pick a transition~$(u,u')$ (not the
  self-loop) such that $u$ and~$u'$ are both labelled with~$h$ and~$v$
  (the other cases would be similar).
  We~assume that~$u$ is (one~of) the minimal
  such~$u$, with ``minimal'' here being defined w.r.t. the sum of the
  length of all paths from the initial state~$q$ of~$\calS$ to~$u$.

  First notice that it~cannot be the case that~$u=q$, because of
  Formulas~\eqref{eq-2suc} and~\eqref{eq-init}. Hence $u$~must have a
  predecessor~$w$. By~minimality of~$u$, that state satisfies the condition in
  the lemma. We~assume that $w$ is labelled with~$h$ and~$\non v$ (the~other
  case, with $\non h$ and~$v$, would be similar). From
  Formula~\eqref{eq-square2} (with~$d=\non v$), there must be a successor~$w'$
  of~$w$ labelled with~$\non v$, and having $u'$ as successor (this is the
  common successor with~$u$). According to Formula~\eqref{eq-2suc}, $w'$~is
  labelled with~$\non h$. Now, applying Formula~\eqref{eq-square1} in~$w$ for
  proposition~$h$, there must be another common successor~$u''$ to~$u$
  and~$w'$, labelled with~$\non h$. We~get a contradiction, since $u$~now has
  two successors but does not satisfies Formula~\eqref{eq-2suc}.
\end{proof}

\begin{lemma}\label{lemma-predlrtb}
  If a state is in~$L$ and not in~$T$ (or~in~$T$ and not in~$L$) then it has
  exactly one predecessor. Only the initial state is both in~$L$ and~$T$.
  Symmetrically, any state in~$R$ or in~$B$ has only one successor, and the
  only state in~$R \cap B$ is labelled with~$s$.
\end{lemma}

\begin{proof}
According to Formula~\eqref{eq-lrtb}, the initial state must be
labelled with~$l$ and~$t$. Moreover, the initial state has two
successors, labelled with~$(h,\non v)$ and~$(\non h, v)$, according to
Formula~\eqref{eq-init}. Formula~\eqref{eq-lrtb} enforces that the
former state satisfies~$\non l$, and the latter satisfies~$\non
t$. Hence no other state will ever satisfy~$l \et t$.

Now, consider a state~$u$ labelled with~$l$ and not with~$t$, and
assume it~has two predecessors~$w$ and~$w'$.  From
Formula~\eqref{eq-lrtb}, any path between the initial state~$q$ and
state~$u$ can only visit $v$-states. So there must be some state
between~$q$ and~$u$ having two $v$-successors, which is forbidden by
Formula~\eqref{eq-2suc}. The proof for~$t$ is similar.

\smallskip
Now, assume that some state~$u$ in~$R$ has two successors. From
Formula~\eqref{eq-lrtb}, both successors will be labelled with~$v$ or both
with~$\non v$, which again contradicts Formula~\eqref{eq-2suc}.

Finally, if a state~$u$ is labelled with both~$r$ and~$b$, then
Formula~\eqref{eq-lrtb} imposes that all its successors must have the
same labelling as~$u$ w.r.t.~$h$ and~$v$. From Lemma~\ref{lemma-alt},
$u$~is labelled with~$s$.
\end{proof}

\begin{lemma}\label{lemma-1hpred}
Pick a state~$u$, different from the initial state. 
Then $u$~is in~$L$ (resp.~in~$T$) if, and only~if, it~has no horizontal (resp.
vertical) predecessor.

Similarly, pick a state~$u$ not labelled with~$s$. Then $u\in R$ (resp. $u\in
B$) if, and only~if, it~has no horizontal (resp. vertical) successor.
\end{lemma}

\begin{proof}
  We begin with proving the equivalence for the horizontal case. The vertical
  case is similar. We~consider the four possible labellings of~$u$ w.r.t.~$h$
  and~$v$:
\begin{itemize}
\item if~$u\models h \et v$: from Formula~\eqref{eq-hv}, there exists
  a path from~$q$ to~$u$ that is made of two parts: first a vertical path
  visiting only $v$-states, followed by a horizontal path visiting only
  $h$-states. In~case $u$ has no $h$-predecessor, it~must be the case
  that the second part is trivial, so that $u$~can be reached from~$q$
  by a path visiting only $v$-states. From Formula~\eqref{eq-lrtb},
  $u$~is labelled with~$l$. 

  Conversely, if $u$~is in~$L$, from
  Lemma~\ref{lemma-predlrtb} it has only one predecessor. That
  predecessor must be labelled with~$v$ (Formula~\eqref{eq-lrtb}),
  hence it cannot be labelled with~$h$ (Lemma~\ref{lemma-alt}).
  So $u$~has no $h$-predecessor.

\item if $u\models \non h \et v$: the same arguments apply, replacing $h$
  with~$\non h$. 

\item if $u\models h \et \non v$: again from Formula~\eqref{eq-hv},
  we~get the existence of a path from~$q$ to~$u$ visiting only
  $v$-states first, and only $h$-states in a second part. Now, because
  $u\models \non v$, the second part must contain at least two states,
  so that $u$ has a predecessor labelled with~$h$. Moreover, as $u$ is
  labelled with~$\non v$, it~cannot be in~$L$ (Formula~\eqref{eq-lrtb}).

\item if $u\models \non h \et \non v$, the same arguments apply.
\end{itemize}

\smallskip
 Now, the proof for~$R$ (and~$B$) is even simpler:
Formula~\eqref{eq-lrtb} precisely says that the states labelled with~$r$
(resp.~$b$) are precisely those that have only vertical (resp. horizontal)
successors. Apart for the $s$-state, this entails that those states do not
have horizontal (resp. vertical) successors.
\end{proof}

\begin{lemma}\label{lemma-orderlt}
Pick two states~$u$ and~$u'$ in~$L$ (resp.~in~$T$). Then there is a vertical
(resp. horizontal) path between~$u$ and~$u'$ (in~one or the other direction). 
\end{lemma}

\begin{proof}
From Formula~\eqref{eq-lrtb}, there exist vertical paths from the initial
state~$q$ to~$u$ and to~$u'$. Let~$\pi$ be the longest common prefix of these
two paths: it~contains at least~$q$. Consider its last state~$x$: if~$x$
is~$u$ or~$u'$, then our result follows. Otherwise, there are two distinct
vertical paths from~$x$ to~$u$ and~$u'$, which means that~$x$ has two vertical
successors, contradicting Formula~\eqref{eq-2suc}. The proof for~$T$ is similar.
\end{proof}

Following Lemma~\eqref{lemma-orderlt}, we~can define a binary
relation~$\orderL$ on~$L$ (resp.~$\orderT$ on~$T$) by letting~$u \orderL u'$
if, and only~if, there is a (vertical) path from~$u$ to~$u'$ (resp.~$u \orderT
u'$ if, and only~if, there is a horizontal path from~$u$ to~$u'$). These are
easily seen to be ordering relations, because $\calS$~is (mostly) acyclic.
Lemma~\ref{lemma-orderlt} entails that these orders are total, with~$q$ as
minimal element.

\medskip 

Now, pick a state~$u$. By~Formulas~\eqref{eq-hv} and~\eqref{eq-lrtb}, there
exist (at~least) one $l$-state~$x$ such that there is a horizontal path
from~$x$ to~$u$. Similarly, there is (at~least) one $t$-state~$y$ with a
vertical path from~$y$ to~$u$. Notice that from Formula~\eqref{eq-2suc},
we~know that there is only one (maximal) horizontal (resp. vertical) path
starting in any given state. We~now prove that the states~$x$ and~$y$ above
are uniquely determined from~$u$. 

\begin{lemma}\label{lemma-uniquexy}
  Let~$u$ be any state of~$\calS$. Let~$x$ and~$x'$ be $l$-states
  (resp.~$t$-states) such that $u$ is on the horizontal (resp. vertical) paths
  from~$x$ and from~$x'$. Then $x=x'$.
\end{lemma}

\begin{proof}
  We~prove the ``horizontal'' case, the other one being similar. So we~assume
  we~have two different states~$x$ and~$x'$, and w.l.o.g. that $x' \orderL x$.
  We show that there exists a ``grid'' containing~$q$, $x$
  and~$u$, \ie, a sequence of states $(u_{i,j})_{0\leq i\leq m, 0\leq j\leq
    n}$, having the following properties:
  \begin{enumerate}
  \item  for all~$0\leq i\leq m-1$ and~$0\leq j\leq
    n$, $(u_{i,j},u_{i+1,j})$ is a horizontal transition;
  \item for all~$0\leq i\leq m$ and~$0\leq j\leq n-1$, $(u_{i,j},u_{i,j+1})$
    is a vertical transition; 
  \item\label{item3} $u_{0,0}=q$, $u_{0,n}=x$ and $u_{m,n}=u$.
  \end{enumerate}
  The grid is built by repeatedly applying Formula~\eqref{eq-square2} as
  follows: first, there is a unique path from~$q$ to~$x$, which defines the
  values $(u_{0,j})_{0\leq j\leq n}$. Similarly, the unique path from~$x$
  to~$u$ defines the values $(u_{i,n})_{0\leq i\leq m}$. Notice that
  \eqref{item3} is fulfilled with this definition.

  Now we~apply Formula~\eqref{eq-square2} to~$u_{0,n-1}$, which has a vertical
  successor~$u_{0,n}$ followed by a horizontal successor~$u_{1,n}$. Hence
  there must exist a horizontal successor~$u_{1,n-1}$ of~$u_{0,n-1}$ of which
  $u_{1,n}$ is a vertical successor. The same argument applies to all states
  between~$u_{0,n-2}$ and~$u_{0,0}$, thus forming a vertical
  path~$(u_{1,j})_{0\leq j\leq n}$. The same argument applies again to form
  the subsequent vertical paths, until building path $(u_{m,j})_{0\leq j\leq
    n}$. 

  Since $x'\orderL x$, there must exist an integer~$k$ for which $u_{0,k}=x'$.
  By~construction, we~know that $u$~appears on the horizontal path
  originating from~$x'$. But it~cannot be the case that there is a~$p$ such
  that $u_{p,k}=u$: this would give rise to a cycle on~$u$. Hence there must
  be a horizontal path from~$u_{m,k}$ to~$u$.

  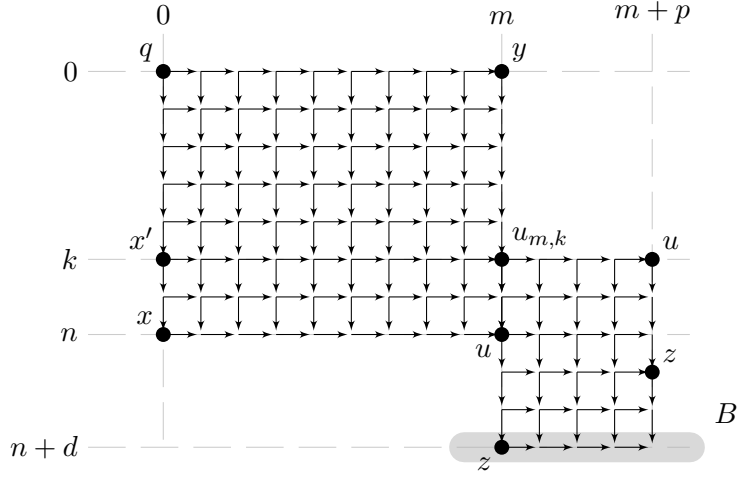
\begin{figure}[!ht]
    \centering
    \begin{tikzpicture}
      \foreach \x/\n in {0/0,2.5/k,3.5/n,5/n+d}
        {\draw[dash pattern=on 5mm off 2mm,black!20!white] (-1,-\x) 
            node[left,black] {$\n$} -- +(8,0);}
      \foreach \x/\n in {0/0,4.5/m,6.5/m+p}
        {\draw[dash pattern=on 5mm off 2mm,black!20!white] (\x,.5) 
        node[above,black] {$\n$} -- +(0,-5.5);}
      \draw[line width=4mm,black!30!white,opacity=.5,line cap=round] (4,-5) --
        (7,-5) node[above right,opacity=1,black] {$B$};
      \draw (0,0) node[fill=black,circle,inner sep=.7mm] (Q) {} 
        node[above left] {$q$};
      \draw (0,-2.5) node[fill=black,circle,inner sep=.7mm] (X') {} 
        node[above left] {$x'$};
      \draw (0,-3.5) node[fill=black,circle,inner sep=.7mm] (X) {} 
        node[above left] {$x$};
      \draw (4.5,0) node[fill=black,circle,inner sep=.7mm] (Y) {} 
        node[above right] {$y$};
      \draw (4.5,-3.5) node[fill=black,circle,inner sep=.7mm] (U) {}
        node[below left] {$u$};
      \draw (6.5,-2.5) node[fill=black,circle,inner sep=.7mm] (U') {}
        node[above right] {$u$};
      \draw (4.5,-5) node[fill=black,circle,inner sep=.7mm] (Z) {}
        node[below left] {$z$};
      \draw (6.5,-4) node[fill=black,circle,inner sep=.7mm] (Z') {}
        node[above right] {$z$};
      \draw (4.5,-2.5) node[fill=black,circle,inner sep=.7mm] (Umk) {} 
        node[above right] {$u_{m,k}$};
      \foreach \x in {0,.5,...,4}
        {\foreach \y in {0,.5,...,3}
          {
            \draw[-latex'] (\x,-\y) -- +(.45,0);
            \draw[-latex'] (\x,-\y) -- +(0,-.45);
          }
        }
      \foreach \x in {0,.5,...,4}
        {\draw[-latex'] (\x,-3.5) -- +(.45,0);}
      \foreach \y in {0,.5,...,3}
        {\draw[-latex'] (4.5,-\y) -- +(0,-.45);}
      \foreach \x in {4.5,5,5.5,6}
        {\foreach \y in {2.5,3,...,4.5}
          {
            \draw[-latex'] (\x,-\y) -- +(.45,0);
            \draw[-latex'] (\x,-\y) -- +(0,-.45);
          }
        }
      \foreach \x in {4.5,5,5.5,6}
        {\draw[-latex'] (\x,-5) -- +(.45,0);}
      \foreach \y in {2.5,3,...,4.5}
        {\draw[-latex'] (6.5,-\y) -- +(0,-.45);}
      \end{tikzpicture}
    \caption{Overview of the construction of the proof of Lemma~\ref{lemma-uniquexy}}
    \label{fig-u}
  \end{figure}
  Let~us recap the situation: we~have a state ($u_{m,k}$) from which there is
  a non-trivial vertical path $(u_{m,k+j})_{0\leq j\leq n-k}$ to~$u$,
  as~well as a non-trivial horizontal path to~$u$, which we~write
  $(u_{m+i,k})_{0\leq i\leq p}$ (see~Fig.~\ref{fig-u}). 
  Now, from Formula~\eqref{eq-HV}, from~$u$ there is a vertical path to
  a~$b$-state~$z$. 
  That path is unique, thanks to Formula~\eqref{eq-2suc}.
  Hence there is a unique integer~$d$ and a unique sequence of states
  $(u_{m,n+i})_{0\leq i\leq d}$ that forms a vertical path from~$u_{m,n}=u$ to
  a $b$-state~$u_{m,n+d}=z$. 
  Notice that $z$ is not labelled with~$s$, because it~has vertical
  predecessors that have horizontal successors.
  Now, starting from the horizontal path
  $(u_{m+i,k})_{0\leq i\leq p}$ and the vertical path $(u_{m,k+j})_{0\leq
    j\leq n-k+d}$ and applying Formula~\eqref{eq-square1}, we~build a grid 
  $(u_{m+i,k+j})_{0\leq i\leq p, 0\leq j\leq n-k+d}$. But since there is only
  one maximal vertical path from~$u$, it~must be the case that $u_{m+p,k+d}=z$:
  hence this state is not labelled with~$s$, but it has a vertical successor
  and belongs to~$B$, which is a contradiction.
\end{proof}

\medskip
We are now ready for proving our result:

\begin{proposition}\label{prop-grid}
Write~$\phi$ for the conjunction of all formulas above. Then $\calS,
q\models\phi$ if, and only if, the part of~$\calS$ that is reachable
from~$q$ is a (two-dimensional) grid (\ie, it~can be defined as the
product of two finite-state ``linear'' Kripke structures).
\end{proposition}

\begin{proof}
One the one hand, it~is clear that a grid can be labelled with~$h$,
$v$, $l$, $r$, $t$, $b$ and~$s$ in such a way that $\phi$~holds.

We~now prove the converse, assuming that all the states of~$\calS$ are
reachable from~$q$.  We assume that~$\calS$ is labelled with~$h$, $v$,
$l$, $r$, $t$, $b$ and~$s$ in a way that witnesses all the formulas
constituting~$\phi$.
 
Using~$h$ and~$v$, we~define the following relations: two states~$u$
and~$u'$ are \emph{$h$-equivalent} if there is a
sequence~$(u_i)_{0\leq i\leq k}$ of states such that $u_0=u$,
$u_k=u'$, the labelling of~$(u_i)_I$ is constant w.r.t.~$h$, and for
all~$0\leq i\leq k-1$, there is a transition $(u_i,u_{i+1})$ or
$(u_{i+1},u_i)$.  Notice that in particular $u$ and~$u'$ have the same
$h$-labelling. $v$-equivalence is defined similarly. Clearly enough,
these are equivalence relations, and we~write $\calH$ and~$\calV$ for
the sets of equivalence classes of these relations. Each set forms a
partition of the set of states of~$\calS$.

We~also define the sets~$L$, $R$, $T$ and~$B$ as the sets of states labelled
with the corresponding atomic propositions ($l$, $r$, $t$ and~$b$,
respectively). First notice that $L$ and~$R$ are sets of~$\calV$, and $T$
and~$B$ are in~$\calH$:

\begin{itemize}
\item $L\in \calV$: the initial state must be labelled with~$l$
  (and~$v$). We~prove that $L$~is the $v$-equivalence class of~$q$:
  first, $L$~contains precisely those states that are reachable
  from~$q$ via a vertical path. Since any vertical predecessor of a
  state in~$L$ is in~$L$ (Lemma~\ref{lemma-predlrtb}), we~get the result.

\item $R\in\calV$: the $s$-state is in~$R$, and we~prove that $R$~is
  the $v$-equivalence class of the $s$-state. For more clarity,
  we~assume that the $s$-state is labelled with~$v$. Then
  $R$~contains exactly the $v$-states from which only $v$-states are
  reachable. This proves the result.
\end{itemize}
The proof for~$T$ and~$B$ is similar. 

\medskip
Now, from Formula~\eqref{eq-hv} and Lemma~\ref{lemma-uniquexy}, any
$H\in\calH$ contains exactly one element from~$L$, which we~write
$l(H)$. Similarly, any $V\in\calV$ is the equivalence class of a
unique element of~$T$, denoted with $t(V)$.  Then any $H\in\calH$ and
$V\in\calV$ have non-empty intersection: this can be proven by
building a grid containing~$q$, $l(H)$ and~$t(V)$. If some~$H$ and~$V$
were to have two states in common, we~would be in a similar situation
as in the proof of Lemma~\ref{lemma-uniquexy}. As~a consequence, all
the elements of~$\calH$ contain the same number of states (namely, the number
of elements of~$\calV$). Similarly
for the elements of~$\calV$. Finally, since each state~$u$ in a
set~$H$ of~$\calH$ also belongs to~$\calV$, it~can be associated with
an element~$t(u)$ of~$T$. This gives rise to an order in each
set~$H$. One easily sees that these orders are compatible with the
``vertical successor'' relation, meaning that if~$u \prec_H u'$, then
their vertical successors~$x$ and~$x'$ satisfy $x \prec_{H'} x'$.
Hence the graph of~$\calS$ is isomorphic to the product of $L$ and~$T$
(augmented with a self-loop on their last states).
\end{proof}

Notice that we only characterised two-dimensional
grids. One-dimensional grids can be characterised by
\[
\exists r.\ \forall\gamma. \All\G(\Ex\X\true \et (\Ex\X\gamma \thn \All\X\gamma))
 \et (\Ex\F(r \et \gamma) \thn \All\G(r\thn \gamma))\et \Ex\F\All\G r .
\]
The first conjunct enforces that each state as exactly one successor;
the second conjunct expresses the fact at most one state is labelled
with~$r$.  The last conjunct enforces that the $r$-state is eventually
reached and never escaped, thus requiring that it~has a self-loop.

\end{document}